\documentclass[11pt]{scrarticle}
\usepackage[utf8]{luainputenc}
\usepackage[T1]{fontenc}
\usepackage[UKenglish]{babel}
\usepackage{amsmath, amsfonts, amsthm, amssymb, mathtools, stmaryrd}

\usepackage[margin=1in,letterpaper]{geometry}

\usepackage[pdfusetitle,hidelinks]{hyperref}
\usepackage[capitalise,noabbrev]{cleveref}
\hypersetup{
    hidelinks,
    colorlinks,
    linkcolor=[rgb]{0.3,0.3,0.6},
    citecolor=[rgb]{0.2,0.6,0.2},
    urlcolor=[rgb]{0.6,0.2,0.2}
}

\usepackage{csquotes}
\usepackage[osf,sc]{mathpazo}
\addtokomafont{disposition}{\normalfont\bfseries} \usepackage{relsize}
\usepackage{microtype}

\usepackage{thmtools, thm-restate}
\usepackage{subcaption}
\usepackage{enumitem}
\usepackage{comment}

\usepackage{todonotes}

\usepackage[
    backend=biber,
    style=alphabetic,
    url=false,
    doi=false,
    isbn=false,
    sorting=nyt,
backref=true,
    maxbibnames=10,
    maxcitenames=10,
    mincitenames=1,
]{biblatex}
\DefineBibliographyExtras{UKenglish}{}
\AtEveryBibitem{\clearlist{language}}
\AtEveryBibitem{\clearfield{issn}}

\usepackage{orcidlink}
\usepackage{booktabs}

\usepackage{comment}

\theoremstyle{definition}
\newtheorem{definition}{Definition}[section]

\newtheorem{example}[definition]{Example}
\newtheorem{remark}[definition]{Remark}

\theoremstyle{plain}
\newtheorem{lemma}[definition]{Lemma}
\newtheorem{corollary}[definition]{Corollary}
\newtheorem{theorem}[definition]{Theorem}

\newtheorem{observation}[definition]{Observation}
\newtheorem{fact}[definition]{Fact}
\Crefname{fact}{Fact}{Facts}

\theoremstyle{remark}
\newtheorem{claim}{Claim}[definition]

\Crefname{claim}{Claim}{Claims}
\newenvironment{claimproof}[1][Proof of Claim]{\begin{proof}[#1] }{ \end{proof}}
\setlist[enumerate, 1]{font=\upshape, itemsep=0.05em}
\setlist[enumerate, 2]{font=\upshape, itemsep=0.05em, nolistsep}
\setlist[itemize, 1]{noitemsep, nolistsep,font=\upshape}
\setlist[itemize, 2]{noitemsep, nolistsep,font=\upshape}

\newcommand{\Tt}{\mathcal{T}}

\DeclareMathOperator{\supDepth}{supDp}
\DeclareMathOperator{\clique}{clique}
\DeclareMathOperator{\child}{children}
\DeclareMathOperator{\td}{td}
\DeclareMathOperator{\vol}{vol}

\let\sup\relax
\DeclareMathOperator{\sup}{sup}

\newcommand{\bbN}{\mathbb{N}}
\newcommand{\bbQ}{\mathbb{K}}

\newcommand{\Pp}{\mathcal{P}}
\newcommand{\Oo}{O}

\usepackage{abstract}

\tikzset{
	vertex/.style={draw,circle,fill=gray},
	every node/.style={anchor=center},
	lbl/.style={color=lightgray}
}
\DefineBibliographyStrings{english}{backrefpage  = {cited on page},
  backrefpages = {cited on pages}
}

\newcommand{\doiorurl}{\iffieldundef{doi}
    {\iffieldundef{url}
       {}
       {\strfield{url}}}
    {http://dx.doi.org/\strfield{doi}}}

\newcommand{\myhref}[1]{\ifboolexpr{test {\ifhyperref}
    and
    not test {\iftoggle{bbx:url}}
    and
    not test {\iftoggle{bbx:doi}}
  }
    {\href{\doiorurl}{#1}}
    {#1}}

\DeclareFieldFormat{title}{\myhref{\mkbibemph{#1}}}
\DeclareFieldFormat
  [article,inbook,incollection,inproceedings,patent,thesis,unpublished]
  {title}{\myhref{\mkbibquote{#1\isdot}}}

\newif\ifcomments
\commentstrue    \ifcomments
  \renewcommand{\comment}[2]{\marginpar{\tiny{\textbf{#1: }\textit{#2}}}}
  \newcommand{\important}[1]{{\color{red}#1}}
  \newcommand{\note}[1]{{\color{blue}#1}}
\else
  \renewcommand{\comment}[2]{}
  \newcommand{\important}[1]{}
  \newcommand{\note}[1]{}
\fi

\declaretheoremstyle[bodyfont=\normalfont \itshape]{thmstyle} 
\declaretheorem[name=Theorem, style=thmstyle]{mainthm}

\Crefname{mainthm}{Theorem}{Theorems}

\newenvironment{proof-sketch}{\medskip\noindent{\em Proof Sketch.}}{\qed\bigskip}
\newenvironment{proof-attempt}{\medskip\noindent{\em Proof attempt.}}{\bigskip}

             \newcommand{\inbrace }[1]{\left\{#1\right\}}

\newcommand{\setdef}[2]{\inbrace{{#1}\ \middle| \ {#2}}}

\newcommand{\F}{\mathbb{F}}

\newcommand{\K}{\mathbb{K}}
\newcommand{\N}{\mathbb{N}}

  \newcommand{\poly}{\operatorname{poly}}

\newcommand{\size}{\operatorname{size}}

\newcommand{\depth}{\operatorname{depth}}

\newcommand{\maxorb}{\operatorname{maxOrb}}
\newcommand{\maxsup}{\operatorname{maxSup}}
\DeclareMathOperator{\tw}{tw}
\DeclareMathOperator{\pw}{pw}

\DeclareMathOperator{\emb}{emb}
\DeclareMathOperator{\perm}{perm}

\DeclareMathOperator{\cl}{cl}
\DeclareMathOperator{\id}{id}
\DeclareMathOperator{\Aut}{\mathbf{Aut}}
\DeclareMathOperator{\Sym}{\mathbf{Sym}}
\DeclareMathOperator{\Orb}{\mathbf{Orb}}
\DeclareMathOperator{\Stab}{\mathbf{Stab}}
\DeclareMathOperator{\StabP}{\mathbf{Stab}^{\bullet}}

\newcommand{\veca}{\boldsymbol{a}}
\newcommand{\vecb}{\boldsymbol{b}}

\newcommand{\vecv}{\boldsymbol{v}}
\newcommand{\vecw}{\boldsymbol{w}}

\newcommand{\calG}{\mathcal{G}}

\newcommand{\calP}{\mathcal{P}}

\newcommand{\calT}{\mathcal{T}}

\newcommand{\calX}{\mathcal{X}}

\usepackage{xspace}
\newcommand{\symVP}{\ensuremath{\mathsf{symVP}}\xspace}
\newcommand{\symVBP}{\ensuremath{\mathsf{symVBP}}\xspace}
\newcommand{\symVF}{\ensuremath{\mathsf{symVF}}\xspace}

\newcommand{\VP}{\ensuremath{\mathsf{VP}}\xspace}

\newcommand{\VNP}{\ensuremath{\mathsf{VNP}}\xspace}

\newcommand{\VBP}{\ensuremath{\mathsf{VBP}}\xspace}

\newcommand{\VF}{\ensuremath{\mathsf{VF}}\xspace}

\newcommand{\VW}{\ensuremath{\mathsf{VW}[1]}\xspace}
\newcommand{\VFPT}{\ensuremath{\mathsf{VFPT}}\xspace}
\newcommand{\FPT}{\ensuremath{\mathsf{FPT}}\xspace}

\renewcommand{\epsilon}{\varepsilon}

\renewcommand{\phi}{\varphi}

\renewcommand{\epsilon}{\varepsilon}
\newcommand{\ignore}[1]{}

 \usepackage{orcidlink}
\title{Lower Bounds in Algebraic Complexity via Symmetry and Homomorphism Polynomials}
\author{Prateek Dwivedi \orcidlink{0000-0002-0572-3721} \\ IT-Universitetet i København \and Benedikt Pago \orcidlink{0000-0001-6377-1230} \\University of Cambridge \and Tim Seppelt \orcidlink{0000-0002-6447-0568} \\ IT-Universitetet i København}
\newcommand{\Xx}{\mathcal{X}}
\DeclareMathOperator{\sgn}{sgn}
\DeclareMathOperator{\colhom}{colhom}

\newcommand{\bitvectors}{\genfrac{\langle}{\rangle}{0pt}{}}

\begin{document}
	\maketitle

    \begin{abstract}
        Valiant's conjecture from 1979 asserts that the circuit complexity classes \VP and \VNP are distinct, meaning that the permanent does not admit polynomial-size algebraic circuits.
        As it is the case in many branches of complexity theory, the unconditional separation of these complexity classes seems elusive.
        In stark contrast, the symmetric analogue of Valiant's conjecture has been proven by Dawar and Wilsenach (2020): the permanent does not admit symmetric algebraic circuits of polynomial size, while the determinant does.
        Symmetric algebraic circuits are both a powerful computational model and amenable to proving unconditional lower bounds.

        In this paper, we develop a symmetric algebraic complexity theory by introducing symmetric analogues of the complexity classes \VP, \VBP, and \VF called \symVP, \symVBP, and \symVF.
        They comprise polynomials that admit symmetric algebraic circuits, skew circuits, and formulas, respectively, of polynomial orbit size.
        Having defined these classes, we show unconditionally
        that
        \[ \symVF \subsetneq \symVBP \subsetneq \symVP. \]
        To that end, we characterise the polynomials in \symVF and \symVBP as those that can be written as linear combinations of homomorphism polynomials for patterns of bounded treedepth and pathwidth, respectively.
        This extends a previous characterisation by Dawar, Pago, and Seppelt (2025) of \symVP.
        The separation follows via model-theoretic techniques and the theory of homomorphism indistinguishability.

        Although \symVBP and \symVP admit strong lower bounds, we are able to show that these complexity classes are rather powerful:
        They contain homomorphism polynomials which are \VBP- and \VP-complete, respectively.
        Vastly generalising previous results,
        we give general graph-theoretic criteria for homomorphism polynomials and their linear combinations to be \VBP-, \VP-, or \VNP-complete.
        These conditional lower bounds drastically enlarge the realm of natural polynomials known to be complete for \VNP, \VP, or \VBP. 
        Under the assumption $\VFPT \neq \VW$,
        we precisely identify the homomorphism polynomials that lie in \VP as those whose patterns have bounded treewidth and thereby resolve an open problem posed by Saurabh (2016).
    \end{abstract}

\newpage
    \tableofcontents

   \newpage
    \section{Introduction}

    Algebraic complexity theory aims to classify families of polynomials according to the size of their smallest algebraic circuit representation. 
    An \emph{algebraic circuit} is a directed acyclic graph whose internal vertices represent the addition and multiplication gates and whose input gates are labelled with variables or field constants. The circuit computes a polynomial by propagating the computation from the input gates to the output gate through the internal gates. The size of the circuit is defined as the total number of gates and wires it contains.
    Two canonical examples of polynomials are the determinant $\det_n = \sum_{\pi \in \Sym_n} (-1)^{\sgn(\pi)} \prod_{i \in [n]} x_{i\pi(i)}$, which admits polynomial-size circuits, and the permanent $\perm_n = \sum_{\pi \in \Sym_n} \prod_{i \in [n]}  x_{i\pi(i)}$, for which only exponential-size circuits are known. 
    Valiant's conjecture \cite{valiant_completeness_1979}, the central open problem in the field, postulates that $(\perm_n)$ does not admit circuits of polynomial size -- or equivalently that $\VP \neq \VNP$.
    Here, $\VP$ denotes the class of polynomial families $(p_n)_{n \in \bbN}$ where each $p_n$ has polynomially bounded degree and can be computed by families of algebraic circuits of size polynomial in $n$. 
    The class $\VP$ is contained in $\VNP$, which consists of all polynomial families that can be reduced to the permanent.

    Besides \VP and \VNP, the complexity classes \VF and \VBP, which are structurally more restrictive and widely believed to be smaller, have been extensively studied.
    The class $\VF$ comprises all polynomials computed by polynomial-size \emph{formulas}, i.e.\ tree-shaped circuits, which are incapable to reuse previous subcomputations. 
    The class \VBP contains the polynomials computed by polynomial-size skew circuits. A circuit is \emph{skew} 
    if all but at most one child of every multiplication gate are input gates.
    This model is equivalent to algebraic branching programs \cite{mahajan_algebraic_2013}.

    The resulting complexity hierarchy is $\VF \subseteq \VBP \subseteq \VP \subseteq \VNP$.
   	While all inclusions are believed to be strict,
   	any unconditional separation appears to be currently out of reach.
	Facing this impasse,
	a substantial line of work 
	has focussed on proving lower bounds
	on restricted circuit models
	such as monotone or multilinear circuits and formulas \cite{agrawal_selection_2014,LST2021}. 
A shortcoming of these results -- such as the exponential lower bound for multilinear formulas computing the permanent~\cite{Raz04} -- is that they apply equally to the determinant, even though the latter belongs to $\mathsf{VBP}$.

    To the best of our knowledge, the only restricted circuit model for which the permanent is provably hard while the determinant is  still easy, are \emph{symmetric} circuits \cite{dawar_symmetric_2020}. 
    It is this symmetry restriction that we are interested in.
    We introduce  symmetric analogues of \VF, \VBP, and \VP and develop a systematic symmetric algebraic complexity theory.
    Most notably, we unconditionally separate the symmetric versions of \VF, \VBP and \VP.
    
    Curiously, our symmetric circuit model imposes a restriction only on one aspect of algebraic computation, while it relaxes it in others.
    Thus, the symmetric classes \symVF, \symVBP, \symVP that we define below are a priori not subclasses of, but orthogonal to \VF, \VBP, and \VP.

    The sense in which they are restricted is, as already said, symmetry on variables:
    Many polynomials studied in algebraic complexity are symmetric under certain permutations of the variables.
    Thus, it is natural to demand that the circuits computing these polynomials also possess these symmetries. 
    To be concrete, consider the permanent $\perm_n \in \mathbb{K}[\mathcal{X}_n]$ in variables $\mathcal{X}_n \coloneqq \{x_{i,j} \mid i,j \in [n]\}$. 
    The group $\Sym_n \times \Sym_n$ acts on $\mathbb{K}[\mathcal{X}_n]$ by mapping variables  $x_{i,j}$ to $x_{\pi(i),\sigma(j)}$ for $(\pi, \sigma) \in \Sym_n \times \Sym_n$.
    It is not hard to see that $\perm_n$ is invariant under this action.
    We call a polynomial\footnote{Throughout, we fix a field $\mathbb{K}$ of characteristic zero.} $p \in \mathbb{K}[\mathcal{X}_n]$ \emph{matrix-symmetric} if it is invariant under the action of $\Sym_n \times \Sym_n$ \cite{DW2022}.
    
    A \emph{matrix-symmetric circuit} for computing a matrix-symmetric polynomial is one where the action of $\Sym_n \times \Sym_n$ on the input variables can be extended to an automorphism of the circuit, that is, a permutation of the gates that preserves wires and non-wires.
    This model was first considered by \textcite{dawar_symmetric_2020}, 
    who showed that the permanent does not admit matrix-symmetric circuits of subexponential size, while the determinant can be computed with polynomial-size symmetric\footnote{The determinant is strictly speaking outside  the scope of this article because it is not matrix-symmetric but \emph{square-symmetric}, i.e.\ it is invariant under the simultaneous action of $\Sym_n$ on rows and columns given by $\pi(x_{i,j}) = x_{\pi(i),\pi(j)}$ for $\pi \in \Sym_n$. Nevertheless, our framework can be adapted to square symmetries, see \cite{DW2022}.} circuits. To be precise, they established a lower bound on the orbit size rather than total circuit size. 
    The \emph{orbit size} of a matrix-symmetric circuit is the size of the largest set of gates that can all be mapped to each other under the action of $\Sym_n \times \Sym_n$, and it is a lower bound for the total size. 
    So, even if symmetric circuits are permitted to be of arbitrary size, 
    they cannot compute the permanent unless they contain exponentially many gates whose subcircuits are mutually symmetric.
	This motivates our definition of the complexity classes \symVP, \symVBP, and \symVF. 
	For precise definitions of symmetric circuits and orbit size, see \cref{sec:preliminaries-first-12-pages}.
    
    \begin{definition}[Symmetric algebraic complexity classes]
    \label{def:symClasses}
    Let $p_n \in \mathbb{K}[\mathcal{X}_n]$ be a family of matrix-symmetric polynomials.
    \begin{enumerate}
    \item $(p_n) \in \symVP$ if the $p_n$ admit matrix-symmetric circuits of polynomial orbit size.
     \item $(p_n) \in \symVBP$ if the $p_n$ admit matrix-symmetric skew circuits of polynomial orbit size.
     \item $(p_n) \in \symVF$ if the $p_n$ admit matrix-symmetric formulas of polynomial orbit size.
     \end{enumerate}
    \end{definition}

These symmetric complexity classes are quite powerful, as we prove that $\symVBP$ and $\symVP$ contain $\VBP$- and $\VP$-complete polynomials, respectively. Nonetheless, they remain restricted enough to be meaningful: \textcite{dawar_symmetric_2020} proved that the permanent polynomial does not belong to $\symVP$.

	We stress that \symVP, \symVBP, and \symVF are in general incomparable with \VP, \VBP, and \VF for three reasons:
	our symmetric classes only contain matrix-symmetric polynomials,
	they do not impose a degree bound on the polynomials, and
	the complexity of symmetric circuits is measured in terms of orbit size rather than total size.
We justify these design choices by proving natural characterisations of our symmetric classes culminating in the following unconditional separation.
    
    \begin{mainthm}[restate=thmSeparation, label=thm:separation]
        $\symVF \subsetneq \symVBP \subsetneq \symVP$.
    \end{mainthm}

    To obtain \cref{thm:separation}, we precisely characterise \symVF and \symVBP in terms of so-called \emph{homomorphism polynomials}.
    This argument is based on the following fundamental observation:
    The matrix-symmetric polynomials $p \in \mathbb{K}[\mathcal{X}_n]$, such as the permanent,
    are precisely the finite linear combinations of homomorphism polynomials of bipartite multigraphs \cite[Lemma~3.1]{DPS2025}.\footnote{Similarly, the square-symmetric polynomials, such as the determinant, are precisely those that can written as linear combinations of suitably defined homomorphism polynomials from directed looped multigraphs. In order not to deal with loops and directed edges, we focus on matrix symmetries. } 
    For a bipartite multigraph $F$ with bipartition $V(F) = A \uplus B$, and a number $n \in \bbN$, the \emph{$n$-th homomorphism polynomial} of $F$~is 
    \begin{equation}
    \label{eq:homPoly}
        \hom_{F,n} \; \coloneqq\; \sum_{h \colon V(F) \to [n]} \; \prod_{\substack{ab \in E(F) \\ a \in A, b \in B}} \; x_{h(a),h(b)}\,.
    \end{equation}
    If the variables of $\hom_{F,n}$ are instantiated with the $\{0,1\}$-entries of the bi-adjacency matrix of some $(n,n)$-vertex bipartite graph~$G$, then $\hom_{F,n}$ evaluates to the number of homomorphisms from $F$ to $G$.
    This motivates calling $F$ the \emph{pattern graph} of the homomorphism polynomial.

    \textcite[Theorem~1.1]{DPS2025} showed that a family of matrix-symmetric polynomials $(p_n)$ is in \symVP if, and only if, $(p_n)$ can be written as linear combinations of homomorphism polynomials for pattern graphs of bounded treewidth. 
    By significantly refining their arguments,
    we characterise \symVF and \symVBP, which were not considered before. 

    \begin{mainthm}[restate=mainThmSymVF, label=thm:characterisationVF, name=Characterising $\symVF$]
    	A family of matrix-symmetric polynomials $(p_n)$ is in $\symVF$ if, and only if, $(p_{n})$ can be written as linear combinations of homomorphism polynomials for patterns of bounded treedepth.
    \end{mainthm}

    While the \emph{treewidth} $\tw(F)$ measures how tree-like a graph $F$ is, 
    the \emph{pathwidth} $\pw(F)$ measures how close $F$ is to being a path, see \cite{bodlaender_partial_1998,nesetril_sparsity_2012}.
    In turn, the \emph{treedepth} $\td(F)$ measures how star-like the graph $F$ is \cite{nesetril_tree-depth_2006}.
    
    \begin{mainthm}[restate=mainThmSymVBP, label=thm:sym-vbp-orbit, name=Characterising $\symVBP$]
    	A family of matrix-symmetric polynomials $(p_n)$ is in $\symVBP$ if, and only if, $(p_{n})$ can be written as linear combinations of homomorphism polynomials for patterns of bounded pathwidth.
    \end{mainthm}
    
    \cref{thm:separation} follows from
    \cref{thm:sym-vbp-orbit,thm:characterisationVF} by invoking involved results from finite model theory \cite{cai_optimal_1992}, graph structure theory \cite{seymour_graph_1993},
    and the theory of homomorphism indistinguishability \cite{roberson_oddomorphisms_2022}, 
    see \cite{seppelt_homomorphism_2024}.
    At first glance, the need for this technical depth is counter-intuitive since treewidth, pathwidth, and treedepth form a strict hierarchy of graph parameters -- a fact that may mislead one to conclude that \cref{thm:separation} is an immediate corollary of
    \cref{thm:sym-vbp-orbit,thm:characterisationVF}.
    However, 
    as we elaborate on in \cref{sec:extended-abstract-unconditional-separation},
    homomorphism polynomials of non-isomorphic pattern graphs are not necessarily linearly independent.
    Specifically, 
    there exist polynomials that can be written both as a linear combination of homomorphism polynomials for bounded-treewidth patterns and as one for unbounded-treewidth patterns.

    This is a blatant difference to  
    the line of work initiated by \textcite{curticapean_homomorphisms_2017} on the fixed-parameter tractability of so-called \emph{graph motif parameters}, i.e.\ linear combinations of homomorphism counts $\sum_F \alpha_F \hom(F, \star)$.
    Here,
    the same graph parameter $\sum_F \alpha_F \hom(F, \star)$
    is evaluated on graphs of arbitrary size and this premise is what renders homomorphism counts of non-isomorphic patterns linearly 
    independent \cite{lovasz_operations_1967}.
    The titular result of \cite{curticapean_homomorphisms_2017} asserts that the fixed-parameter tractability of a parameter $\sum_F \alpha_F \hom(F, \star)$ is dictated by the treewidth of the patterns $F$ whose unique coefficients $\alpha_F$ are non-zero.
    A substantial and beautiful theory \cite{roth_counting_2020,dorfler_counting_2021,roth_detecting_2021,doring_counting_2024,curticapean_counting_2025} has been developed based on this fundamental observation, which strikingly fails in the non-uniform world of algebraic complexity, see \cref{ex:not-linearly-independent}. 

    In the final part of this paper,
    we address this issue 
    by studying the non-symmetric algebraic complexity of homomorphism polynomials and thereby also leave the symmetry-framework from \cite{DPS2025} completely.
    For these polynomials, we prove novel and very general lower bounds, thereby answering an open question raised by \textcite{saurabh_analysis_2016}. Combined with our characterisations of the symmetric algebraic complexity classes, these lower bounds also shed light on the surprising power of the symmetric classes. They imply that $\symVP$ and $\symVBP$ contain a large variety of $\VP$-complete and $\VBP$-complete polynomials, respectively.

    In previous works \cite{durand_homomorphism_2016,mahajan_complete_2018,du_variants_2019,hrubes_hardness_2017,saurabh_analysis_2016},
    homomorphism polynomials have been identified as a source for natural complete polynomials for \VBP, \VP, and \VNP.
    However,
    these publications only studied homomorphism polynomials $(\hom_{F_n,n})$ for specific families of patterns $F_n$ such as $n \times n$ grids or $n$-leaf complete binary trees.
    In the context of monotone circuit complexity,
    homomorphism polynomials have also been studied \cite{KomarathPR23,bhargav_monotone_2025}, here however under the even more restrictive assumption that all $F_n$ are the same graph~$F$.
    Our \cref{thm:uncoloured-hom-complexity}
    gives general graph-theoretic criteria for a family of patterns $F_n$ to have homomorphism polynomials $(\hom_{F_n, n})$ that are \VBP-, \VP-, or \VNP-complete.
    Here, by a \emph{$p$-family} $(F_n)$ of multigraphs, we mean multigraphs satisfying $\lVert F_n \rVert \coloneqq |V(F_n)| + |E(F_n)| \leq q(n)$ for some polynomial $q$ and all $n \in \mathbb{N}$.\label{def:p-family-multigraphs}

    \begin{mainthm}[restate=thmUncolouredHom, label=thm:uncoloured-hom-complexity]
        Let $(F_n)_{n \in \mathbb{N}}$ be a $p$-family of bipartite multigraphs.
        Let $\epsilon > 0$.
        \begin{enumerate}
            \item If $\tw(F_n) \geq n^\epsilon$ for all $n \geq 2$, 
            then $(\hom_{F_n, n})$ is \VNP-complete.
            \item If $\tw(F_n) \in O(1)$ and $\pw(F_n) \geq \epsilon \log(n)$ for all $n \geq 2$,
            then $(\hom_{F_n, n})$ is \VP-complete.
            \item If $\pw(F_n) \in O(1)$ and  $\td(F_n) \geq \epsilon \log(n)$ for all $n \geq 2$,
            then $(\hom_{F_n, n})$ is \VBP-complete.
        \end{enumerate}
        In all cases, hardness holds under constant-depth $c$\nobreakdash-reductions over any field of characteristic zero.
    \end{mainthm}

    In particular, \cref{thm:uncoloured-hom-complexity,thm:sym-vbp-orbit} and \cite{DPS2025} imply that the symmetric classes \symVBP and \symVP respectively contain polynomials that are \VBP- and \VP-complete.
    Building on techniques from \cite{curticapean_complexity_2014,curticapean_count_2024},
    the proof of \cref{thm:uncoloured-hom-complexity} is by devising a chain of algebraic reductions from the complete polynomials of \cite{durand_homomorphism_2016,hrubes_hardness_2017} and employing very recent graph-structural results \cite{chekuri_polynomial_2016,groenland_approximating_2023,hatzel_tight_2024}.
    In \cref{thm:uncoloured-lincomb-hom-complexity}, we generalise \cref{thm:uncoloured-hom-complexity} to linear combinations of homomorphism polynomials.

    Under the assumption $\VP \neq \VNP$,
    \cref{thm:uncoloured-hom-complexity} does not ascertain for all pattern families $(F_n)$ whether $(\hom_{F_n, n}) \in \VP$.
    This is due to the fact that families such as $(\hom_{K_{\log(n), \log(n)},n})$ counting homomorphisms from the complete bipartite graphs $K_{\log(n), \log(n)}$ on $\Theta(\log n)$ vertices are unlikely to be \VNP-complete as argued in \cite[Proposition~3.6.2]{saurabh_analysis_2016}, see also \cite{papadimitriou_limited_1996} and \cite[Section~4.4]{flum_parameterized_2006}.
    In order to deal with such pattern families, a potentially stronger complexity-theoretic assumption is required:

    \begin{mainthm}[restate=thmUncolouredhomParametrised, label=thm:hom-parametrised-simplified]
        Let $(F_n)_{n \in \mathbb{N}}$ be a $p$-family of bipartite multigraphs such that $\tw(F_n)$ is non-decreasing.
        Unless $\VFPT = \VW$,
        it holds that $(\hom_{F_n, n}) \in \VP$
        if, and only if,
        $\tw(F_n) \in O(1)$.
    \end{mainthm}

    Here, the complexity classes \VFPT and \VW,
    as introduced by \textcite{blaser_parameterized_2019}, are algebraic analogues of the well-studied parametrised complexity classes $\mathsf{FPT}$ and $\mathsf{\#W}[1]$, see \cite{flum_parameterized_2006,cygan_parameterized_2015}.
    Assuming $\VFPT \neq \VW$, 
    \cref{thm:hom-parametrised-simplified}
    exhaustively classifies the homomorphism polynomials which are in \VP
    and thus resolves the open problem posed by \textcite[88]{saurabh_analysis_2016} to determine the algebraic complexity of homomorphism polynomials $(\hom_{F_n, n})$ of pattern families $(F_n)$ of treewidth $o(n)$.
    
    Furthermore,
    \cref{thm:hom-parametrised-simplified} equates the complexity-theoretic hypothesis $\VFPT \neq \VW$ 
    with the collapse of the ability of symmetric and non-symmetric algebraic circuits to compute homomorphism polynomials, i.e.\ $(\hom_{F_n, n}) \in \VP \iff (\hom_{F_n, n}) \in \symVP$ for all $p$-families $(F_n)$ of bipartite multigraphs, see \cref{cor:symmetric-vs-non-symmetric-computation}.
    In other words, the hypothesis $\VFPT \neq \VW$ and the symmetry assumptions on the circuits seem to be surprisingly related in terms of the lower bounds that can be derived from them.

    \section{Preliminaries}
    \label{sec:preliminaries-first-12-pages}
    As above, consider the variable set $\calX_{n}$ with the action of $\Sym_n \times \Sym_n$ previously described.
    For any $S \subseteq \calX_n$, let $\Stab(S) \coloneqq \setdef{\pi \in \Sym_n \times \Sym_n}{\pi(S) = S} \leq \Sym_n \times \Sym_n.$
    The \emph{pointwise stabiliser} is defined as
    $\StabP(S) \coloneqq \setdef{\pi \in \Sym_n \times \Sym_n}{\pi(x) = x \text{ for every } x \in S} \leq \Sym_n \times \Sym_n.$

    \paragraph*{Algebraic circuits and formulas.} 
    An \emph{algebraic circuit} $C$ over a set of variables $\calX$ and a field $\K$ is a directed acyclic graph (possibly with multiedges), where each vertex -- called a \emph{gate} -- is labelled by an element of $\calX \cup \K \cup \{+, \times\}$. We write $\lambda(g) \in \calX \cup \K \cup \{+, \times\}$ for the label of a gate. 
    The \emph{input gates}, labelled by elements of $\calX \cup \K$, have no incoming edges. Internal gates are labelled with $+$ or $\times$. The circuit has a unique \emph{output gate}, which has no outgoing edges. The circuit computes a polynomial in $\F[\calX]$ by propagating values from the input gates using addition and multiplication operations at the internal gates. Edges are directed according to the flow of computation, that is, from the output of a gate towards the next gate where the value is used. The \emph{size} of a circuit $C$, denoted $\|C\|$, is the total number of gates and wires (counted with multiplicity). We insist throughout that input gates are unique: For every $x \in \calX \cup \K$, there is a unique input gate $g \in V(C)$ with $\lambda(g) = x$.

    Whenever a permutation group $\Gamma$ acts on the variables $\calX$, we can define what it means for a circuit $C$ over $\calX$ to be $\Gamma$-symmetric. Let $\Aut(C) \leq \Sym(V(C))$ denote the group of permutations of gates that preserve non-edges, edges with multiplicities, and gate types. Formally, $\sigma \in \Sym(V(C))$ is in $\Aut(C)$ if and only if, for every internal gate $g$, $\lambda(\sigma(g)) = \lambda(g)$, and for every pair of gates $(g,h)$, the number and direction of edges between $(\sigma(g), \sigma(h))$ is the same as between $(g,h)$.  
    
    We say that $\pi \in \Gamma$ \emph{extends to an automorphism} of $C$ if there exists a $\sigma \in \Aut(C)$ such that $\lambda(\sigma(g)) = \pi(\lambda(g))$ for every input gate $g$. Here, $\pi(a) = a$ for every $a \in \K$.

    \begin{definition}[Symmetric circuits]
    \label{def:symmetric-ckts}
        Let $\Gamma$ be a group acting on $\calX$. An algebraic circuit $C$ over $\calX$ is \emph{$\Gamma$-symmetric} if every $\pi \in \Gamma$ on the input gates $\calX$ extends to an automorphism of $C$.
    \end{definition}

    In this paper, we mostly focus on restricted circuit models such as \emph{formulas} and \emph{skew circuits}.
    In general, a formula is a circuit without multiedges that is a tree. For symmetric formulas, we require treelikeness only on the internal gates because we always identify input gates that are labelled with the same variable or field element. 

    \begin{definition}[Symmetric formulas]
        A \emph{symmetric formula} is a symmetric circuit $C$ without multiedges such that the subgraph of $C$ induced by the internal gates is a tree.
    \end{definition}

    \begin{definition}[Symmetric skew circuits]
    \label{def:skew-ckts}
        An algebraic circuit is \emph{skew} if for every multiplication gate, at most one of its children is an internal (or a non-input) gate. A symmetric skew circuit is a symmetric circuit that is additionally skew.
    \end{definition}

    Note that \cref{def:symmetric-ckts} does not require that there is a \emph{unique} circuit automorphism $\sigma$ that $\pi \in \Gamma$ extends to. Often, however, uniqueness of this $\sigma$ is desirable in order to have a well-defined action of $\Gamma$ on the entire circuit.
    We call a $\Gamma$-symmetric circuit $C$ \emph{rigid} if the only circuit automorphism in $\Aut(C)$ that fixes every input gate is the identity. This is equivalent to saying that every $\pi \in \Gamma$ extends to a unique $\sigma \in \Aut(C)$. Thankfully, we may always assume that symmetric circuits are rigid, essentially by removing all redundancies in the circuit.
    The following \cref{lem:rigidification}
    generalises \cite[Lemma~7]{AD17} and \cite[Lemma~4.3]{DPS2025}
    by the final two assertions.

    \begin{lemma}[restate=rigidify, label=lem:rigidification, name=Rigidification]
    Let $\Gamma$ be a group acting on a variable set $\calX$, and let $\K$ be a field.
    Let $C$ be a $\Gamma$-symmetric circuit over $\calX \cup \K$.
    There exists a $\Gamma$-symmetric rigid circuit $C'$ with $\lVert C' \rVert \leq \lVert C \rVert$ that computes the same polynomial as $C$.
    If $C$ is a symmetric formula, then $C'$ is a symmetric formula with multiedges. 
    If $C$ is a symmetric skew circuit, then so is $C'$.
    \end{lemma}
    Note that rigidifying a symmetric formula with this lemma yields a formula with multiedges, which is not allowed by definition; however, these will only appear at intermediate steps inside proofs, so unless explicitly stated, formulas do not have multiedges.

    \paragraph*{Complexity measures and supports.} 

    Let $C$ be a $\Sym_n \times \Sym_n$-symmetric circuit over variables $\Xx_n$ and a field $\K$ as defined in \cref{def:symmetric-ckts}. It is standard in algebraic complexity to measure the complexity of $C$ by its size $\lVert C \rVert$. However, in the context of symmetric circuits, we also want to take the symmetry group into account. If $C$ is rigid, then for a gate $g$ in $C$ and $(\pi_1, \pi_2) \in \Sym_n \times \Sym_n$, we write $(\pi_1, \pi_2)(g)$ to denote the image of $g$ under the unique automorphism of $C$ that extends $(\pi_1, \pi_2)$. The \emph{orbit size} of $C$ is defined as 
    \[
       \maxorb(C)  \coloneqq \max_{g \in V(C)} \left| \setdef{(\pi_1, \pi_2)(g)}{(\pi_1, \pi_2) \in \Sym_n \times \Sym_n}\right|.
    \]
    Following previous works, our result relies on the notion of the \emph{support} of a gate, see \cite{AD17,dawar_notions_2025}.
    Since $C$ is $\Sym_n \times \Sym_n$-symmetric, the polynomial computed at the output gate is invariant under that group action. 
    However, this invariance need not hold for internal gates: The polynomials they compute are generally only invariant under a subgroup of $\Sym_n \times \Sym_n$, namely the stabiliser of the gate.
    The \emph{support set} of a gate $g$ captures precisely those input variables that remain fixed under such permutations. We denote by $\sup(g) \subseteq [n] \uplus [n]$ the \emph{minimal support} of the stabiliser group 
    $\Stab(g) \coloneqq \setdef{(\pi_1, \pi_2) \in \Sym_n \times \Sym_n}{(\pi_1,\pi_2)(g) = g}$, that is, $\sup(g)$ is the smallest set $S \subseteq [n] \uplus [n]$ for which $\StabP(S)$ is a subgroup of $\Stab(g)$. It is known \cite[Lemma 4.1]{DPS2025} that this smallest set $S$ is indeed unique, as long as there exists at least one $S \subseteq [n] \uplus [n]$ whose intersection with each copy of $[n]$ has size less than $n/2$ and which satisfies $\StabP(S) \leq \Stab(g)$.
    The \emph{maximum support size} of a $\Sym_n \times \Sym_n$-symmetric circuit $C$ is defined as 
    \[
        \maxsup(C) \coloneqq \max_{g \in V(C)} |\sup(g)|.
    \]
    We have the following relationships between orbit and support size. 
    For an intuitive (oversimplified) explanation of the lemma, one may imagine that the orbit of a gate $g$ essentially corresponds to the orbit of $\sup(g)$, which has size $\binom{n}{|\sup(g)|}$.
    \begin{lemma}[\protect{\cite[Lemmas 2.1 and 2.2]{DPS2025}}]
    \label{lem:constantSupportOfGates}
        For each $n \in \N$, let $C_{n}$ be a $\Sym_n \times \Sym_n$-symmetric rigid circuit.
        Then $\maxorb(C_{n})$ is polynomially bounded in $n$ if, and only if, there exists a constant $k \in \N$ such that $\maxsup(C_{n}) \leq k$ for all $n \in \N$.
    \end{lemma}

    \section{Main contributions}

    \subsection{Characterisations of \symVF and \symVBP by homomorphism polynomials}
    \label{sec:char-sym-contributions}

    We give an overview of the proofs of \cref{thm:characterisationVF} and \cref{thm:sym-vbp-orbit}, which characterise \symVF and \symVBP as the classes of polynomials that can be expressed as linear combinations of bounded treedepth and pathwidth homomorphism polynomials, respectively. Let us begin with \symVF.

    \mainThmSymVF*    

   The harder direction is to show that any symmetric formula with polynomial orbit size computes a linear combination of homomorphism polynomials of bounded treedepth. \textcite[Theorem~1.1]{DPS2025} proved that every matrix-symmetric circuit of polynomial orbit size computes a linear combination of homomorphism polynomials of bounded treewidth. 
   Their proof is by induction on the circuit structure and essentially shows that the circuit corresponds to a \enquote{simultaneous tree decomposition} of all pattern graphs in the linear combination.
This result already implies that the output of a symmetric \emph{formula} of polynomial orbit size must also lie in the class of bounded-treewidth homomorphism counts. It remains to prove that in addition to the treewidth, also the \emph{treedepth} of the pattern graphs is bounded by a constant.
   The crucial technical insight in \cite{DPS2025} is that the support size of the gates corresponds to the treewidth of the pattern graphs. And, by \cref{lem:constantSupportOfGates}, the support size must be constant if the orbit size of the circuit is polynomial.

   Here, we refine this argument by taking treedepth into account. The proof consists of two parts. Firstly, we introduce the measure of \emph{support depth} for symmetric circuits (\cref{def:supportDepth}). The support depth describes how often the support of a gate changes along any path from a leaf to the root. That is, there exists some element in the support of a child gate~$h$ that no longer exists in the support of its parent~$g$, i.e.\ $\sup(h) \setminus \sup(g) \neq \emptyset$. One may think of this as analogous to forget-nodes in tree decompositions \cite{bodlaender_partial_1998}. 
    
    We show in \cref{cor:supportDepthBounded} that a symmetric formula with orbit size $\Oo(n^d)$ has support depth at most $d$. This is the case because, whenever the support changes at some gate $g$, in the above sense, then the gate $g$ has at least $\Omega(n)$ many mutually symmetric children. In a formula, the subcircuits rooted at these children are all disjoint, so if there were more than $d$ such gates $g$ along any root-to-leaf path, then the orbit size of the formula would be at least $\Omega(n^{d+1})$. This shows that, for symmetric formulas, bounded orbit size implies bounded support depth.

    As a next step, we show that when support depth is taken into account, then the technical core of \cite{DPS2025} actually yields an upper bound not only on the treewidth but also on the treedepth of the pattern graphs whose homomorphism polynomials are computed by the symmetric formula. It turns out that the support depth translates into a corresponding notion of depth of the tree decomposition \cite{fluck_going_2024}, which is tied to the treedepth of the pattern graphs.  

   The easier direction of \cref{thm:characterisationVF} requires to show that whenever $F$ is a bipartite multigraph of treedepth $d$, then $\hom_{F,n}$ is expressible by a matrix-symmetric formula of orbit size at most $\Oo(n^d)$. The idea is to use the standard dynamic programming approach for inductively computing homomorphism counts along a depth-$d$ elimination tree of $F$. 
   An \emph{elimination tree} of $F$ is a tree $T$ with the same vertex set as $F$ such that any two vertices that are adjacent in $F$ are in ancestor-descendant relation in $T$. The proof of \cite[Theorem~11]{KomarathPR23} contains an explicit construction of a formula for $\hom_{F,n}$ from such an elimination tree, and we observe that this is symmetric.
    Any linear combination of homomorphism polynomials of bounded treedepth can then simply be written as a symmetric formula that sums up the individual homomorphism polynomials. This proves the easier direction of \cref{thm:characterisationVF}.

    A similar dynamic programming approach also applies in the case of bounded pathwidth: whenever $F$ is a bipartite multigraph of pathwidth $d$, then $\hom_{F, n}$ can be expressed as matrix-symmetric skew circuit of orbit size at most $\Oo(n^d)$. This proves the easier direction of our second main result.
 
    \mainThmSymVBP*

    As above, the starting point for the harder direction is to observe that, because the orbit size of the circuits is assumed to be polynomially bounded, the supports of the gates are of some constant size $k$, implying that the pattern graphs admit constant-width tree decompositions. In the special case of skew circuits, we can show that these tree decompositions can be turned into path decompositions: Intuitively, since every product gate has at most one internal (non-leaf) child and the other children are input gates, the tree decomposition produced in this fashion never needs to branch.

    \subsection{Unconditional separations of symmetric classes}
    \label{sec:extended-abstract-unconditional-separation}
    
    In this section, we show our main result.

    \thmSeparation*

    Given the characterisation of \symVBP and \symVF in \cref{thm:sym-vbp-orbit,thm:characterisationVF} as linear combinations of homomorphism polynomials of bounded treedepth and pathwidth, respectively,
    it may seem that \cref{thm:separation} is immediate.
    For example, the family $(P_n)_{n \in \mathbb{N}}$ of $n$-vertex paths has bounded pathwidth but unbounded treedepth.
    On the one hand, \cref{thm:characterisationVF} places $(\hom_{P_n, n})$ in \symVBP. 
    On the other hand,
    \cref{thm:sym-vbp-orbit} asserts that $(\hom_{P_n, n}) \in \symVF$ if, and only if, for every $n \in \mathbb{N}$, $\hom_{P_n, n}$ can be written as a linear combination of homomorphism polynomials of bounded treedepth. 
    Readers familiar with homomorphism counts, e.g.\ with \cite{curticapean_homomorphisms_2017},
may recall a result of \textcite{lovasz_operations_1967} asserting that the homomorphism counting functions $\hom(F_i, \star)$ are linearly independent for non-isomorphic graphs $F_i$.
    So doesn't this readily imply that $(\hom_{P_n, n}) \not\in \symVF$?

    It does not. The crux here is that linear independence only holds when the $\hom(F, \star)$ can be evaluated at arbitrarily large graphs.
    In the non-uniform world of algebraic complexity,
    homomorphism polynomials can only be evaluated at graphs with a certain number of vertices.
    Thereby, homomorphism polynomials for non-isomorphic patterns fail to be linearly independent in general.
    The following example shows that $(\prod_{v,w \in [n]} x_{v,w}) \in \symVF$ can be written as linear combinations of homomorphism polynomials of patterns of unbounded treewidth, seemingly placing $(\prod_{v,w \in [n]} x_{v,w})$ outside of \symVP.

    \begin{example}\label{ex:not-linearly-independent}
        For $n \in \mathbb{N}$,
        the polynomial $\prod_{v,w \in [n]} x_{v,w}$
        can be written as
        \begin{enumerate}
            \item a linear combination of homomorphism polynomials of patterns of bounded treedepth, by \cref{thm:characterisationVF} since $(\prod_{v,w \in [n]} x_{v,w}) \in \symVF$, and as
            \item a linear combination of homomorphism polynomials comprising all patterns that are quotients of the complete bipartite graph $K_{n,n}$.
            This linear combination is obtained by observing that $\prod_{v,w \in [n]} x_{v,w}$ is the subgraph polynomial for the pattern $K_{n,n}$ and applying M\"obius inversion \cite[Theorem~8.1]{DPS2025}.
            In particular, $\tw(K_{n,n}) \in \Theta(n)$.
        \end{enumerate}
    \end{example}    

    In order to deal with the lack of linear independence,
    we prove a semantic separation of the functions represented by polynomials in \symVF, \symVBP, and \symVP.
    A similar approach was taken by \textcite{dawar_symmetric_2020} when showing that $(\perm_n)\not\in \symVP$:
    They proved that all polynomials in \symVP have counting width $O(1)$ while the permanent has counting width $\Theta(n)$.
    A  family of matrix-symmetric polynomials $(p_n)$ has \emph{counting width} \cite{dawar_definability_2017} at most $k \in \mathbb{N}$ if, whenever two $n$-vertex graphs $G$ and $H$ satisfy the same sentences in the $k$-variable fragment $\mathsf{C}^k$ of first-order logic with counting quantifiers,
    then it holds that $p_n(G) = p_n(H)$.
    Here it is essential that the matrix symmetries allow the polynomial $p_n \in \mathbb{K}[\mathcal{X}_n]$ to be viewed as a parameter of $(n,n)$-vertex bipartite graphs. In particular, the value $p_n(G)$ does not depend on the ordering of the rows and columns of the bi-adjacency matrix of $G$.

    The logic $\mathsf{C}^k$ is well-studied not only in finite model theory \cite{grohe_descriptive_2017} but also e.g.\ in machine learning \cite{grohe_logic_2021} due to its connections to the Weisfeiler--Leman algorithm \cite{cai_optimal_1992,kiefer_weisfeiler-leman_2020} and Graph Neural Networks \cite{xu_how_2019,morris_weisfeiler_2019}.
    Crucially, it was shown in \cite{dvorak_recognizing_2010,dell_lovasz_2018}
    that two graphs $G$ and $H$ satisfy the same $\mathsf{C}^k$-sentences if, and only if,
    they are \emph{homomorphism indistinguishable} over all graphs of treewidth less than $k$, i.e.\ every graph $F$ with $\tw(F) < k$ admits the same number of homomorphisms to $G$ and to $H$.
    Homomorphism indistinguishability over various graph classes was shown to characterise diverse graph isomorphism relaxations such as isomorphism \cite{lovasz_operations_1967},
    quantum isomorphism \cite{mancinska_quantum_2020,kar_npa_2025},
    or equivalence w.\ r.\ t.\ various logics \cite{grohe_counting_2020,montacute_pebble-relation_2024,fluck_going_2024,schindling_homomorphism_2025}.
    Beyond that, the distinguishing power of graph isomorphism relaxations can be described using homomorphism indistinguishability \cite{roberson_oddomorphisms_2022}, see the monograph \cite{seppelt_homomorphism_2024}.
    We adopt this viewpoint when making the following definition
    for e.g.\ $w \in \{ \tw, \pw, \td\}$.

    \begin{definition}\label{def:w-counting-width}
        Let $w$ be an $\mathbb{N}$-valued graph parameter.
        A family of matrix-symmetric polynomials $(p_n)$ has \emph{$w$-counting width}  at most $k \in \mathbb{N}$ if,
        whenever two $(n,n)$-vertex bipartite graphs $G$ and $H$ are homomorphism indistinguishable over all graphs $F$ such that $w(F) < k$,
        then $p_n(G) = p_n(H)$.
    \end{definition}
    
    By \cite{dvorak_recognizing_2010,dell_lovasz_2018},
    $\tw$-counting width coincides with the original counting width from \cite{dawar_definability_2017}.
    Our \cref{thm:characterisationVF,thm:sym-vbp-orbit} and \cite[Theorem~1.1]{DPS2025} yield the following semantic properties of \symVP, \symVBP, and \symVF.

    \begin{corollary}\label{cor:bounded-w-counting-width}
        Let $(p_n)$ be a family of matrix-symmetric polynomials.
        \begin{enumerate}
            \item If $(p_n) \in \symVP$, then the $\tw$-counting width of $(p_n)$ is  $O(1)$.
            \item If $(p_n) \in \symVBP$, then the $\pw$-counting width of $(p_n)$ is  $O(1)$.
            \item If $(p_n) \in \symVF$, then the $\td$-counting width of $(p_n)$ is  $O(1)$.
        \end{enumerate}
    \end{corollary}

    Towards \cref{thm:separation},
    it remains to show that there are polynomials in \symVP and \symVBP of unbounded $\pw$- and $\td$-counting width, respectively.
    This is implied by the following theorem,
    whose first assertion is implied by \cite[Theorem~7.1]{DPS2025}.

    \begin{theorem}[restate=thmCWsinglehom, label=thm:symmetric-classes-single-hom]
        Let $(F_n)$ be a family of bipartite multigraphs.
        \begin{enumerate}
            \item $(\hom_{F_n, n}) \in \symVP$ if, and only if, $\tw(F_n) \in O(1)$.
            \item $(\hom_{F_n, n}) \in \symVBP$ if, and only if, $\pw(F_n) \in O(1)$.
            \item $(\hom_{F_n, n}) \in \symVF$ if, and only if, $\td(F_n) \in O(1)$.
        \end{enumerate}
    \end{theorem}

    In other words, \cref{thm:symmetric-classes-single-hom} 
    asserts that $\hom_{F_n, n}$ cannot be written as linear combinations of homomorphism polynomials for patterns $F'_n$ with $w(F'_n) \in O(1)$
    if $w(F_n)$ is unbounded.
    To prove the theorem, we show that the $w$-counting width of $(\hom_{F_n, n})$ is unbounded if $w(F_n)$ is unbounded and apply \cref{cor:bounded-w-counting-width}.

    Here, the key is to argue that, for $w \in \{\tw, \pw, \td\}$ and $k \in \mathbb{N}$, the class of all graphs $F$ such that $w(F) \leq k$ is \emph{non-uniformly homomorphism distinguishing closed} \cite{DPS2025}.
    A graph class $\mathcal{F}$ is \emph{(uniformly) homomorphism distinguishing closed} \cite{roberson_oddomorphisms_2022}
    if, for all graphs $F \not\in \mathcal{F}$,
    there exist graphs $G$ and $H$
    which are homomorphism indistinguishable over $\mathcal{F}$ but
    admit a different number of homomorphisms from $F$,
    see \cite[Chapter~6]{seppelt_homomorphism_2024}.
    It was shown in \cite{neuen_homomorphism-distinguishing_2024,seppelt_homomorphism_2024,fluck_going_2024} that the class of all graphs $F$ with $w(F) \leq k$ is homomorphism distinguishing closed, for $w$ and $k$ as above.
    These results are based on Cai-Fürer-Immerman graphs \cite{cai_optimal_1992,roberson_oddomorphisms_2022} and duality theorems for the width parameters~$w$ \cite{bienstock_quickly_1991,seymour_graph_1993,nesetril_sparsity_2012}.
    The analogous statement in the non-uniform setting 
    is that the $w$-counting width of the family $(\hom_{F_n, n})$ is in $O(1)$ iff $w(F_n) \in O(1)$.
    Hence, if e.g.\ $(\hom_{F_n, n}) \in \symVP$,
    then, by \cref{cor:bounded-w-counting-width}, 
     the $\tw$-counting width of $(\hom_{F_n, n})$ is in $O(1)$,
     which implies that $\tw(F_n) \in O(1)$.
    \cref{thm:separation} follows from \cref{thm:symmetric-classes-single-hom}
    by considering the family $P_n$ of $n$\nobreakdash-vertex paths and the family $B_n$ of $n$-leaf complete binary trees.
    It holds that $\pw(P_n) = 1$ and $\td(P_n) \in \Theta(\log n)$
    as well as $\tw(B_n) = 1$ and $\pw(B_n) \in \Theta(\log n)$.
    Hence, $(\hom_{P_n, n}) \in \symVBP \setminus \symVF$
    and $(\hom_{B_n, n}) \in \symVP \setminus \symVBP$.

    \Cref{thm:symmetric-classes-single-hom,ex:not-linearly-independent}
    hint at two extremes: 
    While generic matrix-symmetric polynomials may be written as both constant-width and unbounded-width linear combinations of homomorphism polynomials,
    single homomorphism polynomials do not exhibit this intricacy. 
    Generalising the latter case,
    we map out a large class of matrix-symmetric polynomials whose symmetric complexity can also be read off from its expansion in terms of homomorphism polynomials:
    For $n \in \mathbb{N}$,
    let $p_n = \sum_{F} \alpha_{F, n} \hom_{F, n}$ be a finite linear combination of homomorphism polynomials for pairwise non-isomorphic bipartite multigraphs $F$ without isolated vertices\footnote{This assumption does not constitute a loss of generality since coefficients can be rescaled to handle isolated vertices, i.e.\ $\hom_{F + K_1, n} = n \cdot \hom_{F, n}$. Throughout we consider only $\hom$-linear combinations without isolated vertices.}
    and coefficients $\alpha_{F, n} \in \mathbb{K}$.
    The \emph{volume} of such a linear combination is  $\vol(n) \coloneqq \max\{|V(F)| : \alpha_{F, n} \neq 0\}$.
    For a width parameter $w$ as above, write $\max w(p_n) \coloneqq \max\{w(F_n) : \alpha_{F, n} \neq 0 \}$.
    Under the assumption that the volume is sublinear,
    we prove the following \cref{thm:symmetric-lincomb-sublinear} by applying the descriptive complexity monotonicity theorem from \cite[Theorem~7.11]{DPS2025}.
    The first assertion is implied by \cite[Theorem~7.9]{DPS2025}.
    
    \begin{theorem}[restate=thmSymmetricLincomb, label=thm:symmetric-lincomb-sublinear]
        Let $(p_n)$ be a family of linear combinations of homomorphism polynomials for patterns without isolated vertices of volume~$o(n)$.
        \begin{enumerate}
            \item $(p_n) \in \symVP$ if, and only if, $\max \tw(p_n) \in O(1)$.
            \item $(p_n) \in \symVBP$ if, and only if, $\max \pw(p_n) \in O(1)$.
            \item $(p_n) \in \symVF$ if, and only if, $\max \td(p_n) \in O(1)$.
        \end{enumerate}
    \end{theorem}

    To reiterate: For example, \cref{thm:characterisationVF} generally shows that $(p_n) \in \symVF$ if, and only if, each $p_n$ can be written as some linear combination of bounded-treedepth homomorphism polynomials. 
    \cref{thm:symmetric-lincomb-sublinear} yields more specifically that for every concrete linear combination $(\sum_{F} \alpha_{F, n} \hom_{F, n})$ of sublinear volume, the treedepth of the patterns in precisely this linear combination dictates whether it is in \symVF. 
    If the volume is linear or superlinear, then this is not true and the linear combination may not be unique, as in \cref{ex:not-linearly-independent}.

    \subsection{Classical algebraic lower bounds for homomorphism polynomials}
    \label{sec:main-contributions-classical-complexity}
    Having identified homomorphism polynomials as the pivotal protagonists 
    of a symmetric algebraic complexity theory,
    we finally turn to the classical algebraic complexity of homomorphism polynomials.
    \textcite{durand_homomorphism_2016}
    identified specific families of patterns $(F_n)$ whose homomorphism polynomials are \VBP-, \VP-, or \VNP-complete.
    More precisely, they showed that the homomorphism polynomials\footnote{The homomorphism polynomials of \cite{durand_homomorphism_2016} are technically different from our \cref{eq:homPoly} but can be translated, see \cref{lem:colourful-path-vbp-complete,lem:colourful-binary-tree-vp-complete}.} of the family of paths, of complete binary trees, and of cliques are respectively \VBP-, \VP-, and \VNP-complete.
    We give a sweeping generalisation of these results by providing graph-theoretic criteria on the $(F_n)$ under which their homomorphism polynomials are \VBP-, \VP-, or \VNP-complete.

    \thmUncolouredHom*

    \Cref{thm:uncoloured-hom-complexity} is proven by reducing from the hard families given in \cite{durand_homomorphism_2016,hrubes_hardness_2017}.
    Our reduction makes use of colourful homomorphism polynomials and techniques developed by \textcite{curticapean_complexity_2014,curticapean_count_2024} for studying $\mathsf{\#W}[1]$- and $\mathsf{\#P}$\nobreakdash-hardness of homomorphism counts of unbounded-treewidth patterns.
    However, the nature of algebraic computation requires significant alterations to this framework:
    Most crucially, the patterns $F_n$ are counted only in $(n,n)$-vertex bipartite graphs.
    This means that gadget constructions need to be very succinct to fit into the size of the respective target graph.
    In particular, a literal translation of the techniques from \cite{curticapean_complexity_2014,curticapean_count_2024} would yield the assertion of \cref{thm:uncoloured-hom-complexity} only for pattern families with $|V(F_n)| \leq n^{1-\epsilon}$.
    To overcome these challenges, we make use of sophisticated graph-structural results \cite{chekuri_polynomial_2016,groenland_approximating_2023,hatzel_tight_2024}
    and recent insights from homomorphism indistinguishability \cite{roberson_oddomorphisms_2022,seppelt_logical_2024}.

    \Cref{thm:uncoloured-hom-complexity} falls short of establishing the complexity of families of homomorphism polynomials such as $(\hom_{K_{\log(n),\log(n)},n})$ since the treewidth $\Theta(\log(n))$ of the complete bipartite graph $K_{\log(n),\log(n)}$ grows too slowly.
    In \cite[88]{saurabh_analysis_2016}, 
    \citeauthor{saurabh_analysis_2016} asked to determine the complexity of homomorphism polynomials $(\hom_{F_n, n})$ for pattern families $(F_n)$ with $\tw(F_n) \in o(n)$.
    We resolve this question by proving the following theorem:
    \thmUncolouredhomParametrised*
    Here, \VFPT and \VW, as introduced by \textcite{blaser_parameterized_2019}, are the algebraic analogues 
    of the well-known parametrised complexity classes $\mathsf{FPT}$ and $\mathsf{\#W}[1]$, see \cite{cygan_parameterized_2015}.
    Both classes comprise parametrised families of polynomials, i.e.\ families of polynomials $(p_{n,k})$ indexed by both $n,k \in \mathbb{N}$.
    It holds that $(p_{n,k}) \in \VFPT$ if there exist algebraic circuits for $p_{n,k}$ of size at most $f(k) \cdot q(n)$ for an arbitrary function $f \colon \mathbb{N} \to \mathbb{N}$ and a polynomial $q$.
    In contrast, \VW contains polynomials that are believed not to be in \VFPT.
    
    Note that 
    we do not show that $(\hom_{F_n, n})$ is \VW-hard if $\tw(F_n) \in \Omega(1)$
    as $(\hom_{F_n, n})$ is not even a parametrised family, see \cite{curticapean_full_2021}.
    Instead, \cref{thm:hom-parametrised-simplified} is proven by showing that $(\hom_{F_n, n}) \in \VP$ for $\tw(F_n) \in \Omega(1)$ implies the existence of \VFPT-circuits for some \VW-hard parametrised family of polynomials. 
    Despite this indirect argument,
    it turns out that the hypothesis $\VFPT \neq \VW$ is necessary for proving \cref{thm:hom-parametrised-simplified}.
    Furthermore, the following \cref{cor:symmetric-vs-non-symmetric-computation} shows that $\VFPT \neq \VW$ 
    is equivalent to the collapse of symmetric and non-symmetric computation of homomorphism polynomials, i.e.\ 
    if $(\hom_{F_n, n}) \in \VP$, then $(\hom_{F_n, n}) \in \symVP$.
    
    \begin{corollary}[restate=corSymNonSym,label=cor:symmetric-vs-non-symmetric-computation]
        It holds that $\VFPT \neq \VW$ if, and only if,
        for every $p$-family $(F_n)$ of bipartite multigraphs such that $\tw(F_n)$ is non-decreasing,
        \[
            (\hom_{F_n, n}) \in \VP 
            \iff \tw(F_n) \in O(1) 
            \iff (\hom_{F_n, n}) \in \symVP.
        \]
    \end{corollary}

    Whereas the hypothesis  $\mathsf{FPT} \neq \mathsf{\#W}[1]$
    is fundamental to the study of the parametrised complexity of counting problems \cite{cygan_parameterized_2015},
    the hypothesis $\VFPT \neq \VW$ has received little attention \cite{curticapean_full_2021,bhattacharjee_exponential_2024} since it was introduced in \citeyear{blaser_parameterized_2019}.
    We believe that our
    \cref{thm:hom-parametrised-simplified}, 
    representing the algebraic analogue of the seminal characterisation of fixed-parameter tractable homomorphism counts by \textcite{dalmau_complexity_2004},
    will contribute to the significance of $\VFPT \neq \VW$.
    
    Finally, 
    we generalise \cref{thm:hom-parametrised-simplified,thm:uncoloured-hom-complexity}
    to linear combinations of homomorphism polynomials.
    As in \cref{thm:symmetric-lincomb-sublinear},
    this requires additional growth assumptions.
    In the following \cref{thm:uncoloured-lincomb-hom-complexity,thm:uncoloured-lincomb-hom-complexity-parametrised},
    a \emph{$p$-family} of polynomials $(p_n)$ is one whose degree is bounded by a polynomial in $n$. 
    
    \begin{theorem}[restate=thmLincombHomComplexity,label=thm:uncoloured-lincomb-hom-complexity]
        For $\epsilon > 0$,
        let $(p_n) \coloneqq (\sum_F \alpha_{F, n}\hom_{F, n})$ 
        be a $p$-family of linear combinations of homomorphism polynomials of bipartite multigraphs graphs without isolated vertices
        of volume at most $n^{1-\epsilon}$
        and polynomial \emph{dimension} $\dim(n) \coloneqq \lvert \{F \mid \alpha_{F, n} \neq 0 \}\rvert$.
        Then the following hold:
        \begin{enumerate}
            \item If $\max \tw(p_n) \geq n^\epsilon$ for all $n \geq 2$, then $(p_n)$ is
            \VNP-complete.
            \item If $\max \tw(p_n) \in O(1)$ and $\max \pw(p_n) \geq \epsilon \log(n)$ for all $n \geq 2$,
            then $(p_n)$ is \VP-complete,
            \item If $\max \pw(p_n) \in O(1)$ and $\max \td(p_n) \geq \epsilon \log(n)$ for all $n \geq 2$,
            then $(p_n)$ is \VBP-complete,
        \end{enumerate}
         In all cases, hardness holds under constant-depth $c$\nobreakdash-reductions over any field of characteristic zero. 
    \end{theorem}

    \begin{theorem}[restate=thmLincombHomComplexityParametrised,label=thm:uncoloured-lincomb-hom-complexity-parametrised]
        For $\epsilon > 0$, 
        let $(p_n)$ be a $p$-family of linear combinations of homomorphism polynomials of bipartite multigraphs without isolated vertices of volume at most $n^{1-\epsilon}$, polynomial dimension, and non-decreasing $\max \tw(p_n)$.
        Unless $\VFPT = \VW$,
        it holds that 
        \(
            (p_n) \in \VP\) if, and only if, \( \max \tw(p_n) \in O(1).
        \)
    \end{theorem}

    \section{Outlook}

    We have introduced the symmetric algebraic complexity classes \symVF, \symVBP, and \symVP as a new perspective on Valiant's programme, and established the unconditional separation $\symVF \subsetneq \symVBP \subsetneq \symVP$ via characterisations in terms of homomorphism polynomials. Motivated by this, we also proved new, and very general, (conditional) lower bounds for homomorphism polynomials in the classical non-symmetric circuit model, with surprising consequences for the relationship between $\VP$ and $\symVP$. See \cref{tab:overview} for a detailed overview. Further natural questions arising from our work are the following.
    
    \paragraph*{Symmetric depth reductions.} 
    In algebraic complexity theory, natural circuit restrictions such as bounded depth and formulas are extensively studied, in part due to a remarkable depth-reduction phenomenon not known to hold in the Boolean setting. Specifically, any degree-$d$ polynomial computable by a circuit of size $s$ can also be computed by a homogeneous depth-4 circuit or a (possibly non-homogeneous) depth-3 circuit of size $s^{O(\sqrt{d})}$. As a result, proving $n^{\omega(\sqrt{d})}$ lower bounds against such constant-depth circuits would be sufficient to separate $\VP$ from $\VNP$ and resolve Valiant’s conjecture \cite{Sap2021}.
    A natural question is if such a depth reduction result holds for symmetric circuits as well. 
    Our work indicates that this is unlikely in the general sense. The only kind of depth reduction that may be possible is one where the support depth (\cref{def:supportDepth}) is not decreased. Namely, \cref{lem:supportDepthImpliesTreedepth} implies that the support depth of a symmetric circuit controls the treedepth of the pattern graphs whose homomorphism polynomials it computes. Thus, \cref{thm:symmetric-lincomb-sublinear,thm:symmetric-classes-single-hom} represent obstacles to symmetric depth reduction.

    \paragraph*{Linear combinations of linear volume.}
    Our results show that a central parameter for understanding the complexity of linear combinations $(\sum_F \alpha_{F, n} \hom_{F, n})$ of homomorphism polynomials is their volume $\vol(n) \coloneqq \max \{ |V(F)| : \alpha_{F, n} \neq 0\}$.
    As summarised in \cref{tab:overview},
    linear combinations of sublinear volume behave essentially like single homomorphism polynomials and their symmetric and non-symmetric algebraic complexity is described by our results.

    In general, by \cite[Lemma~3.1]{DPS2025},
    every matrix-symmetric polynomial $p \in \mathbb{K}[\mathcal{X}_n]$ can be written as a linear combination of homomorphism polynomials of volume $\leq n$
    and such linear combinations are unique by \cref{lem:bipartite-multigraph-linear-independent}.
    However, as shown in \cref{ex:not-linearly-independent},
    it may be possible to represent a polynomial as a linear combination of homomorphism polynomials of smaller width at the expense of increasing the volume to $> n$.
    This suggests that the realm of linear combinations of at least linear volume is of a fundamentally different nature.
    Understanding their complexity is an intriguing open problem. Specifically, we ask:
    \begin{itemize}
     \item Under which assumptions is the symmetric or non-symmetric algebraic complexity of a linear combination of homomorphism polynomials of linear volume or superpolynomial dimension dictated by the width of the pattern graphs? More concretely, what is the complexity of subgraph polynomials or patterns of linear size such as the permanent, see \cite[Conjecture~8.7]{DPS2025}?
    \item Which complexity-theoretic hypothesis does $\VFPT \neq \VW$ need to be replaced by in order for \cref{cor:symmetric-vs-non-symmetric-computation} to hold for pathwidth/\VBP/\symVBP or treedepth/\VF/\symVF instead of treewidth/\VP/\symVP, see \cite{chen_fine_2015}?
    \item Are there linear combinations $(p_n)$ of homomorphism polynomials of linear volume and polynomial dimension such that $\max \tw(p_n) \in \Omega(1)$, but $(p_n) \in \VP$, even if $\VFPT \neq \VW$? If not, then $\symVP$ and $\VP$ would coincide on all matrix-symmetric polynomials, not just on single homomorphism polynomials as shown in \cref{cor:symmetric-vs-non-symmetric-computation}. 
    \end{itemize}

    \begin{table}\renewcommand{\arraystretch}{1.3}
        \centering
        \begin{tabular}{p{4.3cm}p{5.4cm}p{5.4cm}}\toprule
             & Single homomorphism polynomials $(\hom_{F_n, n})$ & Linear combinations of homomorphism polynomials $(p_n)$ \\ \midrule
             \textbf{Symmetric complexity} & generic \newline\small \cref{thm:symmetric-classes-single-hom} & volume $o(n)$ \newline \cref{thm:symmetric-lincomb-sublinear}  \\
             \textbf{Algebraic complexity} \newline assuming $\VFPT \neq \VW$ &
             $p$-family of patterns $(F_n)$ with $\tw(F_n)$ non-decreasing \newline\small\cref{thm:hom-parametrised-simplified} &
             $p$-family of polynomials~$(p_n)$ with volume at most~$ n^{1-\epsilon}$, polynomial dimension, and non-decreasing $\max \tw(p_n)$ \newline\small\cref{thm:uncoloured-lincomb-hom-complexity-parametrised} \\
             \textbf{Algebraic complexity} \newline\small assuming $\VP \neq \VNP$ & 
             $p$-family of patterns $(F_n)$ with $\tw(F_n) \geq n^\epsilon$ or $\tw(F_n) \in O(1)$ \newline\small\cref{thm:uncoloured-hom-complexity} &
             $p$-family of polynomials~$(p_n)$ with volume at most~$ n^{1-\epsilon}$, polynomial dimension, and $\max \tw(p_n) \geq n^\epsilon$ or $\max \tw(p_n) \in O(1)$ \newline\small\cref{thm:uncoloured-lincomb-hom-complexity} \\
             \bottomrule
        \end{tabular}
        \caption{\itshape Overview of the assumptions on matrix-symmetric polynomials represented as linear combinations of homomorphism polynomials subject to which we precisely characterise membership in \VP and \symVP. From left to right and from top to bottom, the assumptions become more restrictive.}
        \label{tab:overview}
    \end{table}

    \newpage
\section{Extended preliminaries}
    \label{sec:ext-prelims}

    If not stated otherwise, all polynomials and circuits are over some field $\mathbb{K}$ of characteristic zero.
    All logarithms are base-$2$.
    Write $\mathbb{N} = \{0,1,\dots \}$ and $[k] \coloneqq \{1, \dots, k\}$ for $k \in \mathbb{N}$ with $[0] = \emptyset$.

    \subsection{Algebraic complexity classes}

    We give an overview of the complexity classes introduced by \textcite{valiant_completeness_1979}, see also \cite{burgisser_completeness_2000,blaser_parameterized_2019}.
    Throughout, let $\mathbb{K}$ denote a field.
    Let $X = \{ x_1, x_2, \dots \}$ be a set of variables.
    A function $f \colon \mathbb{N} \to \mathbb{N}$ is \emph{$p$-bounded} if there exists a polynomial $q$ such that $f(n) \leq q(n)$ for all $n \in \mathbb{N}$.

    \begin{definition}
        A sequence of polynomials $(f_n) \in \mathbb{K}[X]$ is a \emph{$p$-family} if
        \begin{enumerate}
            \item the number of variables in $f_n$ is $p$-bounded, i.e.\ there exists a $p$-bounded function $q \colon \mathbb{N} \to \mathbb{N}$ such that $f_n \in \mathbb{K}[x_1, \dots, x_{q(n)}]$ for all $n \in \mathbb{N}$ and
            \item the degree $\deg(f_n)$ is $p$-bounded.
        \end{enumerate}
    \end{definition}
    
    For a polynomial $f$, write $\size(f)$ for the size of the smallest algebraic circuit computing~$f$.

    \begin{definition}
        The class \VP consists of all $p$-families $(f_n)$ such that $\size(f_n)$ is $p$-bounded.
    \end{definition}

    \begin{definition}
        The class $\VNP$ consists of all $p$-families $(f_n)$ such that there exists a $p$-family $(g_n)$ in $\VP$ and a $p$-bounded function $q$ for which
        \[
            f_n(x_1, \dots, x_{q(n)}) \;=\; \sum_{e \in \{0,1\}^{q(n)}} g_n(x_1, \dots, x_{q(n)}, e_1, \dots, e_{q(n)}).
        \]
        That is, $\VNP$ consists of polynomial families that can be expressed as exponential sums over polynomially bounded $\VP$-computable functions.
    \end{definition}

    \begin{definition}
        The class $\VF$ consists of all $p$-families $(f_n)$ that can be computed by algebraic formulas of $p$-bounded size.
        A formula is a circuit where the underlying graph is a tree (i.e.\ reuse of subcomputations is not allowed).
    \end{definition}

    \begin{definition}            
        The class $\VBP$ consists of all $p$-families $(f_n)$ that can be computed by skew algebraic circuits of $p$-bounded size. 
        An algebraic circuit is \emph{skew} if, for every multiplication gate, at most one child is an internal (non-input) gate.
    \end{definition}

    Skew circuits are computationally equivalent to \emph{algebraic branching programs} (ABPs), a well studied computational model in algebraic complexity theory. 
    An ABP is a layered directed acyclic graph with designated source and sink nodes, where each edge connects consecutive layers and is labelled by an affine linear form over the input variables. 
    The polynomial computed by the ABP is given by summing over all the path and taking the product of edge labels along each path. 
    For a detailed account of this equivalence and related results, see the survey by Mahajan~\cite{mahajan_algebraic_2013}. In $\symVBP$, we choose to work with skew circuits rather than ABPs because circuits have a more natural notion of symmetry.

    We will show hardness for the aforementioned complexity classes under $p$-projections or constant-depth $c$\nobreakdash-reductions,
    the former being more restrictive than the latter.
    
    \begin{definition}
    \begin{enumerate}
        \item Let $f, g \in \mathbb{K}[X]$.
        The polynomial $f$ is a \emph{projection} of $g$, in symbols $f \leq_p g$, if there is a substitution $r \colon X \to X \cup \mathbb{K}$ such that $f = r(g)$.
        \item Let $(f_n)$ and $(g_n)$ be $p$-families. The family $(f_n)$ is a \emph{$p$-projection} of $(g_n)$, in symbols $(f_n) \leq (g_n)$ if there exists a $p$-bounded function $q \colon \bbN \to \bbN$ such that $f_n \leq_p g_{q(n)}$ for all $n \in \mathbb{N}$.
    \end{enumerate}
    \end{definition}

    In order to define constant-depth $c$\nobreakdash-reduction, consider algebraic circuits with oracle gates.
    For a polynomial $g \in \mathbb{K}[x_1, \dots, x_\ell]$,
    an oracle gate for $g$ has $\ell$ inputs $t_1, \dots, t_\ell$ and computes the value $g(t_1, \dots, t_\ell)$.
    For a polynomial $f \in \mathbb{K}[X]$, $d \in \mathbb{N}$ and a set $G \subseteq \mathbb{K}[X]$ of polynomials,
    write $\size_G(f)$ for the size of the smallest circuit of $f$ with oracle gates for $g \in G$
    and $\size^d_G(f) \in \mathbb{N} \cup \{\infty\}$ for the size of the smallest depth-$d$ circuit of $f$ with oracle gates for $g \in G$.
    If $G$ is a singleton $\{g\}$,
    write $g$ instead of $G$.

    \begin{definition}
        Let $(f_n)$ and $(g_n)$ be $p$-families.
        The family  $(f_n)$ \emph{constant-depth $c$\nobreakdash-reduces} to $(g_n)$
        if there exists a constant $d \in \mathbb{N}$ and 
        a $p$-bounded function $q \colon \mathbb{N} \to \mathbb{N}$
        such that $\size^d_{g_{q(n)}}(f_n)$ is $p$-bounded.
    \end{definition}

    \subsection{Parametrised algebraic complexity classes \VW and \VFPT}

    We recall the definitions of the parametrised Valiant's classes \VW and \VFPT, as introduced by \textcite{blaser_parameterized_2019}.
    They are analogues of the parametrised complexity classes $\mathsf{W}[1]$ and $\mathsf{FPT}$, see \cite{cygan_parameterized_2015}.

    \begin{definition}\label{def:parametrised-p-family}
        A \emph{parametrised $p$-family} is a family $(f_{n,k})$ of polynomials such that
        \begin{enumerate}
            \item there exists a $p$-bounded function $q \colon \mathbb{N} \to \mathbb{N}$ such that $f_{n,k} \in \mathbb{K}[x_1, \dots, x_{q(n)}]$ for all $n,k \in \mathbb{N}$,
            \item there exists a $p$-bounded function $q \colon \mathbb{N} \to \mathbb{N}$ such that $\deg(f_{n,k}) \leq q(n+k)$ for all $n,k \in \mathbb{N}$.
        \end{enumerate}
    \end{definition}

    A function $p \colon \mathbb{N} \times \mathbb{N} \to \mathbb{N}$ is \emph{fpt-bounded} if there exists a $p$-bounded function $q \colon \mathbb{N} \to \mathbb{N}$ and an arbitrary\footnote{In the definition of the non-algebraic class \FPT, this function $g$ is assumed to be computable, see \cite{cygan_parameterized_2015}. In the context of Valiant's non-uniform complexity classes, this assumption is dropped.} function $g \colon \mathbb{N} \to \mathbb{N}$
    such that $p(n,k) \leq g(k) \cdot q(n)$ for all $n,k \in \mathbb{N}$.

    \begin{definition}
        The class \VFPT comprises all parametrised $p$-families $(f_{n,k})$
        for which $\size(f_{n,k})$ is fpt-bounded.
    \end{definition}

    \begin{definition}\label{def:fpt-reductions}
        \begin{enumerate}
            \item A parametrised $p$-family $(f_{n,k})$ \emph{fpt-projects} to a parametrised $p$-family $(g_{n,k})$,
        in symbols $(f_{n,k}) \leq^{fpt}_p (g_{n,k})$,
        if there are functions $r,s,t \colon \mathbb{N} \to \mathbb{N}$
        such that $r$ is $p$-bounded and $f_{n,k}$ is a projection of $g_{r(n)s(k), k'}$ for some $k' \leq t(k)$ and all $n,k \in \mathbb{N}$.
            \item A parametrised $p$-family $(f_{n,k})$ with $f_{n,k} \in \mathbb{K}[x_1, \dots, x_{p(n)}]$ is an \emph{fpt-substitution} of another parametrised $p$-family $(g_{n,k})$ with $g_{n,k} \in \mathbb{K}[x_1, \dots, x_{q(n)}]$,
            in symbols $(f_{n,k}) \leq^{fpt}_s (g_{n,k})$,
            if there are functions $r,s,t \colon \mathbb{N} \to \mathbb{N}$
            such that $r$ is $p$-bounded and there exist polynomials $h_1, \dots, h_{q(r(n)s(k))} \in \mathbb{K}[x_1, \dots, x_{p(n)}]$
            such that
            \[
                f_{n,k} = g_{r(n)s(k),k'}(h_1,\dots, h_{q(r(n)s(k))})
            \]
            for some $k' \leq t(k)$ and $\deg(h_i)$ and $\size(h_i)$ are fpt-bounded as functions of $n,k$, for all $1 \leq i \leq q(r(n)s(k))$.
            \item A parametrised $p$-family $(f_{n,k})$ \emph{fpt-$c$\nobreakdash-reduces} to a parametrised $p$-family $(g_{n,k})$,
        in symbols $(f_{n,k}) \leq^{fpt}_c (g_{n,k})$,
        if there is a $p$-bounded function $q \colon \mathbb{N} \to \mathbb{N}$ and functions $s,t \colon \mathbb{N} \to \mathbb{N}$ such that
        $\size_{G_{q(n)s(k), t(k)}}(f_{n,k})$ is fpt-bounded 
        where $G_{n,k} \coloneqq \{ g_{i,j} \mid i \in [n], j \in [k] \}$.
        \end{enumerate}
    \end{definition}

    The \emph{weft} of an algebraic circuit $C$ is the maximum number of gates of fan-in greater than $2$ on any path from a leaf to the root.
    For $s, k \in \mathbb{N}$,
    write $\bitvectors{s}{k} \subseteq \{0,1\}^s$ for the set of all length-$s$ vectors containing exactly $k$ ones.

    \begin{definition}\label{def:vw-t}
        Let $t \geq 0$.
        The class $\mathsf{VW}[t]$ comprises the parametrised $p$-families $(f_{n,k})$
        for which there exists a $p$-family $(g_n)$ of polynomials $g_n \in \mathbb{K}[x_1, \dots, x_{q(n)}, y_1, \dots, y_{q(n)}]$ for some $p$-bounded function $q \colon \mathbb{N} \to \mathbb{N}$
        such that $g_n$ is computed by a constant-depth unbounded-fan-in circuit of weft $\leq t$ and polynomial size, and
        \[
            (f_{n,k}) \leq^{fpt}_s \left( \sum_{\boldsymbol{e} \in \bitvectors{q(n)}{k}} g_n(x_1, \dots, x_{q(n)}, e_1, \dots, e_{q(n)}) \right).
        \]
    \end{definition}
    In \cite{blaser_parameterized_2019}, it was observed that $\VFPT = \mathsf{VW}[0]$.
    We will mostly be  concerned with \VW.

    \subsection{Matrix-symmetric polynomials and circuits}

   In the extended abstract of the paper, we considered for simplicity only $\Sym_n \times \Sym_n$-symmetric polynomials over the variables $\Xx_n$. In general, most our results on symmetric computation hold also for $\Sym_n \times \Sym_m$-symmetric polynomials over the variables $\Xx_{n,m} = \{ x_{i,j} \mid i \in [n], j \in [m] \}$, for $n \neq m$. From now on, we adopt this more general set-up.
   We now present the omitted proof of the rigidification lemma from \cref{sec:preliminaries-first-12-pages}. 

    \rigidify*
\begin{proof}
We construct $C'$ by merging equivalent gates in $C$. It may be necessary to repeat the following procedure more than once to accomplish rigidity. 
Formally, we define $C'$ and a surjective map $\delta \colon V(C) \to V(C')$ inductively from the input gates of $C$ to the root. The map $\delta$ keeps track of which gates of $C$ have been merged into which gates of $C'$ and just simplifies the presentation. Let $\Stab_{\Aut(C)}(\Xx)$ denote the subgroup of $\Aut(C)$ that fixes every input gate of $C$.

For every variable and field element, there is a unique input gate with that label in $C$, so input gates never violate rigidity.
Thus, we let the input gates of $C'$ be the same as in $C$, and define $\delta$ as the identity map on them. 
Now assume by induction that we have constructed $C'$ and $\delta$ up to layer $d$. We describe the construction on layer $d+1$. Let $V_{d+1} \subseteq V(C)$ be the set of gates in layer $d+1$. For each $\Stab_{\Aut(C)}(\Xx)$-orbit $O \subseteq V_{d+1}$, we introduce a new gate $g_O$ in layer $d+1$ of $C'$, and we let $\delta(g) \coloneqq g_O$ for every $g \in O$. The operation type of $g_O$ is the same as that of each gate in $O$.

To define the connections between layer $d+1$ and layer $d$ in $C'$,   
we first note:
\begin{claim}
Let $g,g' \in V(C)$ be in the same $\Stab_{\Aut(C)}(\Xx)$-orbit. Then there is a bijection $\gamma \colon gE(C) \to g'E(C)$ such that for each $h \in gE(C)$, $h$ and $\gamma(h)$ are in the same $\Stab_{\Aut(C)}(\Xx)$-orbit.
\end{claim}
\begin{claimproof}
Since $g,g'$ are in the same orbit, there exists $\pi \in \Stab_{\Aut(C)}(\Xx)$ such that $\pi(g) = g'$, and the action of $\pi$ on $gE(C)$ defines a bijection $\gamma \colon gE(C) \to g'E(C)$ with the claimed property.
\end{claimproof}
By the claim, for any two gates $g,g' \in O$, it holds that $\{ \delta(h) \mid h \in gE(C) \} = \{ \delta(h) \mid h \in g'E(C)  \}$. Therefore we can pick an arbitrary $g \in O$ and define the set of children of $g_O$ in $C'$ as 
\[
g_OE(C') \coloneqq \{ \delta(h) \mid h \in gE(C) \}.
\]
For each child $\delta(h)$ of $g_O$ in $C'$, we let the multiplicity of the edge
between $g_O$ and $\delta(h)$ be defined as follows. Let $m(g,h)$ denote the multiplicity of the edge between $g$ and $h$ in $C$.
Then in $C'$, the multiplicity of the edge $(g_O, \delta(h))$ is
\[
\sum_{h' \in \delta^{-1}(\delta(h)) \cap gE} m(g,h').
\]
This finishes the construction of $C'$. Note that by construction, two gates $g_1,g_2 \in V(C)$ are in the same $\Stab_{\Aut(C)}(\Xx)$-orbit if and only if $\delta(g_1) = \delta(g_2)$.

\begin{claim}
$C'$ computes the same polynomial as $C$.
\end{claim}
\begin{claimproof}
We show by induction that for every gate $g \in V(C)$, $\delta(g)$ computes the same polynomial as $g$. For the input gates this is clear. Now consider the inductive step for layer $d+1$. Let $g \in V(C)$ be a gate on layer $d+1$, w.l.o.g.\ assume it is a $+$-gate, and let $h_1, \dots, h_m$ be its children in $C$. Then it computes 
\[
g = \sum_{i \in [m]} m(g,h_i) \cdot h_i.
\] 
Now the children of $\delta(g)$ in $C'$ are defined as the $\delta$-images of the children of some gate $g^*$ in the same orbit of $g$ that was used in the construction. Hence, there exists a $\pi^* \in \Aut(C)$ that fixes every input gate, maps $g$ to $g^*$ and the children of $g$ to the children of $g^*$ (preserving edge multiplicities). It is generally true for any $h \in V(C)$ that $\pi(h)$ and $h$ compute the same polynomial whenever $\pi \in \Aut(C)$ is such that it fixes all input gates. So we can write 
\[
g = \sum_{i \in [m]} m(g,h_i) \cdot \pi^*(h_i) = \sum_{i \in [m]} m(g,h_i) \cdot \pi^*(h_i) = \sum_{i \in [m]} m(g,h_i) \cdot \delta(\pi^*(h_i)),
\]
where the last equality holds by induction hypothesis. In the construction of $C'$, the edge multiplicities between $\delta(g^*) = \delta(g)$ and its children $\{\delta(\pi^*(h_i)) \mid i \in [m]\}$ are chosen such that $\delta(g)$ indeed computes the above sum. The argument works exactly the same if $g$ is a multiplication gate. This finishes the inductive step.
\end{claimproof}

\begin{claim}
Every $\pi \in \Gamma$ extends to a circuit automorphism $\sigma \in \Aut(C')$.
\end{claim}
\begin{claimproof}
Let $\pi \in \Gamma$. Since $C$ is $\Gamma$-symmetric, there is a $\sigma' \in \Aut(C)$ that $\pi$ extends to. The $\sigma'$-image of every $\Stab_{\Aut(C)}(\Xx)$-orbit $O \subseteq V(C)$ is again a $\Stab_{\Aut(C)}(\Xx)$-orbit. Thus, we can define a bijection $\sigma \colon V(C') \to V(C')$ by setting $\sigma(g_O) \coloneqq g_{\sigma'(O)}$ for every $\Stab_{\Aut(C)}(\Xx)$-orbit $O \subseteq V(C)$. By construction of $C'$ and because $\sigma' \in \Aut(C)$, $\sigma$ is an automorphism of $C'$. 
\end{claimproof}
The claim shows that $C'$ is $\Gamma$-symmetric (but possibly not rigid).

\begin{claim}
If $C$ is a symmetric formula, then $C'$ is a symmetric formula, possibly with multiedges.
\end{claim}
\begin{claimproof}
Consider the inductive step for layer $d+1$ in the construction of $C'$. We may assume by induction that for every non-input gate $g \in V(C)$ in layers $\leq d-1$, the gate $\delta(g)$ has a unique parent in $C'$. Let $h \in V(C)$ be a gate on layer $d$. Let
\[
P_h \coloneqq \{ g \in V(C) \mid (g,h') \in E(C) \text{ for some } h' \in \delta^{-1}(\delta(h))  \}
\]
The gates in $P_h$ are the potential parents of $\delta(h)$ in $C'$. But, in the construction step of layer $d+1$, the $\delta$-image of all gates in $P_h$ will be the same. Indeed, by construction of $C'$, all gates in $\delta^{-1}(\delta(h))$ are in the same $\Stab_{\Aut(C)}(\Xx)$-orbit in $C$. Since $C$ is a formula, each of these gates has a unique parent, so all gates in $P_h$ are also in the same $\Stab_{\Aut(C)}(\Xx)$-orbit, say, $O$. Then $\delta(g) = g_O$ for each $g \in P_h$. So, $g_O$ is the unique parent of $\delta(h)$ in $C'$. This finishes the inductive step and shows that $C'$ is a tree on its internal gates. 
\end{claimproof}

\begin{claim}
If $C$ is a symmetric skew circuit, then so is $C'$.
\end{claim}
\begin{claimproof}
We have to show that every multiplication gate in $C'$ has at most one child that is a non-input gate of $C'$. Let $g_O \in V(C')$ be a multiplication gate and assume for a contradiction that it has two distinct children $h,h'$ in $C'$ that are not input gates. 
Then, because $C$ is skew, there must be two gates $g_1, g_2 \in O$ in $C$ such that their unique non-input children are not identified by $\delta$. In other words, let $h_1, h_2$ be the unique non-input children of $g_1, g_2$, respectively, in $C$. Then $\delta(h_1) \neq \delta(h_2)$. But this is impossible: Since $g_1, g_2$ are in the same 
$\Stab_{\Aut(C)}(\Xx)$-orbit, there is a $\pi \in \Stab_{\Aut(C)}(\Xx)$ that maps $g_1$ to $g_2$. This $\pi$ must also map $h_1$ to $h_2$, so $h_1, h_2$ are part of the same $\Stab_{\Aut(C)}(\Xx)$-orbit. Hence, by construction of $C'$ and $\delta$, we have $\delta(h_1) = \delta(h_2)$. This contradicts the assumption that $\delta(h_1) \neq \delta(h_2)$.
\end{claimproof}

In case that $C'$ is rigid, we are done, and by the above claims, $C'$ has the properties as stated in the lemma. If $C'$ is not rigid, we repeat the described construction, starting with $C'$. In this way, we eventually obtain a rigid circuit because the construction strictly reduces the number of gates as long as there is one non-singleton orbit. Hence it must terminate with a rigid circuit.
\end{proof}

   \subsection{Treewidth, pathwidth, and treedepth}

    Let $T$ be a graph and $F$ a multigraph.
    A \emph{$T$-decomposition} for $F$ is a map $\beta \colon V(T) \to 2^{V(F)}$ such that
    \begin{enumerate}
        \item $V(F) = \bigcup_{t \in V(T)} \beta(t)$,
        \item for all $uv \in E(F)$, there exists $t \in V(T)$ such that $u,v \in \beta(t)$, and
        \item for all $v \in V(F)$, the subgraph of $T$ induced by the $t \in V(T)$ with $v \in \beta(v)$ is connected.
    \end{enumerate}
    The \emph{width} of a $T$-decomposition $(T, \beta)$ is $\max_{t \in V(T)} |\beta(t)| - 1$.
    The \emph{treewidth} $\tw(F)$ of a multigraph $F$ is the minimum width of a $T$-decomposition for $F$ where $T$ is a tree.
    The \emph{pathwidth} $\pw(F)$ of a multigraph $F$ is the minimum width of a $P$-decomposition for $F$ where $P$ is a path.  

    The \emph{treedepth} $\td(F)$ of a multigraph $F$, as introduced by \textcite{nesetril_tree-depth_2006}, is defined as the minimum height of an elimination tree for $F$.
    A \emph{elimination tree} for $F$ is a tree with vertex set $V(F)$ such that, whenever $uv \in E(F)$, then $u,v$ are in an ancestor-descendent relation in~$T$.

    The following relationship holds between treewidth, pathwidth, and treedepth for every multigraph $F$, by \cite[Section~6.4]{nesetril_sparsity_2012},
    \begin{align}
        \tw(F) \leq \pw(F) \leq \td(F) - 1.\label{eq:tw-td-pw}
    \end{align}

    \section{Characterisations of symmetric complexity classes}
    \label{sec:characterisations}

    In this section, we give the complete proofs of~\cref{thm:characterisationVF} and \cref{thm:sym-vbp-orbit}, which characterise matrix-symmetric formulas and skew circuits of polynomially bounded orbit sizes.
    
We recall some definitions and notations from~\cite{DPS2025} that are essential for the proofs. 
    We begin by introducing the following set, consisting of finite linear combinations of homomorphism polynomials associated with bipartite multigraphs:
    \begin{equation}
    \label{eq:lrcomb-hompoly}
        \mathfrak{G}_{n,m} \;\coloneqq\;
        \setdef{\sum \alpha_i \hom_{F_i, n, m}}{\text{$F_i$ bipartite multigraphs, } \alpha_i \in \K}
        \;\subseteq\; \K[\calX_{n,m}].
    \end{equation}
    In~\cite[Lemma 3.1]{DPS2025}, it was shown that the set $\mathfrak{G}_{n,m}$ is equal to the set of matrix-symmetric polynomials.

    \paragraph*{Labelled bipartite graphs and homomorphism polynomial maps.} For any $\ell, r \in \N$, an \emph{$(\ell, r)$-labelled bipartite graph} is a tuple $\boldsymbol{F} = (F, \veca, \vecb)$, where $F$ is a bipartite multigraph with vertex partition $A \uplus B$, and $\veca \in A^{\ell}$ and $\vecb \in B^{r}$ are sequences of labels on the vertices on the left and right sides, respectively.
    Write $\calG(\ell, r)$ for the set of all $(\ell, r)$-labelled bipartite graphs.
    For $n, m \in \N$, we define the \emph{homomorphism polynomial map} $\boldsymbol{F}_{n,m}\colon [n]^{\ell} \times [m]^{r} \to \K[\calX_{n,m}]$ as
    \[
        \boldsymbol{F}_{n,m}(\vecv, \vecw) \; \coloneqq\; \sum_{\substack{h : A \uplus B \to [n] \uplus [m] \\ h(\veca) = \vecv \\ h(\vecb) = \vecw}} \prod_{(a,b) \in E(F)} x_{h(a), h(b)} \quad \in\quad \bbQ[\mathcal{X}_{n,m}].
    \]
    Note that when $\ell = r = 0$, then this reduces to $\hom_{F, n, m}$, see \cref{eq:homPoly}. We define
    \begin{equation*}
        \mathfrak{G}_{n,m}(\ell, r) \;\coloneqq\; \setdef{\sum \alpha_i \boldsymbol{F}^i_{n,m}}{\alpha_i \in \K, \boldsymbol{F}^i \in \calG(\ell,r)},
    \end{equation*}
    the set of finite linear combinations of homomorphism polynomial maps.

    \paragraph*{Labelled graphs of bounded treewidth.}
    We say $\boldsymbol{F} = (F, \veca, \vecb)$ has \emph{treewidth} less than $k$ if there exists a tree decomposition $(T, \beta)$ of $F$ of width less than $k$, such that there is a node $s \in V(T)$ whose bag $\beta(s)$ contains all the vertices appearing in $\veca$ and $\vecb$. Write $\calT^k(\ell, r) \subseteq \calG(\ell, r)$ for set of all $(\ell, r)$-labelled bipartite graphs of treewidth less than $k$. Analogous to~\cref{eq:lrcomb-hompoly}, we define the set of finite linear combination of homomorphism polynomial maps of $(\ell,r)$-labelled bipartite graphs of treewidth at most $k$:
    \begin{equation}
    \label{eq:lrcomb-hompoly-twk}
        \mathfrak{T}^k_{n,m}(\ell, r) \;\coloneqq\; \setdef{\sum \alpha_i \boldsymbol{F}^i_{n,m}}{\alpha_i \in \K, \boldsymbol{F}^i \in \calT^k(\ell,r)}  \subseteq \mathfrak{G}_{n,m}(\ell, r).
    \end{equation}
    For labels $\vecv \in [n]^\ell$ and $\vecw \in [m]^r$, $\mathfrak{T}^k_{n,m}(\vecv, \vecw) \coloneqq \setdef{\phi(\vecv, \vecw)}{ \phi \in \mathfrak{T}^k_{n,m}(\ell,r)} \subseteq \K[\calX_{n,m}]$. Using the notations, we are now ready to formally state the main result of \cite{DPS2025}.

\begin{theorem}[\protect{\cite[Theorem 1.1]{DPS2025}}]
    \label{thm:symVP-char}
        A family of matrix-symmetric polynomials $(p_{n,m})_{n,m \in \bbN}$ is in $\symVP$ if, and only if, there exists a $k \in \N$ such that $(p_{n,m}) \in \mathfrak{T}^k_{n,m}(0,0)$.
    \end{theorem}

    \paragraph*{Labelled graphs of bounded treedepth.}
    We now define analogues of the set in~\cref{eq:lrcomb-hompoly-twk} tailored to~\cref{thm:characterisationVF}. 
    In~\cite{fluck_going_2024}, the \emph{depth} of a rooted tree decomposition $T$ is defined as
    \[
        \depth(T) \;\coloneqq\; \max_{v \in V(T)} \left\lvert \bigcup_{t \preceq v} \beta(t) \right\rvert,
    \]
    where $t \preceq v$ if and only if $t$ is on the unique path from the root of $T$ to $v$.
    For $k, q \in \bbN$, let $\Tt^{k,q}$ denote the class of bipartite graphs that admit a rooted tree decomposition of width less than $k$ and depth at most $q$. Similarly, let $\Tt^{k,q}(\ell, r)$ denote the class of $(\ell, r)$-labelled bipartite graphs for which there exists a rooted tree decomposition of width less than $k$ and depth at most $q$, such that all label vertices appear together in the root bag.
    
    \begin{lemma}[\protect{\cite[Theorem 14 and Corollary 15]{fluck_going_2024}}]
    \label{lem:treedepth}
    Let $k,q \in \bbN$. Every graph in $\Tt^{k,q}$ has treedepth at most $q$ and treewidth at most $k-1$.
    \end{lemma}
    
    Define the corresponding set of linear combinations as follows:
    \begin{align}
    \label{eq:lrcomb-hompoly-tdk}
        \mathfrak{T}^{k,q}_{n,m}(\ell, r) &\coloneqq \setdef{ \sum \alpha_i\boldsymbol{F}^i_{n,m} }{ \alpha_i \in \bbQ, \boldsymbol{F}^i \in \Tt^{k,q}(\ell,r)} \subseteq \mathfrak{T}^k_{n,m}(\ell, r).
    \end{align}
As before, let $\mathfrak{T}^{k,q}_{n,m}(0,0) \eqqcolon \mathfrak{T}^{k,q}_{n,m}$ denote the set of linear combinations of homomorphism polynomials arising from bipartite graphs of treewidth less than $k$ and treedepth at most $q$. Then~\cref{thm:characterisationVF} proves that a family of polynomials $(p_{n,n})$ belongs to $\symVF$ if and only if there exists a constant $k \in \N$ such that $(p_{n,n}) \in \mathfrak{T}^{k,q}_{n,n}(0,0)$.

    \paragraph*{Labelled graphs of bounded pathwidth.}
    We say that $\boldsymbol{F} = (F, \veca, \vecb)$ has \emph{pathwidth} less than $k$ if there exists a path decomposition $(P, \beta)$ of $F$ such that there is a node $s \in V(T)$ of degree $\leq 1$ whose bag $\beta(s)$ contains all the vertices appearing in $\veca$ and $\vecb$.
    Write $\mathcal{P}^k(\ell, r)$ for the class of all $(\ell, r)$-labelled bipartite graphs of pathwidth less than $k$. Analogous to \cref{eq:lrcomb-hompoly-twk} we define linear combination of polynomial maps of $(\ell, r)$-labelled bipartite graphs of pathwidth at most $k$. Let $k, n, m \in \N$, and define
    \begin{equation}
    \label{eq:lrcomb-hompoly-pwk}
        \mathfrak{P}_{n,m}^k(\ell, r) \;\coloneqq\; \setdef{\sum \alpha_i \boldsymbol{F}_{n,m}^{i}}{ \alpha_i \in \K, \boldsymbol{F}^i \in \mathcal{P}^k(\ell, r)}  \subseteq \mathfrak{T}^k_{n,m}(\ell, r).
    \end{equation}
    Therefore, \cref{thm:sym-vbp-orbit} proves that, for all $n \in \N$, a family $(p_{n,n})$ belongs to $\symVF$ if and only if there exists a constant $k \in \N$ such that $(p_{n,n}) \in \mathfrak{P}^{k}_{n,n}(0,0)$.

    We recall the following example from \cite[Example 3.8]{DPS2025}, which illustrates that every constant is a polynomial in $\mathfrak{P}_{n,m}^k(0,0)$ for $k \geq 2$. Note that similar examples can be constructed for $\mathfrak{T}^{k,q}_{n,m}$.

    \begin{example}
    \label{ex:hompoly-constant}
        Let $\ell, r \in \N$ and consider the $(\ell, r)$-labelled edge-less bipartite graph $\boldsymbol{J} \coloneqq (J, \veca, \vecb)$. Since $J$ has no edges, it trivially belongs to the class $\mathcal{P}^k(\ell, r)$ for any $k \geq \ell + r$.
        For all $n, m \in \N$ and labelling $\vecv \in [n]^{\ell}$, $\vecw \in [m]^r$, we have:
        \[
            \boldsymbol{J}_{n,m}(\vecv, \vecw) \;=\; \sum_{\substack{h : A \uplus B \to [n] \uplus [m] \\ h(\veca) = \vecv \\ h(\vecb) = \vecw}} 1 \;=\; 1.
        \]
        That is, there is exactly one homomorphism from $J$ to $K_{n,m}$ respecting the labelling.
    \end{example}

    The following example illustrates that variables are polynomials in $\mathfrak{P}^k_{n,m}$.

    \begin{example}
    \label{ex:hompoly-var}
        Let $\ell,r \in \N$ and consider \(\boldsymbol{F}=(F,\veca,\vecb)\) be a \((\ell,r)\)-labelled bipartite graph consisting of a single edge.
        Since $F$ admits a path decomposition of width $\ell+r$, hence $\boldsymbol{F}\in \mathcal{P}^k(\ell,r)$ for $k = \ell +r$.
        For $n,m\in\N$, $ v \in[n]$ and $ w\in[m]$, let $\vecv= v \cdots v\in [n]^\ell$ and $\vecw= w \cdots w \in [m]^r$.
        Then it holds that $\boldsymbol{F}_{n,m}(\vecv, \vecw) = x_{v,w}$.
    \end{example}

    \subsection{Operations on homomorphism polynomial maps}

    For the proof of~\cref{thm:symVP-char}, closure properties such as addition and multiplication were proved in~\cite{DPS2025} for $\mathfrak{T}^k_{n,m}$. We require analogous properties to hold for $\mathfrak{T}^{k,q}_{n,m}$ and $\mathfrak{P}^k_{n,m}$. To that end, we revisit certain operations on labelled graphs and their associated homomorphism polynomial maps, as introduced in~\cite[Section 5.2]{DPS2025}. We begin by highlighting that the left and right labels of the labelled graphs can be interchanged. This simplifies the notation for the rest of this section where we only state the results with operation only on the left labels.

    \subsubsection{Disjoint union}

    For labelled graphs $\boldsymbol{F} = (F,\veca, \vecb) \in \calG(\ell, r)$ and $\boldsymbol{F'} = (F',\veca', \vecb') \in \calG(\ell', r')$, define their disjoint union $\boldsymbol{F} \otimes \boldsymbol{F'} \in \calG(\ell+\ell', r+r')$ as $(F \uplus F', \veca\veca', \vecb\vecb')$.
    Analogously, for homomorphism polynomial maps $\phi \in \mathfrak{G}(\ell, r)$ and $\phi' \in \mathfrak{G}(\ell', r')$, define $\phi \otimes \phi' \in \mathfrak{G}(\ell+\ell', r+r')$ by
    \begin{equation}\label{eq:tensor-disjoint-union}
        \left(\phi \otimes \phi'\right)(\vecv, \vecw) \;\coloneqq\;
        \phi(v_1, \dots, v_\ell, w_1, \dots, w_r)\cdot
        \phi'(v_{\ell+1}, \dots, v_{\ell+\ell'}, w_{r+1}, \dots, w_{r+r'}),
    \end{equation}
    for all $\vecv \in [n]^{\ell + \ell'}$ and $\vecw \in [m]^{r + r'}$.

    In~\cite[Lemma 5.10]{DPS2025}, it was observed that for a constant $k$, the set $\calT^k$ and the corresponding set $\mathfrak{T}^k$ are closed under disjoint unions. We note that the same closure property also holds for the bounded treedepth and bounded pathwidth classes.

    \begin{lemma}[Disjoint unions of labelled graphs]\label{lem:disjoint-union}
        Let $k, k', q, q', \ell, \ell', r, r' \in \mathbb{N}$.
        \begin{enumerate}
            \item If $\boldsymbol{F} \in \calT^k(\ell, r)$ and $\boldsymbol{F'} \in \calT^{k'}(\ell', r')$, then $\boldsymbol{F} \otimes \boldsymbol{F'} \in \calT^{k''}(\ell+\ell', r+r')$ where \[ k'' \coloneqq \max\{k, k', \ell+\ell'+r+r'\}. \]
    
            \item If $\boldsymbol{F} \in \mathcal{T}^{k,q}(\ell, r)$ and $\boldsymbol{F'} \in \mathcal{T}^{k', q'}(\ell', r')$, then $\boldsymbol{F} \otimes \boldsymbol{F'} \in \mathcal{T}^{k'', q''}(\ell+\ell', r+r')$ where
            \[
            k'' \coloneqq \max\{k, k', \ell+\ell'+r+r'\}, \quad
            q'' \coloneqq \max\{q + \ell' + r', q' + \ell + r \}.
            \]
    
            \item If $\boldsymbol{F} \in \mathcal{P}^k(\ell, r)$ and $\boldsymbol{F'} \in \mathcal{P}^{k'}(\ell', r')$, then $\boldsymbol{F} \otimes \boldsymbol{F'} \in \mathcal{P}^{k''}(\ell+\ell', r+r')$ for 
            \[
                 k'' \;\coloneqq\; \min\!\left\{\max\{k, k' + \ell + r\},\, \max\{k', k + \ell + r\}\right\}.
            \]
        \end{enumerate}
    \end{lemma}
    \begin{proof}
        The first assertion was proven in \cite[Lemma 5.10]{DPS2025} by joining the tree decompositions of $\boldsymbol{F}$ and $\boldsymbol{F'}$ at a new root bag containing all labelled vertices in $\boldsymbol{F} \otimes \boldsymbol{F'}$.

        For the second claim, observe that the depth of this new rooted decomposition is $q'' \coloneqq \max\{q + \ell' + r', q' + \ell + r \}$.

        For the last claim, let $B_1, \dots, B_s$
        and $B'_1, \dots, B'_{s'}$ denote the path decompositions of $\boldsymbol{F}$ and $\boldsymbol{F'}$, respectively.
        Suppose wlog that $B_1$ and $B'_1$ contain only labelled vertices.
        Then $B_1 \cup B'_1, B_2 \cup B'_1, \dots, B_s \cup B'_1, B'_2, \dots, B'_{s'}$
        is a path decomposition of $\boldsymbol{F} \otimes \boldsymbol{F'}$ with bags of size at most $\max\{k+ \ell' + r', k'\}$.
        Interchanging the role of $\boldsymbol{F}$
        and $\boldsymbol{F'}$ yields the other decomposition.
    \end{proof}

    \cref{eq:tensor-disjoint-union,lem:disjoint-union} imply the following lemma by linearity.

    \begin{lemma}[Disjoint union for homomorphism polynomial maps]
    \label{lem:disjointUnionHomPoly}
    Let $k, k', q, q', \ell, \ell', r, r' \in \mathbb{N}$.
    \begin{enumerate}
        \item If $\phi \in \mathfrak{T}^k_{n,m}(\ell, r)$ and $\phi' \in \mathfrak{T}^{k'}_{n,m}(\ell', r')$, then $\phi \otimes \phi' \in \mathfrak{T}^{k''}_{n,m}(\ell+\ell', r+r')$.
    
        \item If $\phi \in \mathfrak{T}^{k,q}_{n,m}(\ell, r)$ and $\phi' \in \mathfrak{T}^{k', q'}_{n,m}(\ell', r')$, then $\phi \otimes \phi' \in \mathfrak{T}^{k'', q''}_{n,m}(\ell+\ell', r+r') $.
    
        \item If $\phi \in \mathfrak{P}^k_{n,m}(\ell, r)$ and $\phi' \in \mathfrak{P}^{k'}_{n,m}(\ell', r')$, then $\phi \otimes \phi' \in \mathfrak{P}^{k''}_{n,m}(\ell+\ell', r+r')$.
    \end{enumerate}
    Here, $k''$ and $q''$ are as in the respective cases of \cref{lem:disjoint-union}.
    \end{lemma}

    \subsubsection{Gluing products and point-wise multiplication}
    The gluing operation for labelled graphs, denoted by $\boldsymbol{F} \odot \boldsymbol{F'} \in \calG(\ell, r)$, is defined as follows: take the disjoint union of the underlying graphs $F$ and $F'$, and identify the vertices carrying the left labels $\veca$ with the vertices labelled $\veca'$, and similarly for the vertices labelled $\vecb$ and $\vecb'$.

    For homomorphism polynomial maps $\phi, \phi' \in \mathfrak{G}_{n,m}(\ell, r)$, define their pointwise product by
    \[
        (\phi \odot \phi')(\vecv, \vecw) \;\coloneqq\; \phi(\vecv, \vecw) \cdot \phi'(\vecv, \vecw).
    \]
    
    In~\cite[Lemma 5.11]{DPS2025}, it was shown that the set of linear combinations of homomorphism polynomials arising from patterns of bounded treewidth is closed under the gluing operation. We observe that the same closure property holds for the classes defined by bounded treedepth and bounded pathwidth as well.

    \begin{lemma}[Gluing]
    \label{lem:gluing}
    Let $k, \ell, r \in \N$, and define
    \[
        k'' \;\coloneqq\; \min\!\left\{\max\{k, k' + \ell + r\},\, \max\{k', k + \ell' + r'\}\right\}.
    \]
    Then the following hold for the gluing operation $\boldsymbol{F} \odot \boldsymbol{F'} \in \calG(\ell, r)$ defined above:
    \begin{enumerate}
        \item If $\boldsymbol{F}, \boldsymbol{F'} \in \calT^k(\ell, r)$, then $\boldsymbol{F} \odot \boldsymbol{F'} \in \calT^k(\ell, r)$.
    
        \item If $\boldsymbol{F}, \boldsymbol{F'} \in \calT^{k,q}(\ell, r)$, then $\boldsymbol{F} \odot \boldsymbol{F'} \in \calT^{k,q}(\ell, r)$.
    
        \item If $\boldsymbol{F} \in \calP^k(\ell, r)$ and $\boldsymbol{F'} \in \calP^{k'}(\ell, r)$, then $\boldsymbol{F} \odot \boldsymbol{F'} \in \calP^{k''}(\ell, r)$.
    \end{enumerate}
    \end{lemma}
    \begin{proof}
        The first claim in proved in \cite[Lemma 5.11]{DPS2025}. 
        The same argument also proves the second claim.
        For the third claim, let $B_1, \dots, B_s$ and $B'_1, \dots, B'_r$ denote the bags of a path decomposition of $\boldsymbol{F}$ and $\boldsymbol{F}'$, respectively, such that $B_1$ and $B'_1$ are comprised precisely of the labelled vertices, respectively.
        Then both $B'_1 \cup B_1, B'_2 \cup B_1,  \dots, B'_r \cup B_1, B_1, B_2 \dots, B_s$
        and $B_1 \cup B'_1, B_2 \cup B'_1,  \dots, B_s \cup B'_1, B'_1,B'_2, \dots, B'_r$
        are path decompositions of $\boldsymbol{F} \odot \boldsymbol{F}'$.
        The width of the former is $\max\{k' + \ell + r, k\} -1$
        while the width of the latter is $\max\{k + \ell' + r', k'\} -1$, as desired.        
    \end{proof}

    For pathwidth we also prove the following lemma for repeated application of the gluing operation.

    \begin{lemma}
     \label{lem:repeat-gluing}
        Let $k, \ell, r \in \N$.
        If $\boldsymbol{F}_1, \dots, \boldsymbol{F}_s \in \mathcal{P}^{k}(\ell, r)$, then $\boldsymbol{F}_1 \odot \dots \odot \boldsymbol{F}_s \in \mathcal{P}^{k + \ell +r}(\ell, r)$.  
    \end{lemma}
    \begin{proof}
        The proof is by applying \cref{lem:gluing} inductively.
        Consider
        \[
            \boldsymbol{F}_1 \odot \dots \odot \boldsymbol{F}_s
            = (\boldsymbol{F}_1 \odot \dots \odot \boldsymbol{F}_{s-1}) \odot \boldsymbol{F}_s.
        \]
        By hypothesis and assumption, the first factor is in $\mathcal{P}^{k + \ell +r}(\ell, r)$ and the second factor in $\mathcal{P}^{k}(\ell, r)$.
        By \cref{lem:gluing}, their product is in $\mathcal{P}^{k'}(\ell, r)$
        for
        \[
            k' = \min\!\left\{\max\{k + \ell + r, k + (\ell + r) \}, \max\{k, (k + \ell + r) + (\ell +r) \}\right\}
            = k + \ell + r,
        \]
        as desired.
    \end{proof}

    \begin{lemma}[Gluing for homomorphism polynomial maps]
    \label{lem:gluingHomPolyMaps}
    Let $k, \ell, \ell', r, r',n,m \in \N$, and define
    \[
        k'' \;\coloneqq\; \min\!\left\{\max\{k, k' + \ell + r\}, \max\{k', k + \ell' + r'\}\right\}.
    \]
    Then the following hold for the gluing operation $\phi \odot \phi'$ defined above:
    \begin{enumerate}
        \item If $\phi, \phi' \in \mathfrak{T}^k_{n,m}(\ell, r)$ then $\phi \odot \phi' \in \mathfrak{T}^{k}_{n,m}(\ell, r)$
    
        \item If $\phi, \phi' \in \mathfrak{T}^{k,q}_{n,m}(\ell, r)$ then
        $\phi \odot \phi' \in \mathfrak{T}^{k, q}_{n,m}(\ell, r)$
    
        \item If $\phi \in \mathfrak{P}^k_{n,m}(\ell, r)$ and $\phi' \in \mathfrak{P}^{k'}_{n,m}(\ell, r)$, then $\phi \odot \phi' \in \mathfrak{P}^{k''}_{n,m}(\ell, r)$.
    \end{enumerate}
    \end{lemma}
    
    \subsubsection{Restricted sums}
    Let $\ell \geq 1$ and $\boldsymbol{F} = (F, \veca, \vecb) \in \calG(\ell,r)$ be a labelled bipartite graph. For any $i\in [\ell]$, define
    $$
        \Sigma_i \, \boldsymbol{F} \;\coloneqq\; \left(F, \veca[i/], \vecb\right) \;\in\; \calG(\ell-1, r)
    $$ 
    as the graph obtained by dropping the $i$-th left label. Analogously, for a homomorphism polynomial map $\phi$ define 
    $$
        \left(\Sigma_i\, \phi\right)(\vecv[i/], \vecw) \;\coloneqq\; \sum_{u \in [n]} \phi(\vecv[i/u], \vecw),
    $$ 
    for all $\vecv \in [n]^\ell$ and $\vecw \in [m]^r$. Here, $\vecv[i/u]$ is the tuple $\vecv$ with the $i$-th entry replaced by $u$.

    Since unlabelling does neither affect the underlying bipartite graphs nor their decompositions,
    the following identities hold, see \cite[Lemma~5.9]{DPS2025}.

    \begin{lemma}[Unlabelling]
    \label{lem:sum}
    Let $k, \ell, r \in \N$ with $\ell \geq 1$, and let $i \in [\ell]$. 
    \begin{enumerate}
        \item If $\boldsymbol{F} \in \calT^k(\ell, r)$, then $\Sigma_i \boldsymbol{F} \in \calT^k(\ell - 1, r)$.
    
        \item If $\boldsymbol{F} \in \calT^{k,q}(\ell, r)$, then $\Sigma_i \boldsymbol{F} \in \calT^{k,q}(\ell - 1, r)$.
    
        \item If $\boldsymbol{F} \in \calP^k(\ell, r)$, then $\Sigma_i \boldsymbol{F} \in \calP^k(\ell - 1, r)$.
    \end{enumerate}
    \end{lemma}

    The identities from \cref{lem:sum} can be lifted to homomorphism polynomials.
    Moreover, they can be generalised to the following restricted sum operator.
    Let $i \in [\ell]$ and $J \subseteq [\ell] \setminus \{i\}$.
    For $\phi \in \mathfrak{G}_{n,m}(\ell, r)$, define
    \[
        \left(\Sigma_{i, J}\, \phi\right)(\vecv[i/], \vecw) \;\coloneqq\; \sum_{u \in [n] \setminus \setdef{v_j}{j \in J}} \phi(\vecv[i/u], \vecw)
    \]
    The following identities hold.
    
    \begin{lemma}[Restricted summation for homomorphism polynomial maps]\label{lem:restricted-sums}
    Let $k, \ell, r \in \N$ with $\ell \geq 1$, and let $i \in [\ell]$ and $J \subseteq [\ell] \setminus \{i\}$. 
    Suppose that $\ell + r\leq k$ and $k \leq q$.
    Let $n,m \in \mathbb{N}$.
    The following hold for the restricted sum operation defined above:
    \begin{enumerate}
        \item If $\phi \in \mathfrak{T}^k_{n,m}(\ell, r)$, then $\Sigma_{i, J} (\phi) \in \mathfrak{T}^k_{n,m}(\ell - 1, r)$.
    
        \item If $\phi \in \mathfrak{T}^{k,q}_{n,m}(\ell, r)$, then $\Sigma_{i, J}(\phi) \in \mathfrak{T}^{k,q}_{n,m}(\ell - 1, r)$.
    
        \item If $\phi \in \mathfrak{P}^k_{n,m}(\ell, r)$, then $\Sigma_{i, J}(\phi) \in \mathfrak{P}^k_{n,m}(\ell - 1, r)$.
    \end{enumerate}
    \end{lemma}
    \begin{proof}
        The first assertion follows from \cite[Lemma~5.12]{DPS2025}.
        For the other two assertions, recall its proof.
        It is shown that $\Sigma_{i, J} \phi = \Sigma_i (\delta \otimes \phi)$
        for some $\delta \in \mathfrak{G}_{n,m}(\ell, r)$ that is obtained from some graph $\boldsymbol{D}^{i,j} \in \mathcal{G}(\ell, r)$ without unlabelled vertices by taking gluing powers and linear combinations.
        Here, $\boldsymbol{D}^{i,j}$ is the $(\ell, r)$-labelled edge-less $(\ell-1, r)$-vertex bipartite graph whose labels are placed such that the $i$-th and $j$-th left-label reside on the same vertex while all other labels reside on distinct vertices, see \cite[(6)]{DPS2025}.
        The homomorphism polynomial map $\delta$ is given by $p(\sum_{j \in J} \boldsymbol{D}^{i,j}_{n,m})$
        for some polynomial $p\in \mathbb{K}[x]$.

        For the treedepth case, by \cref{lem:gluingHomPolyMaps}, $\delta \otimes \phi \in \mathfrak{T}^{k,q}_{n,m}(\ell, r)$.
        By \cref{lem:sum}, $\Sigma_i (\delta \otimes \phi) \in \mathfrak{T}^{k,q}_{n,m}(\ell, r)$.

        For the pathwidth case,
        more care is required.
        Since all vertices in $\boldsymbol{D}^{i,j}$
        are labelled, the same holds true for any graph obtained from it by taking gluing products.
        Hence, every constituent of $\delta$ is a $(\ell,r)$-labelled bipartite graph without unlabelled vertices.
        Since all labelled vertices of any constituent of $\phi$ are in the same bag of some path decomposition,
        taking the gluing product of $\delta$ with $\phi$,
        does not increase the pathwidth.
        In symbols,  $\delta \otimes \phi \in \mathfrak{P}^{k}_{n,m}(\ell, r)$.
        By \cref{lem:sum}, $\Sigma_i (\delta \otimes \phi) \in \mathfrak{P}^{k}_{n,m}(\ell-1, r)$.
    \end{proof}

    \subsubsection{Restricted products}
    Let $i\in [\ell]$ and $J \subseteq [\ell] \setminus \{i\}$. For a $\phi \in \mathfrak{P}^k_{n,m}(\ell, r)$, define the restricted product on $\mathfrak{P}^k_{n,m}(\ell, r)$ as follows:
    \[
        \left(\Pi_{i,J}\phi\right) \; \coloneqq \; \prod_{u \in [n] \setminus \setdef{v_j}{j \in J}} \phi(\vecv[i/u], \vecw),
    \]
    for all $\vecv \in [n]^\ell$ and $\vecw \in [m]^r$. In \cite[Section 5.2.7]{DPS2025}, it was proved that $\mathfrak{T}^k$ is closed under restricted products.
In the following lemma, we prove the analogous result for the bounded treedepth $\mathfrak{T}^{k,q}_{n,m}$ and pathwidth class $\mathfrak{P}^k_{n,m}$. The proof closely follows the argument in~\cite[Section 5.2.7]{DPS2025}, and we encourage the reader to consult that section for background before proceeding.

    \begin{lemma}[Restricted products]
    \label{lem:restrictedProd-Pathwidth}
        Let $k,\ell, r, n,m,q \in \mathbb{N}$.
        Suppose that $\ell +r \leq k$.
        For $i \in [\ell]$ and $J \subseteq [\ell] \setminus \{i\}$, the following holds:
        \begin{enumerate}
            \item If $\phi \in \mathfrak{T}^{k,q}_{n,m}(\ell, r)$, then $\Pi_{i, J}(\phi) \in \mathfrak{T}^{k,q}_{n,m}(\ell - 1, r)$.
            \item If $\phi \in \mathfrak{P}^k_{n,m}(\ell, r)$, then $\Pi_{i, J}(\phi) \in \mathfrak{P}^{3 k}_{n,m}(\ell-1, r)$
        \end{enumerate}
    \end{lemma}
    \begin{proof}
        We prove the second assertion, and the first follows analogously by replacing the bounded-pathwidth operations with those for bounded treedepth discussed earlier.
        
        As in \cite[Corollary 5.18]{DPS2025} the proof is by induction on $|J|$.
        First, consider the base case $J = \emptyset$, which corresponds to \cite[Theorem~5.13]{DPS2025}.
        By this theorem, $\Pi_{i}(\phi) \in \mathfrak{T}^k_{n,m}(\ell-1, r)$.
        The constituents of this linear combination are, by \cite[(8)]{DPS2025},
        of the form
        \[
            \Pi_i^\pi(\boldsymbol{F}_1, \dots, \boldsymbol{F}_n) \coloneqq \bigodot_{ P \in [n]/\pi} \left(  \Sigma_i \bigodot_{v \in P} \boldsymbol{F}^v \right)
        \]
        where $\pi$ is a partition of $[n]$ and $\boldsymbol{F}^1, \dots, \boldsymbol{F}^n \in \mathcal{P}^k(\ell, r)$ are some constituents of $\phi$.
        By \cref{lem:repeat-gluing}, since $\ell + r \leq k$,
        it holds that $\bigodot_{v \in P} \boldsymbol{F}^v \in \mathcal{P}^{2k}(\ell, r)$
        for every $P \subseteq [n]$.
        By \cref{lem:sum},
        $\Sigma_i \bigodot_{v \in P} \boldsymbol{F}^v \in \mathcal{P}^{2k}(\ell-1, r)$.
        Lastly, by \cref{lem:gluing}, $\Pi_i^\pi(\boldsymbol{F}_1, \dots, \boldsymbol{F}_n) \in \mathcal{P}^{3k}(\ell-1, r)$.
        By \cite[(8)]{DPS2025}, the pattern graphs appearing in the linear combination   $\Pi_i(\phi)$ are all of this form, so $\Pi_i(\phi) \in \mathfrak{P}^{3k}_{n,m}(\ell-1, r)$.
        This concludes the proof of the base case.

        For the inductive step, consider the last equation in the proof of \cite[Corollary~5.18]{DPS2025}:
        \begin{equation} \label{eq:pathwidth-gluing}
            \Pi_{i, J}(\phi) = \psi \odot  \Pi_{i, J'}(\phi) + (1-\psi) \odot \left( \Pi_{i, J'}(\phi + \boldsymbol{D}^{i,j}_{n,m} ) - \Pi_{i, J'}(\phi)\right).
        \end{equation}
        Here, $J'$ is some subset of $J$ such that $|J| = |J'| + 1$.
        Furthermore, $\boldsymbol{D}^{i,j}$ is the $(\ell, r)$-labelled edge-less $(\ell-1, r)$-vertex bipartite graph whose labels are placed such that the $i$-th and $j$-th left-label reside on the same vertex while all other labels reside on distinct vertices, see \cite[(6)]{DPS2025}.
        Clearly, $\boldsymbol{D}^{i,j} \in \mathcal{P}^{k}(\ell, r)$.
        Finally, $\psi \coloneqq p(\frac1n \sum_{j' \in J'} (\Sigma_i \boldsymbol{D}^{i,j'}_{n,m}))$ for some polynomial $p \in \bbQ[x]$.
        It holds that $\Sigma_i \boldsymbol{D}^{i,j'} \in \mathcal{P}^{k}(\ell -1, r)$.
        Hence, by \cref{lem:repeat-gluing}, $\psi \in \mathfrak{P}^{2k}_{n,m}(\ell-1, r)$.

        By the inductive hypothesis, each outermost gluing product in \cref{eq:pathwidth-gluing}
        has one factor in $\mathfrak{P}^{3k}_{n,m}(\ell-1, r)$
        and one factor in $\mathfrak{P}^{2k}_{n,m}(\ell-1, r)$.
        By \cref{lem:gluing},
        $\Pi_{i, J}(\phi) \in \mathfrak{P}^{k''}_{n,m}(\ell-1, r)$ for
        \[
            k'' = \min\{ \max\{3 k, 2k + \ell-1 + r \}, \max\{ 2k, 3 k + \ell-1 + r \}\}
            \leq 3k,
        \]
        as desired.
    \end{proof}

\subsection{Matrix-symmetric formulas \texorpdfstring{(\symVF)}{(Symmetric VF)}}
    \label{sec:symVF}

In this section, we present the complete proof of~\cref{thm:characterisationVF}. The two directions of the equivalence are established separately in~\cref{sec:hompoly-to-symVF} and~\cref{sec:symVF-to-hompoly}.

    \subsubsection{Bounded treedepth homomomorphism polynomials are in \texorpdfstring{$\symVF$}{Symmetric Formulas}}
    \label{sec:hompoly-to-symVF}

    We begin by proving that homomorphism polynomials of patterns of bounded treewidth can be computed by matrix-symmetric rigid formulas of polynomial orbit size.
    
    \begin{lemma}
    \label{lem:singleTreedepthHomCount}
        Let $F$ be a bipartite multigraph of treedepth $d$, and let $n,m \in \bbN$. Then $\hom_{F,n,m}$ can be computed by a $\Sym_n \times \Sym_m$-symmetric rigid formula of size at most $(|V(F)| \cdot |E(F)| \cdot (n+m))^d$, support size at most $d$, and depth at most $d$.
    \end{lemma}
    \begin{proof}
    Let $T$ be an elimination tree of $F$. We construct a formula $C$ for $\hom_{F,n,m}$ by induction on $T$. 
    For every $u \in V(T)$, let $P(u) \subseteq V(T) = V(F)$ denote the vertices appearing in the unique path from the root of $T$ to $u$ (excluding $u$). 
    In the circuit $C$, we create, for each $u \in V(T)$ and each bipartition-preserving map $\gamma \colon P(u) \to [n] \uplus [m]$, a gate $g_{u,\gamma}$ that is the root of a subformula $C_{u,\gamma}$. This subformula is defined as follows. Let $E(u) \subseteq E(F)$ denote the set of all edges $uv$ in $F$, for a $v \in P(u)$. Assume that $u$ is in the left side of the bipartition of $F$; in the other case, replace $h \colon \{u\} \to [n]$ with $h \colon \{u\} \to [m]$, and $x_{h(u)\gamma(w)}$ with $x_{\gamma(w)h(u)}$ in the expression below.
    \[
    C_{u,\gamma} \coloneqq \sum_{h \colon \{u\} \to [n]} \prod_{uw \in E(u)} x_{h(u)\gamma(w)} \cdot \prod_{v \in uE(T)} C_{v,\gamma \cup h}.  
    \]
    To enforce rigidity of the circuit, we apply the following trick. Fix some bijection $\ell \colon uE(T) \to [|uE(T)|]$.
    In the actual circuit, we multiply each $C_{v,\gamma \cup h}$ in the product $\prod_{v \in uE(T)} C_{v,\gamma \cup h}$ with $1^{\ell(v)}$. This does not change the semantics but it ensures that no circuit automorphism can permute the subcircuits appearing in this product. This removes potential non-trivial circuit automorphisms that move internal gates while keeping all input gates fixed.
    
    Note that $C$ thus defined is indeed a formula because every subformula $C_{v,\gamma \cup h}$ has a unique parent in $C$. This is, for the parent $u$ of $v$ in $T$, the product gate in $C_{u,\gamma}$ for the particular choice of $h$ in the sum in $C_{u,\gamma}$. 
    \begin{claim}
    For every $u \in V(T)$ and every $\gamma \colon P(u) \to [n] \uplus [m]$,
    the subformula
    $C_{u,\gamma}$ is a rigid $\StabP(\gamma(P(u)))$-symmetric formula.
    \end{claim}
    \begin{claimproof}
    This is shown by induction on the structure of $T$. If $u$ is a leaf, then the claim is clearly true for $C_{u,\gamma}$. 
    Now suppose $u$ is a non-leaf node of $T$. By induction, each $C_{v,\gamma \cup h}$ appearing in $C_{u,\gamma}$ is $\StabP(\gamma(P(u)) \cup \{h(u)\})$-symmetric. Likewise, $\prod_{uw \in E(u)} x_{h(u)\gamma(w)}$ is $\StabP(\gamma(P(u)) \cup \{h(u)\})$-symmetric. Hence, the sum over all $h \colon \{u\} \to [n]$ is $\StabP(\gamma(P(u))$-symmetric because its summands are symmetric to each other. Also by induction, each $C_{v,\gamma \cup h}$ is rigid, and this is still true for $C_{u,\gamma}$: By our circuit construction, every non-trivial automorphism of $C_{u,\gamma}$ has to permute summands of the outer sum. But then it moves some input gate $x_{h(u)\gamma(w)}$ to another input gate $x_{h'(u)\gamma(w)}$ with $h'(u) \neq h(u)$. Therefore, $C_{u,\gamma}$ is rigid.
    \end{claimproof}
    The claim shows that $\gamma(P(u))$ is a support for the sum and product gates appearing in $g_{u,\gamma}$. Thus, the support size of $C$ is indeed at most $d$, since $|P(u)| \leq d$.
    For $u$ being the root of $T$, $P(u) = \emptyset$, so by the claim, the total circuit $C$ is $\Sym_n \times \Sym_n$-symmetric, as desired.
    It is clear by definition that it computes $\hom_{F,n,m}$, see also the proof of \cite[Theorem~11]{KomarathPR23}.
    
    It remains to estimate the size of $C$. By induction we can see that $|C_{u,\gamma}| \leq (|V(F)| \cdot |E(F)| \cdot (n+m))^c$, where $c$ is the height of the subtree of $T$ rooted at $u$: Indeed, if $u$ is a leaf, then $|C_{u,\gamma}| \leq |E(F)| \cdot (n+m)$. If $u$ is an internal node with a height $c$ subtree, then, using the induction hypothesis, 
    \begin{align*}
        |C_{u,\gamma}| &\leq (n+m) \cdot |V(F)| \cdot |E(F)| \cdot (|V(F)| \cdot |E(F)| \cdot (n+m))^{c-1} \\
        &= (|V(F)| \cdot |E(F)| \cdot (n+m))^{c}
    \end{align*}
    Here, the factor $|V(F)|$ is an upper bound on $|uE(T)|$, the number of subcircuits $C_{v,\gamma \cup h}$, and the factor $|E(F)|$ is an upper bound on the size of the product $\prod_{uw \in E(u)} x_{h(u)\gamma(w)}$.
    \end{proof}

\begin{corollary}
\label{cor:linCombTreedepthHomCount}
Fix $d \in \bbN$.
Let $(p_{n,m})_{n,m \in \bbN}$ be a family of polynomials such that each $p_{n,m}$ is a linear combination of homomorphism polynomials $\hom_{F_i,n,m}$, where each bipartite multigraph $F_i$ has treedepth at most $d$. 
Then each $p_{n,m}$ admits a $\Sym_n \times \Sym_m$ symmetric rigid circuit of orbit size at most $(n+m)^d$ and depth at most $d$.
If additionally, the number of patterns $F_i$ in each $p_{n,m}$, and the sizes $|F_i|$ are polynomial in $(n+m)$, then the circuit size is polynomial as well.
\end{corollary}
\begin{proof}
Apply \cref{lem:singleTreedepthHomCount} to each individual homomorphism polynomial in the linear combination $p_{n,m}$, yielding a rigid symmetric formula with support size and depth~$d$ for each $\hom_{F_i,n,m}$. Then the formula $C$ for $p_{n,m}$ is the weighted sum of these formulas. To enforce rigidity of $C$, we can simply multiply each $\hom_{F_i,n,m}$ with $1^{\ell(i)}$, for a different $\ell(i) \in \bbN$ for each $i$.
By \cref{lem:constantSupportOfGates}, each of the formulas for $\hom_{F_i,n,m}$ has orbit size at most $(n+m)^d$. By construction, all automorphisms of $C$ must stabilise the respective formulas for the $\hom_{F_i,n,m}$, so the orbit size of $C$ is not greater than the orbit size of the formulas for the $\hom_{F_i,n,m}$. The second part of the statement follows then directly from the circuit size bound in \cref{lem:singleTreedepthHomCount}.
\end{proof}

\subsubsection{\texorpdfstring{$\symVF$}{Symmetric VF} computes bounded-treedepth homomorphism polynomials}
    \label{sec:symVF-to-hompoly}

Let $(C_{n,m})_{n,m \in \bbN}$ be a family of matrix-symmetric formulas of orbit size at most $(n+m)^d$, where $d \in \bbN$ is some constant.
Our goal is to show that the polynomials computed by the $C_{n,m}$ are linear combinations of homomorphism polynomials for patterns of constant treedepth. 
By \cref{thm:symVP-char}, the treewidth of the patterns is constant. 
In order to also bound the treedepth, we introduce the notion of \emph{support depth} of symmetric circuits, and prove two results: Firstly, symmetric formulas of polynomial orbit size must have constant support depth. 
Secondly, constant support depth of the circuits also implies constant treedepth of the patterns whose homomorphism polynomials are computed by them. 

\paragraph{Bounding the support depth of symmetric formulas.}
Let $C$ be a $\Sym_n \times \Sym_m$-symmetric formula. Call a gate $g$ in $C$ \emph{large} if there exist at least $\min\{n,m\}/2$ distinct children of $g$ that are all in the same $\Stab(g)$-orbit. 
Every child of $g$ whose $\Stab(g)$-orbit has size at least $\min\{n,m\}/2$ is called a \emph{central} child of $g$.
\begin{lemma}
\label{lem:largeGatesInFormulas}
Let $C$ be a $\Sym_n \times \Sym_m$-symmetric formula. Assume that there is a path from the root to an input gate that passes through $d$ large gates, where it always proceeds with one of its central children. Then $\maxorb(C) \geq \frac{1}{2^{d}} \cdot (\min\{n,m\})^d$.
\end{lemma}
\begin{proof}
We show the following claim by induction:
Let $P$ be a path starting in a large gate $g$ and ending in an input gate that passes through $d$ large gates in total, and always proceeds with a central child of a large gate. Then the subformula rooted at $g$ contains a gate whose $\Stab(g)$-orbit has size at least $\frac{1}{2^{d}} \cdot (\min\{n,m \})^d$.

For $d=1$, we do not need to make use of $P$. Let $g$ be the large gate that $P$ starts in. Then $g$ has $\geq \min\{n,m\}/2$ children that are in the same $\Stab(g)$-orbit, so there is a gate with $\Stab(g)$-orbit size at least $\frac{1}{2} \cdot \min\{n,m\}$ in the subformula rooted at $g$.

Now consider the case $d+1$ and let $g_1$ be the large gate that $P$ starts in. Let $g_2$ be the next large gate on the path. By the induction hypothesis, the subformula rooted at $g_2$ contains a gate, say, $g^*$, whose $\Stab(g_2)$-orbit has size at least $\frac{1}{2^d} \cdot (\min\{n,m \})^d$. Because $P$ proceeds from $g_1$ with a central child $h$, there are at least $n/2$ children $h'$ of $g_1$ such that the subformula rooted at $h'$ is $\Stab(g_1)$-symmetric to the one rooted at $h$. 
Consider an arbitrary such $h' = \pi_{h'}(h)$, for a $\pi_{h'} \in \Sym_n \times \Sym_m$. Then in the subcircuit rooted at the large gate $\pi_{h'}(g_2)$, there appears the gate $\pi_{h'}(g^*)$, and its $\Stab(\pi_{h'}(g_2))$-orbit has size at least $\frac{1}{2^d} \cdot (\min\{n,m \})^d$.
So, with each $h' = \pi_{h'}(h) \in \Orb_{\Stab(g_1)}(h)$, we can associate a different gate $\pi_{h'}(g^*)$ with large orbit in the subformula rooted at $h'$. Since all these gates are in the same $\Stab(g_1)$-orbit, that orbit has size at least $\frac{1}{2} \cdot \min\{n,m \} \cdot \frac{1}{2^d} \cdot (\min\{n,m \})^d = \frac{1}{2^{d+1}} \cdot (\min\{n,m \})^{d+1}$. This finishes the inductive step.
\end{proof}

The above lemma turns a lower bound on the number of large gates on a path into an orbit size lower bound for symmetric formulas. It remains to relate this to the supports of the gates.

In a rigid $\Sym_n \times \Sym_m$-symmetric circuit $C$, we call a gate $g$ \emph{support-changing} if $g$ has a child $h$ such that $\sup(h) \setminus \sup(g) \neq \emptyset$.
Every such child $h$ is then called a \emph{support-changing child} of $g$. We call a path from the root of $C$ to an input gate \emph{central} if, for every support-changing gate $g \in P$, the successor of $g$ in $P$ is a support-changing child of $g$.

\begin{definition}[Support depth]
\label{def:supportDepth}
Let $C$ be a $\Sym_n \times \Sym_m$-symmetric circuit.
Let $\Pp$ denote the set of all paths from the root of $C$ to an input gate.
The \emph{support depth} $\supDepth(P)$ of a path $P \in \Pp$ is the number of gates $g \in P$ such that $g$ is support-changing and $P$ contains a support-changing child of $g$.
The \emph{support depth} of a $\Sym_n \times \Sym_m$-symmetric circuit $C$ is 
\[
\supDepth(C) \coloneqq \max_{P \in \Pp} \supDepth(P). 
\]
\end{definition}

\begin{lemma}
\label{lem:supportChangingGatesAreLarge}
Let $C$ be a rigid $\Sym_n \times \Sym_m$-symmetric formula such that $\maxsup(C) \leq \min\{n,m\}/2$. 
Let $g \in V(C)$ be a support-changing gate. 
Then $g$ is large and every support-changing child of $g$ is central.
\end{lemma}
\begin{proof}
Since $g$ is support-changing, let $h$ be one of its children with $\sup(h) \setminus \sup(g) \neq \emptyset$. 
Assume that $\sup_L(h) \setminus \sup_L(g) \neq \emptyset$. 
The case where $\sup_R(h) \setminus \sup_R(g) \neq \emptyset$ is analogous, where $n$ has to be replaced by $m$.
Since we are assuming that $|\sup(h)_L| \leq n/2$,
we have
\begin{align*}
    |\Orb_{\StabP(\sup(g))}(\sup(h)_L \setminus \sup(g)_L)| 
    &= \left\lvert \left\{ S \subseteq [n] \setminus \sup(g)_L : |S| = |\sup(h)_L \setminus \sup(g)_L|  \right\}\right\rvert  \\
    & \geq n/2.
\end{align*}
Every element of $\Orb_{\StabP(\sup(g))}(\sup(h)_L \setminus \sup(g)_L)$ corresponds to a different child of $g$, so $g$ has at least $n/2$ children that are all in the same $\StabP(\sup(g))$-orbit.
Since $\StabP(\sup(g)) \leq \Stab(g)$, these gates are in particular in the same $\Stab(g)$-orbit. Hence, $g$ fulfils our above definition of a large gate, and every support-changing child is central.  
\end{proof}

\begin{corollary}
\label{cor:supportDepthBounded}
Let $C$ be a $\Sym_n \times \Sym_m$-symmetric formula such that $\maxsup(C) \leq \min\{n,m\}/2$. Let $d \in \bbN$ be a constant.
Suppose that $\supDepth(C) \geq d$. Then $\maxorb(C) \geq \frac{1}{2^d} \cdot (\min\{n,m\})^d$.
\end{corollary}
\begin{proof}
Let $P$ be a path witnessing that $\supDepth(C) \geq d$.
By \cref{lem:supportChangingGatesAreLarge}, $P$ passes through $d$ large gates where it proceeds with a central child. 
Hence $\maxorb(C) \geq \frac{1}{2^d} \cdot (\min\{n,m\})^d$ by \cref{lem:largeGatesInFormulas}.
\end{proof}

\paragraph*{Constant support depth implies constant treedepth of the patterns.}

Now we are going to prove that a symmetric formula $C_{n,m}$ with constant support size and constant support depth always computes a polynomial in $\mathfrak{T}^{k,q}_{n,m}$, for some constants $k$ and $q$. By \cref{lem:treedepth}, this means that all appearing pattern graphs have constant treedepth.

The result follows essentially from the proof of \cite[Theorem~1.1]{DPS2025}, when taking into account that in our setting here, not only the support size but also the support depth of the circuits is bounded. Concretely, in \cite[Lemma~5.21]{DPS2025} it is shown inductively that every gate $g$ in a rigid $\Sym_n \times \Sym_m$-symmetric circuit with support size at most $k$ computes a polynomial in $\mathfrak{T}^{2k}_{n,m}(\vec{\sup}_L(g), \vec{\sup}_R(g))$. Here, $\vec{\sup}_L(g), \vec{\sup}_R(g)$ denote the sets $\sup_L(g), \sup_R(g)$, written as ordered tuples such that the ordering matches the corresponding labels in the pattern graphs. 

We argue that if additionally, the support depth of the circuit is bounded by $d$, then each gate $g$ computes a polynomial in $\mathfrak{T}^{2k,d \cdot k}_{n,m}(\vec{\sup}_L(g), \vec{\sup}_R(g))$. The proof is by induction on the circuit structure, and the inductive step is encapsulated in the following lemma.
For a gate $g$, we denote by $\child(g)$ the set of children, i.e.\ gates that are inputs of $g$ in the circuit.

 \begin{lemma}
 \label{lem:inductiveStepTreedepthProof}
 Let $k\in \bbN$.
 Let $C$ be a rigid $\Sym_{n} \times \Sym_{m}$-symmetric circuit. 
 Let $g \in V(C)$ be a multiplication or summation gate with $|\sup(g)| \leq k$.
 Assume that for every $h \in \child(g)$, the polynomial $p_h$ is in $\mathfrak{T}_{n,m}^{2k,q(h)}(\vec{\sup}_L(h), \vec{\sup}_R(h))$ for some $q(h) \in \bbN$. 
 For each $h \in \child(g)$, let $\delta(h) \coloneqq |\sup(g) \setminus \sup(h)|$. Let $q \coloneqq \max_{h \in \child(g)} q(h)+\delta(h)$.
The polynomial computed at $g$ is in $\mathfrak{T}_{n,m}^{2k,q}(\vec{\sup}_L(g), \vec{\sup}_R(g))$. 
\end{lemma}	

The proof of this lemma is exactly the same as the proof of 
\cite[Lemma 5.22]{DPS2025}. The only difference is that now, we also care about the depth of the tree decompositions, not just the width. The proof of \cite[Lemma 5.22]{DPS2025} relies on the fact that the gluing product preserves the treewidth, which by \cref{lem:gluingHomPolyMaps} is also true for the depth. 

Since the proof of Lemma 5.22 in \cite{DPS2025} spans several pages, we only give a summary here and point out where exactly the depth of the tree decompositions grows.
In the setting of \cref{lem:inductiveStepTreedepthProof}, let $g \in V(C)$ be a gate with $|\sup(g)| \leq k$.

The general case is that $g$ is support-changing, but it is instructive to consider first the easier case where it is not.
If $g$ is not support-changing, then for every $h \in \child(g)$, $\sup(h) \subseteq \sup(g)$. By the assumptions of the lemma, for each $h \in \child(g)$, $p_h$ is in $\mathfrak{T}_{n,m}^{2k,q(h)}(\vec{\sup}_L(h), \vec{\sup}_R(h))$. 

Because $\sup(h) \subseteq \sup(g)$, we can show that $p_h$ is also in 
$$\mathfrak{T}_{n,m}^{2k,q(h)+|\sup(g) \setminus \sup(h)|}(\vec{\sup}_L(g), \vec{\sup}_R(g)).$$ 
The reason is this: For every $\boldsymbol{F}_{n,m}(\vec{\sup}_L(h), \vec{\sup}_R(h))$
appearing in $p_h$, we can add $|\sup(g) \setminus \sup(h)|$ many labelled isolated vertices to $F$, yielding a new pattern graph $F'$. Formally, $F'$ is the result of the disjoint union of $F$ with these isolated vertices (see \cref{lem:disjoint-union}).
Because $|\sup(g)| \leq k$, the width of the resulting decomposition is still at most $2k$ but the depth increases by the number of isolated vertices we have added.
Adding more labels to the pattern graphs in this way does not change the polynomial $p_h$ because it corresponds to multiplying each $\boldsymbol{F}_{n,m}(\vec{\sup}_L(h), \vec{\sup}_R(h))$ with $1 = \boldsymbol{J}_{n,m}(\vec{\sup}_L(g) \setminus \vec{\sup}_L(h), \vec{\sup}_R(g) \setminus \vec{\sup}_R(h))$ from \cref{ex:hompoly-constant} (see also the proof of \cite[Lemma 5.21]{DPS2025}). So by \cref{lem:disjointUnionHomPoly}, $p_h \in \mathfrak{T}_{n,m}^{2k,q(h)+|\sup(g) \setminus \sup(h)|}(\vec{\sup}_L(g), \vec{\sup}_R(g))$.

Since we can do this for every $h \in \child(g)$, the polynomial computed at $g$ is the sum or the product over polynomials $p_h \in \mathfrak{T}_{n,m}^{2k,q(h)+\delta(h)}(\vec{\sup}_L(g), \vec{\sup}_R(g))$, where $\delta(h) = |\sup(g) \setminus \sup(h)|$, for all $h \in \child(g)$.
By definition, $\mathfrak{T}_{n,m}^{k,q} \subseteq \mathfrak{T}_{n,m}^{k,q'}$ for every $q' \geq q$. Hence,
\[
\{p_h \mid h \in \child(g)\} \subseteq \mathfrak{T}_{n,m}^{2k,q},
\]
for $q \coloneqq \max_{h \in \child(g)} q(h)+\delta(h)$.
Since $\mathfrak{T}_{n,m}^{2k,q}(\vec{\sup}_L(g), \vec{\sup}_R(g))$ is closed under sums by definition and closed under products by \cref{lem:gluingHomPolyMaps}, the polynomial computed by $g$ is then also in $\mathfrak{T}_{n,m}^{2k,q}(\vec{\sup}_L(g), \vec{\sup}_R(g))$.

It remains to deal with the case that $g$ is support-changing.
Then we group the children of $g$ into $\Stab_{\StabP(\sup(g))}$-orbits. This is possible because by definition of support, the group $\StabP(\sup(g)) \leq \Sym_n \times \Sym_m$ stabilises the gate $g$ and hence fixes $\child(g)$ setwise. Each orbit is dealt with separately. 
Let $h \in \child(g)$ and $O_h \subseteq \child(g)$ its $\StabP(\sup(g))$-orbit. 
Let $S(h) \coloneqq \bigcap_{h' \in O_h} \sup(h') \subseteq \sup(g)$.
We write $S_L(h) \coloneqq S(h) \cap \sup_L(g)$ and $S_R(h) \coloneqq S(h) \cap \sup_R(g)$.
For any two gates $h', h'' \in O_h$, there is a pair of permutations $(\pi, \sigma) \in \Sym_n \times \Sym_m$ such that $(\pi,\sigma)(h') = h''$. Hence also $\pi(\vec{\sup}_L(h')) = \vec{\sup}_L(h'')$ and $\sigma(\vec{\sup}_R(h')) = \vec{\sup}_R(h'')$.  
So, the supports of the gates in $O_h$ are symmetric to each other and agree on $S(h)$. Let $\ell \coloneqq |\sup_L(h)|$ and $r \coloneqq|\sup_R(h)|$. By symmetry, these support sizes are the same for all gates in $O_h$.
Suppose that $g$ is a product gate. Then, according to the proof of \cite[Lemma 5.21]{DPS2025}, what $g$ computes is
\begin{align}
\prod_{O_h} \prod_{h' \in O_h} p_{h'} &= \prod_{O_h} \prod_{h' \in O_h} \prod_{\boldsymbol{v} \in ([n] \setminus \sup_L(g))^{\ell - |S_L(h)|}} \prod_{\boldsymbol{w} \in ([m] \setminus \sup_R(g))^{r - |S_R(h)|}} \phi_{h'}(\vec{S}_L(h) \boldsymbol{v}, \vec{S}_R(h) \boldsymbol{w}). \label{eq:star}
\end{align}
Here, the product $\prod_{O_h}$ ranges over all $\StabP(\sup(g))$-orbits that $\child(g)$ is partitioned into, and $\phi_{h'}(\vec{S}_L(h) \boldsymbol{v}, \vec{S}_R(h) \boldsymbol{w})$ denotes the homomorphism polynomial map with instantiated images for the labels that is computed by $p_{h'}$. 

By assumption of the lemma we know that $\phi_{h'} \in \mathfrak{T}_{n,m}^{2k,q(h)}(\ell, r)$.
By the same argument as in the previous case, we can extend each $\phi_{h'}$ to a homomorphism polynomial map with at most $2k$ labels by adding $|\sup(g) \setminus \sup(h')|$ many isolated vertices to the pattern graphs. Then we can view each $p_{h'}$ as a polynomial in $\mathfrak{T}_{n,m}^{2k,q(h)+\delta(h)}(\vec{\sup}_L(g)\vec{\sup}_L(h'), \vec{\sup}_R(g)\vec{\sup}_R(h'))$.
The value $q$ as defined before is an upper bound for the depth, so we can also say $p_{h'} \in \mathfrak{T}_{n,m}^{2k,q}(\vec{\sup}_L(g)\vec{\sup}_L(h'), \vec{\sup}_R(g)\vec{\sup}_R(h'))$.

Now the first assertion of \cref{lem:restrictedProd-Pathwidth} implies that a product of the form as in \cref{eq:star} evaluates to a polynomial in $\mathfrak{T}_{n,m}^{2k,q}(\vec{\sup}_L(g), \vec{\sup}_R(g))$. That is, the big symmetric product has the effect of \enquote{forgetting} the labels outside of $\sup(g)$. 

If $g$ is a summation gate, the argument is similar, using \cref{lem:restricted-sums}. 
This finishes the proof overview of \cref{lem:inductiveStepTreedepthProof}. For details, we refer to the proof of Lemma 5.22 in \cite{DPS2025}, where $\mathfrak{T}_{n,m}^{k}$ just has to be replaced by $\mathfrak{T}_{n,m}^{k,q}$ everywhere.
By applying \cref{lem:inductiveStepTreedepthProof} inductively on the circuit structure, we obtain:
\begin{lemma}
\label{lem:supportDepthImpliesTreedepth}
Let $k, d \in \bbN$.
Let $C$ be a $\Sym_n \times \Sym_m$-symmetric circuit with $\maxsup(C) \leq k$ and $\supDepth(C) \leq d$. Then the polynomial computed by $C$ is in $\mathfrak{T}_{n,m}^{2k,d \cdot k}$.
\end{lemma}
\begin{proof}
For every gate $g \in V(C)$, let $d(g)$ denote the support-depth of the subcircuit rooted at $g$. 
By induction on the circuit structure, we show for every $g \in V(C)$: The polynomial $p_g$ computed by $g$ is in $\mathfrak{T}_{n,m}^{2k,d(g) \cdot k+ \sup(g)}(\vec{\sup}_L(g), \vec{\sup}_R(g))$, where $d(g)$ is the support-depth of the subcircuit rooted at $g$. 
To begin with, for every input gate $g$, $p_g = x_{ij}$ for some $i \in [n], j \in [m]$. This is in $\mathfrak{T}_{n,m}^{1,2}(i, j)$, and we have $d(g) = 0, |\sup(g)| = 2, k \geq 2$, so the statement holds for $g$.

Now let $g \in V(C)$ be an internal gate and assume the statement holds for all its children. By \cref{lem:inductiveStepTreedepthProof}, $p_g \in \mathfrak{T}_{n,m}^{2k,q}(\vec{\sup}_L(g), \vec{\sup}_R(g))$, 
where $q = \max_{h \in \child(g)} q(h)+\delta(h)$ for $\delta(h) = |\sup(g) \setminus \sup(h)|$. 

By the inductive hypothesis, $q(h) \leq d(h) \cdot k + |\sup(h)|$. 
Suppose $g$ is not support-changing. Then $\sup(h) \subseteq \sup(g)$ for each $h \in \child(g)$, and we have $d(g) = \max_{h \in \child(g)} d(h)$. Therefore, $q \leq d(g) \cdot k + |\sup(g)|$. This finishes the inductive step in case that $g$ is not support-changing.

If $g$ is support-changing, then $d(g) \geq d(h)+1$ for every support-changing child $h$. Recall that $h$ is a support-changing child if $\sup(h) \setminus \sup(g) \neq \emptyset$.
We have for every support-changing $h \in \child(g)$:
\begin{align*}
d(g) \cdot k + |\sup(g)| &\geq d(h) \cdot k + k + |\sup(g)|\\
&\geq d(h) \cdot k + |\sup(h)|+|\sup(g)|\\
&\geq q(h)+\delta(h).
\end{align*}
For every $h \in \child(g)$ that is not support-changing, we have $d(g) \geq d(h)$ and $\sup(h) \subseteq \sup(g)$.
Then $q(h)+\delta(h) \leq d(g) \cdot k + |\sup(g)|$ by the same reasoning as in the case where $g$ is not support-changing.

In total, as desired, $d(g) \cdot k + |\sup(g)|$ is an upper bound for $q = \max_{h \in \child(g)} q(h)+\delta(h)$. This finishes the inductive step.
\end{proof}

    This completes the preparations for the proof of \cref{thm:characterisationVF}.
    \mainThmSymVF*
    \begin{proof}
The backward direction was proved in \cref{cor:linCombTreedepthHomCount}.
        For the other direction, let $(C_{n})_{n \in \bbN}$ be a family of $\Sym_n \times \Sym_n$-symmetric formulas of orbit size at most $O(n^d)$, where $d \in \bbN$ is some constant (note that the theorem is stated for the case $n = m$). 
        By \cref{lem:rigidification}, we may assume that these symmetric formulas are rigid (but possibly involve multiedges). The rigidity is needed to obtain a well-defined notion of support. By \cref{lem:constantSupportOfGates}, there exists a constant $k \in \bbN$ such that $\maxsup(C_{n}) \leq k$ for all $n \in \bbN$. 
        Then by \cref{cor:supportDepthBounded},
        $\supDepth(C_{n}) \leq d$. 
        Hence by \cref{lem:supportDepthImpliesTreedepth}, $C_{n}$ computes a polynomial $p \in \mathfrak{T}_{n,n}^{2k,d \cdot k}$. By \cref{lem:treedepth}, $p$ is a linear combination of homomorphism polynomials of patterns of treedepth at most $d \cdot k$, which is constant. 
    \end{proof}
    
    \cref{thm:characterisationVF} is the only one of our main results where it is essential for the proof that the theorem is stated for $\Sym_n \times \Sym_n$-symmetric polynomials rather than more general $\Sym_n \times \Sym_m$-symmetric ones. The result could also be adapted to the latter setting but then, the growth ratio of $n$ and $m$ would play a role: For example if $m$ grows very slowly in $n$, then the support depth of a formula with respect to $\sup_R$ may be much larger than the support depth with respect to $\sup_L$. To avoid these intricacies, we only consider the case $m = n$ in this theorem.

    \subsection{Matrix-symmetric skew circuits (\texorpdfstring{\symVBP}{Symmetric VBP})}

    We start with the upper bound, showing that any homomorphism polynomial of bounded pathwidth indeed admits a matrix-symmetric skew circuit of polynomial size. The circuit construction is analogous to the one in \cite[Lemma 5.23]{DPS2025} for treewidth.

    \begin{lemma}[Homomorphism polynomials in $\symVBP$]\label{lem:pathwidth-small-circuit}
        Let $k, n, m \in \mathbb{N}$.
        Let $F$ be a bipartite multigraph of pathwidth less than $k$.
        Then $\hom_{F, n,m}$ admits a matrix-symmetric skew circuit of size polynomial in $\lVert F \rVert$, $n$, $m$, and orbit size at most $(n+m)^k$.
    \end{lemma}
    \begin{proof}
        Let $T$ be a path decomposition of $F$ with root $r$. 
        For $s \in V(T)$, let $T_s$ denote the subtree of $T$ rooted at $s$. Let $A \uplus B = V(F)$ be the bipartition of $F$.
        
        We build the circuit by induction on $T$. For every $s \in V(T)$, let $F^s$ denote the subgraph of $F$ induced by the vertices in $\bigcup_{t \in V(T_s)} \beta(t)$.
        Let $s \in V(T)$, and let $\ell = |\beta(s) \cap A|$, $r = |\beta(s) \cap B|$.
        In our circuit construction, we ensure that for every $\boldsymbol{v} \in [n]^\ell, \boldsymbol{w} \in [m]^r$, there is a gate computing $\boldsymbol{F}^s_{n,m}(\boldsymbol{v},\boldsymbol{w})$.

        In the base case, let $s \in V(T)$ be the leaf of the path. 
        Then $\boldsymbol{F}^s_{n,m}(\boldsymbol{v},\boldsymbol{w})$ is a single monomial and can be represented by a single multiplication gate whose fan-in equals the number of edges in $F^s$.

        In the inductive case, let $s \in V(T)$ be an internal node of the path decomposition and $t \in V(T)$ its child.
        Fix $\boldsymbol{v} \in [n]^\ell$,
 		$\boldsymbol{w} \in [m]^r$, $\boldsymbol{a} \in A^\ell$,
 		$\boldsymbol{b} \in B^r$ such that
 		$\beta(s) = \{a_1, \dots, a_\ell, b_1, \dots, b_r\}$.
        Let $E^s \coloneqq E(F^s) \setminus E(F^t)$. 
        Let $\beta(t) = \{a'_1, \dots, a'_{\ell'}, b'_1, \dots, b'_{r'}\}$ be an enumeration of the labels in $\beta(t)$. 
        For simplicity, assume that the elements of $\beta(s) \cap \beta(t)$ appear at the same indices in the enumerations of $\beta(s), \beta(t)$. Let $I_A(s,t) \coloneqq \{ i \in [\ell] \mid a_i = a'_i \}$ and $I_B(s,t) \coloneqq \{ i \in [r] \mid b_i = b'_i \}$ be the indices of the elements of $\beta(s)$ that appear also in $\beta(t)$.
        Let 
        \[
         I(\boldsymbol{v}, \boldsymbol{w}) \coloneqq \{ (\boldsymbol{v'}, \boldsymbol{w'}) \in [n]^{\ell'} \times [m]^{r'} \mid v'_i = v_i, w'_j = w_j \text{ for all } i \in I_A(s,t), j \in I_B(s,t)   \}. 
        \]
        
        Then we can write:
        \begin{align*}
        \boldsymbol{F}^s_{n,m}(\boldsymbol{v},\boldsymbol{w}) = \left(\sum_{(\boldsymbol{v'}, \boldsymbol{w'}) \in  I(\boldsymbol{v}, \boldsymbol{w})} \boldsymbol{F}^t_{n,m}(\boldsymbol{v'},\boldsymbol{w'}) \right) \cdot \prod_{a_ib_j \in E^s} x_{v_iw_j} 
        \end{align*}
        Here, the sum corresponds to forgetting the labels in $\beta(t) \setminus \beta(s)$.
        By induction, we have a skew circuit for $\boldsymbol{F}^t_{n,m}(\boldsymbol{v'},\boldsymbol{w'})$, and clearly, the above expression for $\boldsymbol{F}^s_{n,m}(\boldsymbol{v},\boldsymbol{w})$ can again be realised using multiplication gates that are skew.

        We apply \cref{lem:rigidification} to the final circuit to make it rigid, which allows us to speak about the supports of gates.
        At any stage, the subcircuit computing $\boldsymbol{F}^s_{n,m}(\boldsymbol{v}, \boldsymbol{w})$ is by construction $\StabP_{\Sym_n \times \Sym_m}(\boldsymbol{v}, \boldsymbol{w})$-symmetric, so $\boldsymbol{v}\boldsymbol{w}$ is a support of the gate computing $\boldsymbol{F}^s_{n,m}(\boldsymbol{v}, \boldsymbol{w})$. Note that $|\boldsymbol{v}\boldsymbol{w}| \leq k$, so by \cref{lem:constantSupportOfGates}, the orbit size of the circuit is bounded by $(n+m)^k$.
 		The overall circuit is $\Sym_n \times \Sym_m$-symmetric,
 		since every $(\pi, \sigma) \in \Sym_n \times \Sym_m$ takes the subcircuit computing $\boldsymbol{F}^s_{n,m}(\boldsymbol{v}, \boldsymbol{w}))$ to the subcircuit computing $\boldsymbol{F}^s_{n,m}(\pi(\boldsymbol{v}), \sigma(\boldsymbol{w})))$.
        The size of the circuit is linear in $|V(T)|+|E(F)|$, so it is polynomial in $\lVert F \rVert$.
    \end{proof}

    For the lower bound, we use the following lemma, whose proof is deferred to the subsequent section. 
    \begin{lemma}[restate=SkewToHompoly, label=lem:skew-to-hompoly]
        Every polynomial computed by a rigid $\Sym_n \times \Sym_n$-symmetric skew circuit $C$ with $k \coloneqq \maxsup(C)$
        is in $\mathfrak{P}_{n,m}^{9k}$.
    \end{lemma}

    The characterisation for $\symVBP$ (\cref{thm:sym-vbp-orbit}) follows immediately from the following theorem.

\begin{theorem}[Characterising $\symVBP$]
        A family of $\Sym_n \times \Sym_m$-symmetric polynomials $(p_{n,m})$ is in \symVBP if, and only if, $(p_{n,m})$ can be written as linear combinations of homomorphism polynomials for patterns of bounded pathwidth.
    \end{theorem}
    \begin{proof}
For the backward direction, use \cref{lem:pathwidth-small-circuit} on each homomorphism polynomial in the linear combination, and then rigidify it using \cref{lem:rigidification}.
        For the forward direction, let $C_{n,m}$ be a $\Sym_n \times \Sym_m$-symmetric skew circuit with orbit size at most $\Oo((n + m)^d)$, where $d \in \N$ is a fixed constant. By \cref{lem:rigidification}, we assume that $C_{n,m}$ is rigid, which gives us a well-defined notion of support for its gates. By \cref{lem:constantSupportOfGates}, there exists a constant $k \in \N$ such that $\maxsup(C_{n,m}) \leq k$. Then by \cref{lem:skew-to-hompoly}, $p_{n,m} \in \mathfrak{P}_{n,m}^{9k}$.
    \end{proof}

    \subsubsection{\texorpdfstring{$\symVBP$}{Symmetric VBP} computes bounded-pathwidth homomorphism polynomials}
    \label{sec:lin-comb-hom-polys}

    In this section we give the complete proof of~\cref{lem:skew-to-hompoly}, which proves the harder direction of \cref{thm:sym-vbp-orbit}. 
The main ingredient in the proof of~\cref{lem:skew-to-hompoly} is the following lemma, which shows that the polynomial computed at each gate of a matrix-symmetric skew circuit is a linear combination of homomorphism polynomials of bounded pathwidth.

    \begin{lemma}
    \label{lem:closure-hompoly-lrComb}
        Let $k,n,m \in \N$ and $C$ be a rigid $\Sym_n \times \Sym_m$-symmetric skew circuit. Consider a gate $g$ in $C$ which is either an addition or a multiplication gate such that $|\sup(g)| \leq k$. If for every child $h$ of $g$, the polynomial computed by $h$ is in $\mathfrak{P}_{n,m}^{9k}(\vec{\sup}_L(h), \vec{\sup}_R(h))$, then the polynomial computed at $g$ is in $\mathfrak{P}_{n,m}^{9k}(\vec{\sup}_L(g), \vec{\sup}_R(g))$. \end{lemma}

    The proof of the lemma follows along the lines of \cite[Lemma 5.22]{DPS2025} and spans several pages. We briefly outline the main ideas here and refer the reader to \cite[Section 5.2]{DPS2025} for more details.

    \begin{proof}[Proof sketch]
        We sketch the proof for the case that $g$ is a product gate. The other case is only easier, where \cref{lem:restricted-sums} is used in place of \cref{lem:restrictedProd-Pathwidth}. 
        In a skew circuit, each multiplication gate has at most one child that is an internal (non-input) gate, while the remaining children are input gates. Let $p_g$ denote the polynomial computed at gate $g$. We write $p_g \coloneqq p_{g_1} \cdot p_{g_2}$ to express this computation as a product of two subcomputations, where $g_1$ refers to the internal child (if any), and $g_2$ denotes a new gate we may simply insert that computes the product of the input-gate children of $g$. We handle the product of inputs first.
 
        Let $G := \StabP(\sup(g)) \leq \Sym_n \times \Sym_m$.
        By the definition of support, every $\pi \in G$ extends to a circuit automorphism that stabilises $g$; hence it also fixes $\child(g_2)$ set-wise and partitions $\child(g_2)$ into orbits. 
        Let $\Omega$ be the set of these orbits, and for $h \in \child(g_2)$, let $O_h \in \Omega$ be the $G$-orbit of $h$.
        We can regroup the product at $g_2$ as
        \[
          p_{g_2} \;=\; \prod_{O_h \in \Omega} \; \prod_{h' \in O_h} \; p_{h'}.
        \]
        Since all $h' \in O_h$ are input gates, $\left|\sup_L(h')\right| = \left|\sup_R(h')\right| \leq 1$ for every $h' \in O_h$. In the following we only consider the case where each $h'$ is a variable input gate rather than a constant. Then $\left|\sup_L(h')\right| = \left|\sup_R(h')\right| = 1$.
        As argued for \cref{eq:star} and \cite[Lemma 5.21]{DPS2025} all the gates in an orbit compute the same instantiated homomorphism polynomial map with different labels. 
        For every $O_h \in \Omega$, let $\phi_{h} \in \mathfrak{P}^k_{n,m}(1,1)$ denote this homomorphism polynomial map.
        Let $S(h) \coloneqq \bigcap_{h' \in O_h} \sup(h') \subseteq \sup(g)$ and $S_L(h) \coloneqq S(h) \cap [n], S_R(h) \coloneqq S(h) \cap [m]$. Note that $|S_L(h)|, |S_R(h)| \leq 1$ because $|\sup_L(h)| = |\sup_R(h)| = 1$.
        Then, 
        \begin{equation}
        \label{eq:prodGate-symVBP}
            p_{g_2} \;=\; \prod_{O_h \in \Omega} \; \prod_{v \in ([n] \setminus \sup_L(g))^{1 - \left|S_L(h)\right|}} \; \prod_{w \in ([m] \setminus \sup_L(g))^{1 - \left|S_R(h)\right|}} \phi_h\left(\vec{S}_L(h) v, \vec{S}_R(h) w\right).
        \end{equation}

In the above expression, $\phi_h$ corresponds to a variable, i.e.\ $\phi_h(v,w) = x_{vw}$ for very $v \in [n], w \in [m]$ (see~\cref{ex:hompoly-var}). This is the homomorphism polynomial of an edge with labelled endpoints.    
        We may view this polynomial $x_{vw}$ also as a homomorphism polynomial $\phi'_h(\vec{\sup}_L(g) v, \vec{\sup}_R(g) w) \in \mathfrak{P}^{k}_{n,m}(\vec{\sup}_L(g) v, \vec{\sup}_R(g)w)$ by simply adding $|\sup(g) \setminus \sup(h)|$ many isolated labelled vertices to the labelled edge. This is the same argument as in the proof of \cref{lem:inductiveStepTreedepthProof} above; it uses the fact that the homomorphism polynomial of isolated labelled vertices is $1$, so taking the disjoint union of these vertices and any other pattern graph does not change the homomorphism polynomial (cf.\ \cref{ex:hompoly-constant} and \cref{lem:disjointUnionHomPoly}). 
        With this, we have:
        \[
            p_{g_2} \;=\; \prod_{O_h \in \Omega} \; \prod_{v \in ([n] \setminus \sup_L(g))^{1 - \left|S_L(h)\right|}} \; \prod_{w \in ([m] \setminus \sup_R(g))^{1 - \left|S_R(h)\right|}} \phi'_h\left(\vec{\sup}_L(g) v, \vec{\sup}_R(g) w\right).
        \]
        Note that one or both of the products $\prod_{v \in ([n] \setminus \sup_L(g))^{1 - \left|S_L(h)\right|}}$, $\prod_{w \in ([m] \setminus \sup_R(g))^{1 - \left|S_R(h)\right|}}$ may be empty (e.g.\ in case that $S_L(h) \subseteq \sup_L(g)$).
        Apply \cref{lem:restrictedProd-Pathwidth} to each of these products that are non-empty. With this, we conclude that the product of polynomials computed by gates within each single orbit $O_h$ satisfies
        \[
            \phi_{O_h}(\vec{\sup}_L(g), \vec{\sup}_R(g)) \coloneqq \prod_{h' \in O_h} p_{h'} \;\in\; \mathfrak{P}^{3^{2} \cdot k}_{n,m}(\vec{\sup}_L(g), \vec{\sup}_R(g)).
        \]  
        To complete the argument, we incorporate the internal child $g_1$ of the multiplication gate $g$. By the inductive hypothesis, we have
        $p_{g_1} \in \mathfrak{P}^{9 \cdot k}_{n,m}(\vec{\sup}_L(g_1), \vec{\sup}_R(g_1))$.
        As before, we extend the labelling of $p_{g_1}$ by appending $\left|\sup(g) \setminus \sup(g_1)\right| < k$ isolated labelled vertices. With \cref{lem:disjointUnionHomPoly}, this gives
        \[
            \phi_{g_1} \in \mathfrak{P}^{9 \cdot k}_{n,m}(\vec{\sup}_L(g), \vec{\sup}_R(g)).
        \]
        Therefore, the polynomial computed at gate $g$ can be expressed as
        \[
            p_g \;=\; p_{g_1} \cdot \prod_{O_h \in \Omega} \prod_{h' \in O_h} p_{h'} 
            \;=\; \phi_{g_1}(\vec{\sup}_L(g), \vec{\sup}_R(g)) \cdot 
            \prod_{O_h \in \Omega} \phi_{O_h}(\vec{\sup}_L(g), \vec{\sup}_R(g)).
        \]
        Applying the gluing operation using \cref{lem:gluingHomPolyMaps} yields
        \[
            p_g \in \mathfrak{P}^{9 \cdot k}_{n,m}(\vec{\sup}_L(g), \vec{\sup}_R(g)). \qedhere
        \]
    \end{proof}

    \SkewToHompoly*
    \begin{proof}
        We proceed by induction on the structure of $C_{n,m}$ and show that, for every gate $g$, the polynomial computed at $g$ is in $\mathfrak{P}_{n,m}^{9k}(\vec{\sup}_L(g), \vec{\sup}_R(g))$.
        For the base case, consider the leafs where every gate $g$ of $C_{n,m}$ is either labelled with a variable or a constant. 
        Any constant is a polynomial that belongs to $\mathfrak{P}_{n,m}^{0}(0,0)$ (see \cref{ex:hompoly-constant}).
        If $g$ is an input gate labelled by $x_{i,j}$, consider a graph $F$ consisting of a single edge. 
        Then $x_{i,j} = \boldsymbol{F}_{n,m}(i,j) \in \mathfrak{P}_{n,m}^{2}(i,j)$. Note that $2 \leq k = \maxsup(C)$. This establishes the base case. 

        For the inductive step, the gate $g$ is either a multiplication or an addition gate. In either case we use~\cref{lem:closure-hompoly-lrComb}.
        In particular, the polynomial computed at gate $g$ lies in $\mathfrak{P}_{n,m}^{9k}(\vec{\sup}_L(g), \vec{\sup}_R(g))$.
        Finally, since the minimal support of the output gate is empty, the polynomial computed by $C_{n,m}$ belongs to $\mathfrak{P}_{n,m}^{9k}(0,0) = \mathfrak{P}_{n,m}^{9k}$.
    \end{proof}

    \section{Unconditional separations of symmetric classes}

    This section contains the full details for the results stated in \cref{sec:extended-abstract-unconditional-separation}.
    First, we make \cref{def:w-counting-width} precise.

    A function $p$ on $\bigcup_{n,m \in \mathbb{N}} \bbQ^{n \times m}$ is a \emph{graph parameter} if it is matrix-symmetric, i.e.\ for $X \in \bbQ^{n,m}$ and $(\sigma, \pi) \in \Sym_n \times \Sym_m$, it holds that $p(X) = p((\sigma, \pi)(X))$.
    
    \begin{definition}
        Let $w$ be an $\mathbb{N}$-valued graph parameter.
        Let $n,m \in \mathbb{N}$.
        Let $p$ be a $\Sym_n \times \Sym_m$-invariant function on $\{0,1\}^{n \times m}$.
        The \emph{$w$-counting width} of $p$ is the least integer $k \in \mathbb{N}$ such that for all $G, H \in \{0,1\}^{n,m}$, if $G$ and $H$ are homomorphism indistinguishable over all bipartite multigraphs $F$ with $w(F) < k$,
        then $p(G) = p(H)$.
    \end{definition}

    The $w$-counting width is well-defined since, by \cite{lovasz_operations_1967}, for every $n,m \in \mathbb{N}$,
    there exist finitely many graphs $F_1, \dots, F_r$ such that two $(n,m)$-vertex simple graphs $G$ and $H$ are isomorphic if, and only if, they are homomorphism indistinguishable over $F_1, \dots, F_r$.
    Then $k \coloneqq \max\{ w(F_i)  + 1 \}$ is an upper bound on the $w$-counting width of any $\Sym_n \times \Sym_m$-invariant function on $\{0,1\}^{n \times m}$.

	\begin{remark}
    By \cite{dvorak_recognizing_2010,dell_lovasz_2018}, a graph parameter $p$ has $\tw$-counting width at most~$k$
    if, and only if, it has counting width \cite{dawar_definability_2017} at most~$k$, i.e.\
    any two $\mathsf{C}^k$-equivalent graphs $G$ and $H$ satisy $p(G) = p(H)$.
    Similarly, by \cite{grohe_counting_2020},
    $\td$-counting width at most $k$ corresponds to the quantifier-depth-$k$ fragment $\mathsf{C}_k$ of first-order logic with counting quantifiers.
    By \cite{montacute_pebble-relation_2024},
    $\pw$-counting width at most $k$ corresponds to the restricted-conjunction $k$-variable fragment $\curlywedge \mathsf{C}_k$.
    However, these logical characterisations are not required for what is proven here.
    \end{remark}

    We can now formally prove \cref{thm:symmetric-classes-single-hom}.

    \thmCWsinglehom*
    \begin{proof}
        The forward implications follow from \cref{cor:bounded-w-counting-width}.
        The first assertion is already implied by \cite[Theorem 7.1]{DPS2025}.
        We give a uniform proof for all three backward implications.
        To that end, let $w \in \{\tw, \pw, \td\}$.
        Let $(F_n)$ be a family of bipartite multigraphs such that $w(F_n)$ is unbounded.

        It must be shown that the $w$-counting width of $(\hom_{F_n,n})$ is unbounded.
        To that end, consider the following definition from \cite[34]{DPS2025}.
        For a class  $\mathcal{F}$ of bipartite graphs with fixed bipartition and $n,m \in \mathbb{N}$,
        \[
            \cl_{n}^{\textup{bip}}(\mathcal{F})
            \coloneqq
            \{F \mid \forall (n,n)\text{-vertex } G, H.\ G \equiv_{\mathcal{F}} H \implies \hom(F, G) = \hom(F, H) \}
        \]
        Here, $F$ ranges over bipartite simple graphs,
        $G$ and $H$ range over simple bipartite graphs, $\hom(F,G)$ denotes the number of bipartition-preserving homomorphisms from $F$ to $G$,
        and $\equiv_{\mathcal{F}}$ denotes homomorphism indistinguishability over $\mathcal{F}$.

        For $k \in \mathbb{N}$, let $\mathcal{F}_k \coloneqq \{F \text{ bipartite simple graphs} \mid w(F) \leq k\}$.
        It must be shown that $\cl_{n}^{\text{bip}}(\mathcal{F}_k) \subseteq \mathcal{F}_k$ for all $n$ sufficiently larger than $k$.
        This follows as in the proof of \cite[Theorem 7.1]{DPS2025}
        under the following assumptions, see \cite[Proviso 7.2]{DPS2025}:
        For all simple (not necessarily bipartite) graphs $F, F', F_1, F_2$,
        \begin{enumerate}
            \item if $F$ is a minor of $F'$, 
                  then $w(F) \leq w(F')$,
            \item $w(F_1 + F_2) \leq \max\{ w(F_1) , w(F_2) \}$, where $F_1 + F_2$ denotes the disjoint union of $F_1$ and $F_2$, and
            \item if there exists a weak oddomorphism $F \to F'$, see \cite{roberson_oddomorphisms_2022}, then $w(F') \leq w(F)$.
        \end{enumerate}
        The first two assumptions are well-known properties of $w \in \{\tw, \pw, \td\}$, see \cite{nesetril_sparsity_2012}.
        The final assumption is rather involved, see \cite{seppelt_homomorphism_2024} for more context.
        It was established for treewidth in \cite[Corollary~13]{neuen_homomorphism-distinguishing_2024},
        for pathwidth in \cite[Theorem~6.4.6]{seppelt_homomorphism_2024},
        and for treedepth in \cite[Theorem~3]{fluck_going_2024}.
    \end{proof}

    \thmSymmetricLincomb*
    \begin{proof}
        By \cite[Theorem~7.11]{DPS2025} and \cref{thm:symmetric-classes-single-hom} recalling that \cite[Proviso~7.2]{DPS2025} is satisfied as argued in the proof of \cref{thm:symmetric-classes-single-hom}.
    \end{proof}

    \section{Classical algebraic complexity of homomorphism polynomials and their linear combinations}
    \label{sec:valiant}

    In this section, 
    we explore the connections between our symmetric complexity classes \symVP, \symVBP, and \symVF on the one hand
    and Valiant's complexity classes \VNP, \VP, \VBP, and \VF on the other hand.
    Motivated by the characterisations of the polynomials in \symVP, \symVBP, \symVF
    as linear combinations of homomorphism polynomials of patterns of bounded treewidth, pathwidth, and treedepth, respectively, see \cref{thm:symVP-char,thm:sym-vbp-orbit,thm:characterisationVF},
    we study the classical algebraic complexity of homomorphism polynomials.
    The outline of this section is as follows:

    In \cref{ssec:colourful}, we introduce colourful homomorphism polynomials, recast hardness results from \cite{durand_homomorphism_2016,hrubes_hardness_2017} in this language,
        and prove \cref{lem:minor-projection} which relates colourful homomorphism polynomials for patterns $F$ and $F'$ when $F$ is a minor of $F'$.
        The main result is \cref{thm:colourful-hom-vnp-hard}, a colourful analogue of \cref{thm:uncoloured-hom-complexity}.
        
    \Cref{ssec:uncoloured} is concerned with reducing colourful homomorphism polynomials to uncoloured homomorphism polynomials as in \cref{eq:homPoly}.
        The main technical contribution is \cref{lem:single-hom-projection-minors} which reduces a colourful homomorphism polynomial $\colhom_{S, n}$ for a pattern $S$ to the homomorphism polynomial $\hom_{F, 2^{\Delta(S)} \cdot |V(S)| \cdot n}$ for a pattern $F$ containing $S$ as minor.
        The crux here is to make do with a blow-up of the number of variables by only $2^{\Delta(S)} \cdot |V(S)|$ where $\Delta(S)$ denotes the maximum degree in $S$.
        A direct translation of the main result of  \cite{curticapean_count_2024} to algebraic complexity, 
        however, would only yield a reduction to $\hom_{F, 2^{\Delta(S)} \cdot |V(F)| \cdot n}$.
        This blow-up cannot be afforded when proving any of the results of this section.
        We overcome this issue by 
        using ideas from \cite{seppelt_logical_2024}.
        Finally, we prove \cref{thm:uncoloured-hom-complexity}.
    \thmUncolouredHom*
        
    In \cref{ssec:linearcombinations},
    we devise \cref{lem:uncoloured-simple-lovasz-interpolation}
    which reduces a homomorphism polynomial for a pattern $F$ to a linear combination of homomorphism polynomials containing $F$ with non-zero coefficient.
    This lemma is based on an interpolation argument based on the linear independence of homomorphism polynomials of patterns of sublinear volume, see \cref{lem:bipartite-multigraph-linear-independent}.
    \Cref{lem:uncoloured-simple-lovasz-interpolation} yields a proof of \cref{thm:uncoloured-lincomb-hom-complexity}.

    \thmLincombHomComplexity*

    The remaining section \cref{ssec:parametrised} provides lower bounds for homomorphism polynomials under the assumption $\VFPT \neq \VW$.
    The proofs of the remaining \cref{thm:uncoloured-lincomb-hom-complexity-parametrised}, \cref{thm:hom-parametrised-simplified}, and \cref{cor:symmetric-vs-non-symmetric-computation}
    make use of most of the machinery developed in \cref{ssec:colourful,ssec:uncoloured,ssec:linearcombinations}.
    
    \thmUncolouredhomParametrised*
    \corSymNonSym*
    \thmLincombHomComplexityParametrised*

    \subsection{Colourful homomorphism polynomials}
    \label{ssec:colourful}

    We translate the machinery developed in \cite{curticapean_complexity_2014,curticapean_count_2024} to algebraic complexity.
    The first step is to introduce colourful homomorphism polynomials.
    In contrast to ordinary homomorphism polynomials, they keep track not only of image vertices of homomorphisms but also of their preimages.
    Colourful homomorphism polynomials were also considered in \cite[Definition~7]{KomarathPR23} in the context of monotone algebraic complexity under the term \emph{coloured isomorphism polynomial}, see also \cite{bhargav_monotone_2025}.
    
    \begin{definition}\label{def:colourful-hom-polynomial}
        For a bipartite multigraph $F$ with bipartition $A \uplus B = V(F)$ and a number $n \in \mathbb{N}$,
        consider the \emph{$n$-th colourful homomorphism polynomial}
        \[
            \colhom_{F, n} \coloneqq \sum_{h \colon V(F) \to [n]} \prod_{\substack{ab \in E(F) \\ a \in A , b \in B}} x_{(a, h(a)), (b, h(b))}
        \]
        whose variables are indexed by $(A \times [n]) \times (B \times [n])$.    
    \end{definition}
    
    The input to such a polynomial can be thought of as a bipartite edge-weighted $F$-coloured graph.
    That is, a bipartite graph $G$ with a homomorphism $c \colon G \to F$ and edge weights $E(G) \to \mathbb{K}$.
    Here, each \emph{colour class} $c^{-1}(v)$ for $v \in V(F)$ is of size~$n$.

    \begin{lemma}\label{lem:colourful-containment}
        For every $p$-family $(F_n)_{n \in \mathbb{N}}$ of bipartite multigraphs,
        \begin{enumerate}
            \item $(\colhom_{F_n, n})_{n \in \mathbb{N}}$ is in \VNP.
            \item if $\tw(F_n) \in O(1)$, then $(\colhom_{F_n, n})_{n \in \mathbb{N}}$ is in \VP,
            \item if $\pw(F_n) \in O(1)$, then $(\colhom_{F_n, n})_{n \in \mathbb{N}}$ is in \VBP,
            \item if $\td(F_n) \in O(1)$, then $(\colhom_{F_n, n})_{n \in \mathbb{N}}$ is in \VF.
        \end{enumerate}
    \end{lemma}
    \begin{proof}
        The first claim follows from Valiant's criterion \cite[Proposition~2.20]{burgisser_completeness_2000}.
        The remaining claims follow by inspecting the proofs of \cite[Theorem~5.3]{DPS2025} and
        \cref{lem:pathwidth-small-circuit,lem:singleTreedepthHomCount}, see also \cite[Theorem~11]{KomarathPR23}.
    \end{proof}

    \subsubsection{Hard families of colourful homomorphism polynomials}

    By rephrasing results of \textcite{hrubes_hardness_2017,durand_homomorphism_2016},
    we obtain the following families of colourful homomorphism polynomials which are complete for \VNP, \VP, and \VBP, respectively.
    For $n \in \mathbb{N}$, let $G_{n \times n}$ denote the $n$-by-$n$ grid.
    
    \begin{lemma}\label{lem:colourful-grid-vnp-complete}
        The family $(\colhom_{G_{n \times n}, 4n^2})_{n \in \mathbb{N}}$ is $\mathsf{VNP}$-complete under $p$-projections over any field.
    \end{lemma}
    \begin{proof}
        Containment in $\VNP$ follows from \cref{lem:colourful-containment}.
        We follow the reduction in \cite[Theorem~3.2]{curticapean_complexity_2014_arxiv} of the parametrised clique counting problem to the parametrised problem of counting colourful grids.
        By \cite[Theorem~5.5]{hrubes_hardness_2017},
        the family
        \begin{equation} \label{eq:hrubes}
            \clique_n \coloneqq \sum_{\substack{A \subseteq [2n] \\ |A| = n}} \prod_{\substack{i,j \in A \\ i < j}} y_{i,j}
        \end{equation}
        is $\mathsf{VNP}$\nobreakdash-complete under $p$-projections over any field.
        We project the polynomial $\textup{clique}_n$ from the polynomial $\colhom_{G_{n \times n}, 4n^2}$. 
        Let $y_{u,v}$ for $1 \leq u < v \leq 2n$ denote the variables of $\textup{clique}_n$. 

        A $p$-projection is given by defining a bipartite edge-weighted $G_{n \times n}$-coloured graph $G$.
        To that end, let $V(G_{n \times n}) = [n] \times [n]$.
        The vertices of $G$ are given as follows:
        \begin{itemize}
            \item For each $i \in [n]$ and $v \in [2n]$, the colour class of $(i,i)$ contains the vertex $(i,i, v,v)$.
            \item For each $i, j \in [n]$, 
            $i \neq j$, and every $u, v \in [2n]$ such that $u \neq v$, the colour class $(i,j)$ contains the vertex $(i, j, u, v)$.
        \end{itemize}
        The edge weights are given as follows:
        \begin{itemize}
            \item For every $i \in [n]$, $i \leq j \leq n-1$, $u \in [2n]$, and $1 \leq v < v' \leq 2n$,
            the edge $(i,j,u,v)(i,j+1,u,v')$ has weight $y_{v,v'}$.
            \item For every $j \in [n]$, $1 \leq i \leq j-1$, $v \in [2n]$, and $1 \leq u < u' \leq 2n$, the edge $(i,j,u,v)(i+1,j,u',v)$ has weight $1$.\item All other edges have weight zero.
        \end{itemize}
        With this assignment, the polynomial $\colhom_{G_{n \times n}, 4n^2}$ evaluates as follows:
        \begin{align*}
            \colhom_{G_{n \times n}, 4n^2}(G)
            &= \sum_{h \colon [n] \times [n] \to [2n] \times [2n]} \prod_{uv \in E(F)} x_{(u, h(u)), (v, h(v))}
        \end{align*}
        Let $A \subseteq [2n]$ with $|A| = n$ and write $A = \{v_1, \dots, v_n \}$ where $v_1 < v_2 < \dots < v_n$.
        Up to renaming variables,
        the monomial $\prod_{i,j \in A, i < j} y_{i,j}$ in $\clique_n$
        is equal to the monomial of $\colhom_{G_{n \times n}, 4n^2}(G)$
        for $h \colon [n] \times [n] \to [2n] \times [2n]$ given by $(i,j) \mapsto (v_i, v_j)$.
        
        Conversely, it must be argued that only such monomials survive the evaluation of the polynomial $\colhom_{G_{n \times n}, 4n^2}$ at $G$.
        Each map $h \colon [n] \times [n] \to [2n] \times [2n]$ amounts to picking out a vertex $(i,i,v,v)$ for the diagonal vertices $(i,i) \in V(F)$. Let $v_i \in [2n]$ denote this vertex $v$.
        For an off-diagonal vertex $(i,j) \in V(F)$, $i \neq j$,
        the map $h$ selects a vertex $(i,j, u, w)$.
        By definition of the edges in $G$, it follows that $u = v_i$ and $w = v_j$.
        Furthermore, $v_1 < v_2 < \dots < v_n$.
        Lastly note that the weights are as desired.
    \end{proof}

    For $n \in \mathbb{N}$, write $B_n$ for the largest complete (perfect) binary tree with at most $n$ leaves.

    \begin{lemma}\label{lem:colourful-binary-tree-vp-complete}
        The family $(\colhom_{B_n, n^6})_{n \in \mathbb{N}}$ is \VP-complete under constant-depth $c$\nobreakdash-reductions over any field of characteristic zero.
    \end{lemma}
    \begin{proof}
        Containment in \VP follows from \cref{lem:colourful-containment}.
        By \cite[Theorems 3.8 and 4.5]{durand_homomorphism_2016}, the family $(p_m)$ of polynomials
        \[
            p_m \coloneqq \sum_{h \colon V(B_m) \to [m^6]} \prod_{v \in R(T_m)} X_{h(v)} \prod_{uv \in E(F)} Y_{h(u),h(v)}
        \]
        is \VP-complete under constant-depth $c$\nobreakdash-reductions:
        Here, the variables $X_i$ are indexed by $i \in [m^6]$, the variables $Y_{ij}$ are indexed by unordered pairs $ij \in \binom{[m^6]}{2}$.
        The set $R(B_m) \subseteq V(B_m)$ denotes the set of all non-root vertices of $B_m$  that are the right children of their respective parents.
        We show that $p_m$ is computed by a constant-depth $m^{O(1)}$-size circuit with oracle gates to $\colhom_{B_m, m^6}$.
        To that end, define a $B_m$-coloured edge-weighted graph $G = (x_{(i, u), (j,v)})_{i,j \in V(B_m), u,v \in [m^6]}$.
        \begin{itemize}
            \item If $j$ is the right child of $i$, let $x_{(i, u), (j,v)} \coloneqq X_v Y_{u, v}$.
            \item If $j$ is the left child of $i$, let $x_{(i, u), (j,v)} \coloneqq Y_{u, v}$.
            \item Otherwise, let $x_{(i, u), (j,v)} \coloneqq 0$.
        \end{itemize}
        With this assignment, it follows that $\colhom_{B_m, m^6}$ evaluates to $p_m$.
    \end{proof}
    
    For $n\in \mathbb{N}$, let $P_n$ denote the path on $n$ vertices.

    \begin{lemma}\label{lem:colourful-path-vbp-complete}
        The family $(\colhom_{P_n, n^2})_{n \in \mathbb{N}}$ is \VBP-complete under constant-depth $c$\nobreakdash-reductions over any field of characteristic zero.
    \end{lemma}
    \begin{proof}
        Containment in \VBP follows from \cref{lem:colourful-containment}.
        By \cite[Theorem 5.1]{durand_homomorphism_2016},
        the family of polynomials
        \begin{equation}\label{eq:durand-vbp}
            \sum_{h \colon [m+1] \to [m^2]} X_{h(1)} X_{h(m+1)} \prod_{i \in [m]} Y_{h(i)h(i+1)}
        \end{equation}
        is \VBP-complete under constant-depth $c$\nobreakdash-reductions.
        There is a variable $X_i$ for every $i \in [m^2]$ and $Y_{ij}$ for every unordered pair $ij \in \binom{[m^2]}{2}$.
        Thus, implicitly, $h(i) \neq h(i+1)$ for all $i \in [m]$.
        This family of polynomials projects to $(\colhom_{P_n, n^2})_{n \in \mathbb{N}}$.
        Define a $P_{m+2}$-coloured edge-weighted graph $G = (x_{(i, u), (i+1,v)})_{i \in [n-1], u,v \in [m^2]}$
        by setting
        \begin{itemize}
            \item $x_{(1,u),(2,v)}$ to zero if $u \neq v$ and to $X_u$, otherwise.
            \item $x_{(i,u),(i+1,v)}$, for $1 < i < m-1$, to $Y_{uv}$ if $u \neq v$ and to zero otherwise.
            \item $x_{(m+1, u), (m+2, v)}$ to zero if $u \neq v$ and to $X_u$, otherwise.
        \end{itemize}
        With this assignment, it follows that $\colhom_{P_{m+2}, m^2}$ evaluates precisely to the polynomial in \cref{eq:durand-vbp}.
        By padding, this further reduces to $\colhom_{P_{m+2}, (m+2)^2}$.
    \end{proof}

    \subsubsection{Projecting colourful homomorphism polynomials along minors}

    The following \cref{lem:minor-projection}
    relates the colourful homomorphism polynomials of patterns $F$ and $F'$ when $F$ is a minor of $F'$.
    It is the algebraic analogue of \cite[Lemma~3.1]{curticapean_complexity_2014_arxiv}.

    \begin{lemma}\label{lem:minor-projection}
        Let $F$ and $F'$ be bipartite simple graphs such that $F$ is a minor of $F'$.
        For $n \in \mathbb{N}$,
        $\colhom_{F, n}$ is a projection of $\colhom_{F', n}$ over any field.
    \end{lemma}
    \begin{proof}
        We follow the proof of \cite[Lemma~3.1]{curticapean_complexity_2014_arxiv}.
        Let $V(F) = \{1,\dots, k\}$.
        Let $B_0, B_1, \dots, B_k \subseteq V(F')$ denote branch sets for $F$ in $F'$.
        That is, each induced subgraph $F'[B_i]$ for $i \in [k]$ is connected
        and,
        for every $ij \in E(F)$, there is an edge in $F'$ from some vertex in $B_i$ to some vertex in $B_j$.
        Pick $w_i \in B_i$ for $i \in [k]$.
        We seek to construct an $F'$-coloured edge-weighted graph $G'$ on at most $n$ vertices per colour class such that $\colhom_{F, n} = \colhom_{F', n}(G')$.
        Write $y_{(i, u), (j, v)}$ where $i,j \in V(F)$ and $u,v \in [n]$ for the variables of $\colhom_{F, n}$.
        The graph $G'$ is constructed as follows:
        \begin{enumerate}
            \item For every $i \in V(F)$ and $v \in [n]$, it contains a copy of the $F'$-coloured graph $F'[B_i]$, denoted by $L_{i,v}$.
            \item For $ij \in E(F)$, $u ,v \in [n]$, 
            insert edges all from $L_{i,u}$ to $L_{j,v}$ with weight $1$, 
            except for the edge from the copy of $w_i$ in $L_{i,u}$ to the copy of $w_j$ in $L_{j,v}$ which carries weight  $y_{(i,u), (j, v)}$.
            \item For $ij \not\in E(F)$, $u,v \in [n]$, 
            insert all edges between $L_{i,u}$ and $L_{j,v}$ of weight $1$.
            \item Add a copy of $F'[B_0]$ to $G'$ and connect it to all other vertices with weight $1$.
        \end{enumerate}
        Note that the resulting graph is $F'$-coloured with at most $n$ vertices per colour class.

        A map $h \colon V(F) \to [n]$
        induces a map $h' \colon V(F') \to [n]$ by mapping $i'$ to $h(j)$ where $i' \in B_j$.
        The corresponding monomial in $\colhom_{F', n}(G')$ is $\prod_{ij \in E(F)} y_{(i,h(i)), (j, h(j))}$,
        as appearing in $\colhom_{F, n}$.

        Conversely, 
        let $h' \colon V(F') \to [n]$ be the map indexing a monomial of $\colhom_{F', n}(G')$.
        We show that $h'$ must be induced by some $h \colon V(F) \to [n]$ as above.
        Let $r,s \in B_i$ for some $i \in [k]$.
        If $h'(r) \neq h'(s)$,
        then the monomial indexed by $h'$ is zero since $F'[B_i]$ is connected and the gadgets $L_{i, h'(r)}$ and $L_{i, h'(s)}$ are disjoint.
    \end{proof}

    Equipped with \cref{lem:minor-projection},
    we can prove \cref{thm:colourful-hom-vnp-hard}, which describes the complexity of colourful homomorphism polynomials.
    To that end, we recall \cref{eq:tw-td-pw} and the following relationship between treewidth, pathwidth, and treedepth.
    For every $n$-vertex graph $F$, by \cite[Corollary~6.1]{nesetril_sparsity_2012},
    \begin{equation}
        \td(F) \leq \left(\tw(F) +1 \right) \log(n).\label{eq:treedepth-treewidth}
    \end{equation}
    Crucially, \cref{eq:treedepth-treewidth,eq:tw-td-pw} imply that any $p$-family $(F_n)$ of bounded treewidth satisfies $\pw(F_n) \in O(\log n)$.
    Similarly, if $\pw(F_n) \in O(1)$, then $\td(F_n) \in O(\log n)$ by \cref{eq:tw-td-pw}.
    For this reason, we need the following very recent graph minor theorems in order to show \VP- and \VBP-hardness for families of colourful homomorphism polynomials with $\pw(F_n) \in \Theta(\log n)$ and with $\pw(F_n) \in O(1)$ and $\td(F_n) \in \Theta(\log n)$, respectively.
    The bounds guaranteed by the graph minor theorems from \cite{kawarabayashi_polynomial_2021} for treedepth and from \cite{groenland_approximating_2023} for pathwidth are too crude for this purpose.

    \begin{theorem}[\cite{chekuri_polynomial_2016,groenland_approximating_2023,hatzel_tight_2024}]\label{thm:effective-minor}
        There exists a universal constant $\tau > 0$ such that, for every graph $F$ and $k \in \mathbb{N}$,
        \begin{enumerate}
            \item if $\tw(F) \geq k^\tau$, then $F$ contains a $k \times k$ grid as a minor,
            \item if $\pw(F) \geq \tau (\tw(F) + 1) k$, then $F$ contains a complete binary tree of height $k$ as a minor,
            \item if $\td(F) \geq \tau (\pw(F) +1 ) k$, then $F$ contains a $2^k$-vertex path as subgraph.
        \end{enumerate}
    \end{theorem}

    The following \cref{thm:colourful-hom-vnp-hard} is our main result on the algebraic complexity of colourful homomorphism polynomials.

    \begin{theorem}\label{thm:colourful-hom-vnp-hard}
        Let $\epsilon > 0$ and $(F_n)$ be a $p$-family of bipartite simple graphs.
        \begin{enumerate}
            \item if $\tw(F_n) \geq n^\epsilon$ for all $ n \geq 2$,
            then $(\colhom_{F_n, n})$ is \VNP-complete under $p$-projections over any field.
            \item if $\tw(F_n) \in O(1)$ and $\pw(F_n) \geq \epsilon \log(n)$ for all $ n \geq 2$,
            then $(\colhom_{F_n, n})$ is \VP-complete under constant-depth $c$\nobreakdash-reductions over any field of characteristic zero.
            \item if $\pw(F_n) \in O(1)$ and $\td(F_n) \geq \epsilon \log(n)$ for all $ n \geq 2$,
            then $(\colhom_{F_n, n})$ is \VBP-complete under constant-depth $c$\nobreakdash-reductions over any field of characteristic zero.
        \end{enumerate}
    \end{theorem}
    \begin{proof}
        In all cases, containment follows from \cref{lem:colourful-containment}.
        Invoking \cref{lem:colourful-grid-vnp-complete,lem:colourful-path-vbp-complete,lem:colourful-binary-tree-vp-complete,lem:minor-projection},
        the proof is by projecting the colourful homomorphism polynomials of grids, complete binary trees, and paths from the $\colhom_{F_n, n}$. Let $\tau > 0$ denote the universal constant from \cref{thm:effective-minor}.
        \begin{enumerate}
            \item Suppose that $\tw(F_n) \geq n^\epsilon$ for all $ n \geq 2$.
                  Let $p$ be the polynomial $n^{\max\{\lceil \tau/\epsilon \rceil, 3 \}}$.
                  Then $\tw(F_{p(n)}) \geq p(n)^\epsilon \geq n^\tau$
                  and $p(n) \geq 4n^2$.
                  By \cref{thm:effective-minor}, $F_{p(n)}$ contains the $n \times n$ grid as a minor.
                  Hence, $\colhom_{F_{p(n)}, p(n)}$ projects to $\colhom_{G_{n\times n}, 4n^2}$ via \cref{lem:minor-projection}.
                  The latter family is \VNP-complete by \cref{lem:colourful-grid-vnp-complete}.
            \item Suppose that $\tw(F_n) \in O(1)$ and $\pw(F_n) \geq \epsilon \log(n)$ for all $ n \geq 2$.
                  Write $t \coloneqq \max\{ \tw(F_n)\} + 1$.
                  Let $p$ be the polynomial $n^{\max\{\lceil \tau t/\epsilon \rceil, 6 \}}$.
                  Then $\pw(F_{p(n)}) \geq \epsilon \log(p(n)) \geq t \tau \log(n)$.
                  Hence, by \cref{thm:effective-minor}, $F_{p(n)}$ contains the complete binary tree $B_n$ as a minor.
                  Furthermore, $p(n) \geq n^6$.
                  Hence, $\colhom_{F_{p(n)}, p(n)}$ projects to $\colhom_{B_{n}, n^6}$ via \cref{lem:minor-projection}.
                  The latter family is \VP-complete by \cref{lem:colourful-binary-tree-vp-complete}.
            \item Suppose that $\pw(F_n) \in O(1)$ and $\td(F_n) \geq \epsilon \log(n)$ for all $ n \geq 2$.
                  Write $t \coloneqq \max\{ \pw(F_n)\} + 1$.
                  Let $p$ be the polynomial $n^{\max\{\lceil \tau t/\epsilon \rceil, 2 \}}$.
                  Then $\tw(F_{p(n)}) \geq \epsilon \log(p(n)) \geq t \tau \log(n)$.
                  Hence, by \cref{thm:effective-minor}, $F_{p(n)}$ contains the $n$-vertex path as a minor.
                  Furthermore, $p(n) \geq n^2$.
                  Hence, $\colhom_{F_{p(n)}, p(n)}$ projects to $\colhom_{P_{n}, n^2}$ via \cref{lem:minor-projection}.
                  The latter family is \VBP-complete by \cref{lem:colourful-path-vbp-complete}.\qedhere
        \end{enumerate}
    \end{proof}

    As a side note, we observe that the statement of \cref{thm:colourful-hom-vnp-hard} also holds for multigraphs.
    In this case, completeness is under constant-depth $c$\nobreakdash-reductions, even for \VNP.

    \begin{lemma}
        Let $(F_n)$ be a $p$-family of bipartite multigraphs
        and $(F'_n)$ the family of their underlying bipartite simple graphs.
        Then $(\colhom_{F'_n, n})$ constant-depth $c$\nobreakdash-reduces to $(\colhom_{F_n, n})$ over any field.
    \end{lemma}
    \begin{proof}
        Fix $F \coloneqq F_n$ and $F' \coloneqq F'_n$ for some $n \in \mathbb{N}$.
        Let $A \uplus B = V(F) = V(F')$ denote their bipartition.
        For an edge $uv \in E(F)$, write $\nu(uv) \in \mathbb{N}$ for its multiplicity in $F$. 
        Then
        \begin{align*}
            \colhom_{F_{n}, n}
            &= \sum_{h \colon V(F) \to [n]} \prod_{uv \in E(F)} x_{(u, h(u)), (v, h(v))} \\
            &= \sum_{h \colon V(F) \to [n]} \prod_{uv \in E(F')} \left(x_{(u, h(u)), (v, h(v))}\right)^{\nu(uv)}.
        \end{align*}
        The right hand-side is $\colhom_{F'_{n}, n}$ where $x_{(u,a), (v, b)}^{\nu(uv)}$ has been substituted for the variable $x_{(u,a), (v, b)}$  for $u,v \in V(F)$ and $a,b \in [n]$.
        Since $(F_n)$ is a $p$-family, all multiplicities $\nu(uv)$ are at most polynomial in $n$.
        Therefore, this substitution can be performed as a depth-$2$ $c$\nobreakdash-reduction.
    \end{proof}

    \subsection{Uncoloured homomorphism polynomials}
    \label{ssec:uncoloured}

    In this section, we prove \cref{thm:uncoloured-hom-complexity}.
    To that end, we develop an algebraic reduction from colourful to uncoloured homomorphism polynomials following ideas in \textcite{curticapean_count_2024}.
    The main reduction is provided by \cref{lem:single-hom-projection-minors}.
    As an intermediate step, we introduce coloured homomorphism polynomials.

    \begin{definition}
        Let $F$ be a bipartite multigraph with bipartition $V(F) = A \uplus B$ and $C$ a finite set.
        Let $c \colon V(F) \to V(C)$ be a map.
        For $n \in \mathbb{N}$,
        define the \emph{$n$-th coloured homomorphism polynomial}
        \[
        \colhom_{F, c, n} = \sum_{h \colon V(F) \to [n]} 
        \prod_{\substack{ab  \in E(F) \\ a \in A, b \in B}} x_{(c(a), h(a)), (c(b), h(b))}
        \]
        in variables indexed by $(C \times [n]) \times (C \times [n])$.
    \end{definition}

    Note that the colourful homomorphism polynomial $\colhom_{F, n}$ of $F$ as defined in \cref{def:colourful-hom-polynomial} coincides with $\colhom_{F, \id, n}$.
    The following \cref{lem:uncolour} shows how colourful homomorphism polynomials relate to uncoloured homomorphism polynomials.

    \begin{lemma}\label{lem:uncolour}
        Let $C$ denote a finite set of colours and $n \in \mathbb{N}$.
        Let $F$ be a bipartite multigraph.
        After identifying the variable indices $[|C| \cdot n]^2$ with $(C \times [n])^2$ on the left hand-side,
            it holds that
        \[ \hom_{F, |C| \cdot n} = \sum_{c \colon V(F) \to C} \colhom_{F, c, n}. \]
    \end{lemma}
    \begin{proof}
        The proof is by unpacking the summation.
        \begin{align*}
            \hom_{F, |C| \cdot n} 
            &= \sum_{h \colon V(F) \to C \times [n]} \prod_{ab \in E(F)} x_{h(a), h(b)}\\
            &= \sum_{c \colon V(F) \to C} 
              \sum_{h \colon V(F) \to [n]}
              \prod_{ab \in E(F)} x_{(c(a),h(a)), (c(b),h(b))} \\
            &= \sum_{c \colon V(F) \to C} 
              \colhom_{F, c, n}. \qedhere
        \end{align*}
    \end{proof}
    To make the lemma more explicit,
    for $G \in \mathbb{K}^{(C \times [n]) \times (C \times [n])}$,
            write $G^\circ$ for the corresponding matrix in $G \in \mathbb{K}^{(|C| \cdot n) \times (|C| \cdot n)}$.
    In the language of \cite{curticapean_count_2024}, 
    $G$ is a $C$-coloured graph on $n$~vertices per colour classes while $G^\circ$ is an uncoloured graph with $|C| \cdot n$ vertices.
    \Cref{lem:uncolour} asserts that $\hom_{F, |C| \cdot n}(G^\circ) = \sum_{c \colon V(F) \to C} \colhom_{F, c, n}(G)$.

    In order to facilitate interpolations, we prove the following \cref{lem:product-colourful}, which describes how coloured homomorphism polynomials evaluate on product graphs.
    It corresponds to a well-known identity on homomorphism counts into categorical products of graphs \cite[(5.30)]{lovasz_large_2012}.
    Let $C$ be a finite set of colours and $n,m \in \mathbb{N}$.
    For $C$-coloured edge-weighted bipartite graphs $G = (x_{(c,i),(c',i')})_{c,c' \in C, i,i' \in [n]}$
    and $H = (y_{(c,j),(c',j')})_{c,c' \in C, j,j' \in [m]}$,
    define $G \times H$ via $z_{(c, i,j), (c', i',j')} \coloneqq x_{(c,i),(c',i')} \cdot y_{(c,j),(c',j')}$
    for $c,c' \in C$, $i,i' \in [n]$, and $j,j' \in [m]$.
    This is a $C$-coloured edge-weighted bipartite graph with $n\cdot m$ vertices per colour class.

    \begin{lemma}\label{lem:product-colourful}
        Let $C$ be a finite set of colours.
        Let $F$ be a bipartite multigraph and $c \colon V(F) \to C$.
        For $n,m \in \mathbb{N}$ and $G$, $H$, and $G \times H$ as above, 
        \[
            \colhom_{F, c, nm}(G\times H) = \colhom_{F, c, n}(G)\colhom_{F, c, m}(H).
        \]
    \end{lemma}
    \begin{proof}
        The proof is by identifying $[nm]$ with $[n] \times [m]$ and rearranging the summation.
        \begin{align*}
            &\colhom_{F, c, nm}(G\times H)  \\
            &= \sum_{h \colon V(F) \to [n] \times [m]} \prod_{uv \in E(F)} z_{(c(u), h(u)), (c(v), h(v))} \\
            &= \sum_{h_1 \colon V(F) \to [n]} \sum_{h_2 \colon V(F) \to [m]} \prod_{uv \in E(F)} z_{(c(u), h_1(u), h_2(u)), (c(v), h_1(v), h_2(v))} \\
            &= \sum_{h_1 \colon V(F) \to [n]} \sum_{h_2 \colon V(F) \to [m]} \prod_{uv \in E(F)} x_{(c(u), h_1(u)), (c(v), h_1(v))} \cdot  y_{(c(u), h_2(u)), (c(v), h_2(v))} \\
            &= \colhom_{F, c, n}(G)\colhom_{F, c, m}(H).\qedhere
        \end{align*}
    \end{proof}

    The following \cref{lem:colhom-interpolation} 
    uses a standard interpolation argument to restrict a linear combination of homomorphism polynomials to patterns with a certain number of edges.
    The key insight is that $\colhom_{F, n}$ is a homogenous polynomial of degree $|E(F)|$.
    
    \begin{lemma}\label{lem:colhom-interpolation}
        Let $C$ be a finite set of colours and $N,r \in \mathbb{N}$.
        For $i\in [r]$, 
        let $F_i$ be bipartite multigraphs with $\lVert F_i \rVert \leq N$ and $c_i \colon V(F_i) \to C$.
        Let $\alpha_i$ be coefficients.
        For $n,k\in \mathbb{N}$, write $p \coloneqq \sum_{i \in [r]} \alpha_i \colhom_{F_i, c_i, n}$
        and $p_k \coloneqq \sum_{|E(F_i)| = k} \alpha_i \colhom_{F_i, c_i, n}$.
        Over any infinite field, 
        there exists a depth-$3$ circuit for $p_k$ of size polynomial in $N + n$ using oracle gates for $p$.
    \end{lemma}
    \begin{proof}
        Let $t$ be a fresh variable.
        For a $C$-coloured edge-weighted graph $G$, write $tG$ for the $C$-coloured edge-weighted graph obtained by multiplying each edge weight by $t$.
        For any $F$ and $c \colon V(F) \to C$,
        \begin{equation*}
            \colhom_{F, c, n}(tG) = \sum_{h \colon V(F) \to [n]} \prod_{uv \in E(F)} \left( t \cdot x_{(c(u), h(u)), (c(v),h(v))} \right)
            = t^{|E(F)|} \colhom_{F, c, n}(G).
        \end{equation*}
        Thus, by evaluating $p(tG) = \sum t^{|E(F_i)|} \alpha_i  \colhom_{F_i, c_i, n}(G) = \sum_k p_k(G) t^{|E(F_i)|}$
        on at most $N$ points, one may interpolate $p_k(G)$.
        In other words, there exist field elements $t_0, \dots, t_N$ and coefficients $\beta_0, \dots, \beta_N$
        such that $p_k(G) = \sum_{i=0}^N \beta_i p(t_i G)$.
        This circuit is of size polynomial in $N+n$ and of depth $3$.
    \end{proof}

    We now tie together the previous observations to prove the main reduction lemmas~\ref{lem:single-hom-projection} and~\ref{lem:single-hom-projection-minors}.
    The first lemma allows to extract the colourful homomorphism polynomial for a simple graph $S$ 
    from the homomorphism polynomial for a multigraph $F \supseteq S$ containing $S$ as a subgraph.
    The proof is somewhat similar to the proof of \cite[Theorem~9(b)]{curticapean_count_2023_arxiv}
    but uses arguments specific to polynomials and thereby yields a more efficient reduction.
    Concretely, a direct translation of \cite[Theorem~9(b)]{curticapean_count_2023_arxiv} to algebraic complexity yields a reduction from $\colhom_{S, n}$ to $\hom_{F, 2^{d-1} \cdot |V(F)| \cdot n}$ under the assumptions of \cref{lem:single-hom-projection}.
    The dependency on $|V(F)|$ in the latter polynomial is however too costly to prove any of the results from \cref{sec:main-contributions-classical-complexity}.

    \begin{lemma}\label{lem:single-hom-projection}
        Let $F$ be a bipartite multigraph
        and $S \subseteq F$ a connected simple subgraph of $F$ of maximum degree at most $d$.
        Let $n \in \mathbb{N}$.
        Over any field of characteristic zero,
        there exists a constant-depth circuit for $\colhom_{S, n}$ of size polynomial in $\lVert F \rVert + n$ 
        using oracle gates for $\hom_{F, 2^{d-1} \cdot |V(S)| \cdot n}$.
    \end{lemma}
    \begin{proof}
        We first turn $\hom_{F, n}$ into a linear combination of $S$-coloured homomorphism polynomials via \cref{lem:uncolour}.
        That is,
        \begin{equation}\label{eq:color-f}
            \hom_{F, |V(S)| \cdot n} = \sum_{c \colon V(F) \to V(S)} \colhom_{F, c, n}.
        \end{equation}
        Let $x_{(a,i),(b,j)}$ for $a,b \in V(S)$ and $i,j \in [n]$ denote the variables of this polynomial.
        Next we perform the substitution $x_{(a,i),(b,j)} \coloneqq 1 + y_{(a,i),(b,j)}$.
        Under this substitution,
        \begin{align}
            \colhom_{F, c, n}(\boldsymbol{1} + \boldsymbol{x}) &= \sum_{h \colon V(F) \to [n]} \prod_{uv \in E(F)} \left(1 + x_{(c(u),h(u)),(c(v),h(v))} \right) \notag \\
            &= \sum_{h \colon V(F) \to [n]} \sum_{E \subseteq E(F)} \prod_{uv \in E} x_{(c(u),h(u)),(c(v),h(v))} \notag \\
            &= \sum_{F' \subseteq F} \colhom_{F', c, n}(\boldsymbol{x}).\label{eq:subgraphs-inclusion-exclusion}
        \end{align}
        where the final sum ranges over all subgraphs $F'$ of $F$ with vertex set $V(F') = V(F)$.
        By applying \cref{lem:colhom-interpolation} to \cref{eq:color-f,eq:subgraphs-inclusion-exclusion},
        we obtain a circuit for
        \begin{equation}\label{eq:only-s-many-edges-survived}
            q_n \coloneqq \sum_{c \colon V(F) \to V(S)} \sum_{\substack{F' \subseteq F \\ \lvert E(F') \rvert = \lvert E(S) \rvert }}\colhom_{F', c, n}.
        \end{equation}

        Let $S_0$ and $S_1$ denote the even and odd \textsmaller{CFI} graphs of $S$, viewed as $S$-coloured graphs as considered in \cite{roberson_oddomorphisms_2022,curticapean_count_2024}.
        Note that $S_0$ and $S_1$ have at most $2^{d-1}$ vertices per colour class.
        \begin{claim}\label{claim:cfi-colourful-s}
            Let $S$ be a connected simple graph.
            Let $H$ be an $S$-coloured multigraph without isolated vertices on $|E(S)|$ edges.
            When counting colour-respecting homomorphisms, $\hom(H, S_0) \neq \hom(H, S_1)$
            if, and only if, $H \cong S$ as $S$-coloured graphs.
        \end{claim}
        \begin{claimproof}
            Write $h \colon H \to S$ for the $S$-colouring of $H$.
            By \cite[Theorem~3.13]{roberson_oddomorphisms_2022},
            it holds that $\hom(H, S_0) \neq \hom(H, S_1)$ if, and only if, $h$ is a weak oddomorphism.
            Here, when counting $\hom(H, S_0)$, edge multiplicities in $H$ are immaterial since $S_0$ and $S_1$ are simple.
            The backwards implication is immediate from \cite[Lemma~3.14]{roberson_oddomorphisms_2022}.
            For the forward direction, if $h$ is a weak oddomorphism, then each $h^{-1}(s)$ for $s \in V(S)$ contains a vertex $v$ such that $\deg_{H}(v) \geq \deg_S(s)$.
            More precisely, there is one neighbour in every $h^{-1}(t)$ for $t \in N_S(s)$.
            Since $H$ has the same number of edges as $S$, this means that $H \cong S$ as $S$-coloured graphs.
         \end{claimproof}

        In order to apply \cref{claim:cfi-colourful-s}, isolated vertices must be handled.
        This is straightforward since
        \begin{equation}\label{eq:isolated-vertices}
            \colhom_{F + K_1, c, n} = n \cdot \colhom_{F, c|_{V(F)}, n}.
        \end{equation}
         By \cref{claim:cfi-colourful-s,lem:product-colourful}, for $q_n$ as in \cref{eq:only-s-many-edges-survived},
         \begin{align}
             & q_{2^{d-1} \cdot n}(G \times S_0) - q_{2^{d-1} \cdot n}(G \times S_1) \notag \\
             &= \sum_{c \colon V(F) \to V(S)} \sum_{\substack{F' \subseteq F \\ \lvert E(F') \rvert = \lvert E(S) \rvert }} \colhom_{F', c, n}(G) \cdot \left(\hom(F', S_0) - \hom(F', S_1) \right) \notag \\
             &= \sum_{c \colon V(F) \to V(S)} \sum_{\substack{F' \subseteq F \\ \lvert E(F') \rvert = \lvert E(S) \rvert }} \alpha_{F'} \colhom_{F'', c|_{V(F'')}), n}(G) \cdot \left(\hom(F'', S_0) - \hom(F'', S_1) \right) \label{eq:isolated-vertices-intermediate}\\
             &= \sum_{\substack{c \colon V(F) \to V(S) \\ F' \subseteq F \\ (F'', c|_{V(F'')}) \cong S}}  \beta_{F'} \colhom_{S, \id, n}(G) \notag \\
             &= \colhom_{S, \id, n}(G) \sum_{\substack{c \colon V(F) \to V(S) \\ F' \subseteq F \\ (F'', c|_{V(F'')}) \cong S}}  \beta_{F'}.\label{eq:cfi-interpolation}
         \end{align}
         In  \cref{eq:isolated-vertices-intermediate},
         $F''$ denotes the graph obtained from $F'$ by deleting all isolated vertices.
         The positive coefficient $\alpha_{F'}$ depends on $F'$, $S$, and $S$ and can be computed via \cref{eq:isolated-vertices}.
         Analogously, $\beta_{F'} \coloneqq \alpha_{F'}\cdot \left(\hom(F'', S_0) - \hom(F'', S_1) \right) > 0$ for $(F'', c|_{V(F'')}) \cong (S, \id)$.
         Hence, \cref{eq:cfi-interpolation} yields a circuit for $\colhom_{S, \id, n}(G)$.
         To summarise,
         we obtain a circuit for $\colhom_{S, \id, n}$
         of size polynomial $\lVert F \rVert + n$ and constant-depth
         using oracle gates to $\hom_{F, 2^{d-1} \cdot |V(S)| \cdot n}$.
    \end{proof}

    Next we strengthen \cref{lem:single-hom-projection} by proving an analogue for minors rather than subgraphs.
    Analogous to \cref{eq:subgraphs-inclusion-exclusion},
    we construct a circuit for a linear combination of homomorphism polynomials of minors by a suitable substitution.

    Key is the following algebraic and bipartite analogue of \cite[Theorem~10]{seppelt_logical_2024}.
    For an $(n,n)$-vertex bipartite graph $G = (x_{i,j})_{i,j \in [n]}$,
    write $G^-$ for the $(2n,2n)$-vertex graph $G^- = (y_{(s,i),(t,j)})_{i,j \in [n], s,t \in \{A, B\}}$ given by
    \[
        y_{(s,i),(t,j)} \coloneqq \begin{cases}
            1, & \text{if } s = t \text{ and }  i = j,\\
            0, & \text{if } s = t \text{ and }  i \neq j,\\
            x_{i,j}, & \text{if } s = A \text{ and } t = B, \\
            x_{j,i}, & \text{if } s = B \text{ and } t = A. \\
        \end{cases}
    \]
    The purpose of this construction is to obtain the following expression which comprises homomorphisms counts from minors of $F$.
    
    \begin{lemma}\label{lem:quotient}
        Let $F$ be a bipartite multigraph.
        For $n \in \mathbb{N}$ and $G$, $G^-$ as above,
        \[
            \hom_{F, 2n}(G^-) = \sum_{s \colon V(F) \to \{A, B\}} \hom_{F \oslash s, n}(G).
        \]
    \end{lemma}
    Here, $F \oslash s$ is defined as follows.
        For $s \in V(F) \to \{A, B\}$,
        define $L \coloneqq \{uv \in E(F) \mid s(u) = s(v) \}$.
        Define an equivalence relation $\sim_L$ on $V(F)$ by declaring $u \sim_L v$ if $u$ and $v$ are in the same connected component of the subgraph of $F$ with vertex set $V(F)$ and edge set $L$.
        The edges of the multigraph $F \oslash s$ are the edges induced by $E(F) \setminus L$.
        The bipartition is given by~$s$.
        Indeed, if $u \sim_L v$, then $s(u) = s(v)$.
        Furthermore, if $u \not\sim_L v$ are adjacent, then there exists $uv \in E(F) \setminus L$ and hence $s(u) \neq s(v)$.
    \begin{proof}[Proof of \cref{lem:quotient}]
        It holds that
        \begin{align*}
            \hom_{F, 2n}(G^-)
            &= \sum_{h \colon V(F) \to [n]} \sum_{s \colon V(F) \to \{A, B\}} \prod_{uv \in E(F)} y_{(s(u),h(u)), (s(v),h(v))} \\
            &= \sum_{h \colon V(F) \to [n]} \sum_{s \colon V(F) \to \{A, B\}} \prod_{\substack{uv \in E(F) \\ s(u) = s(v)}} \delta_{h(u), h(v)} \prod_{\substack{uv \in E(F) \\ s(u) = A \\ s(v) = B}} x_{h(u), h(v)}  \prod_{\substack{uv \in E(F) \\ s(u) = B \\ s(v) = A}} x_{h(v), h(u)} \\
            &= \sum_{s \colon V(F) \to \{A, B\}} \sum_{h \colon V(F \oslash s) \to [n]} \prod_{uv \in E(F \oslash s)}  x_{h(u),h(v)} \\
            &= \sum_{s \colon V(F) \to \{A, B\}} \hom_{F \oslash s, n}. \qedhere
        \end{align*}
    \end{proof}

    \Cref{lem:quotient} allows to prove the following refinement of \cref{lem:single-hom-projection}.

    \begin{lemma}\label{lem:single-hom-projection-minors}
        Let $F$ be a bipartite multigraph
        and $S$ be a connected simple bipartite minor of $F$ of maximum degree at most $d$.
        Let $n \in \mathbb{N}$.
        Over any field of characteristic zero,
        there exists a constant-depth circuit for $\colhom_{S, n}$ of size polynomial in $\lVert F \rVert + n$ 
        using oracle gates for $\hom_{F, 2^{d} \cdot |V(S)| \cdot n}$.
    \end{lemma}

    \begin{proof}
        We first argue that every simple minor $S$ of $F$ is of the form $F' \oslash s$ for some subgraph $F' \subseteq F$ with $V(F') = V(F)$ and some map $s \colon V(F) \to \{A, B\}$.
        Let $B_0, B_1, \dots, B_k$ be branch sets for $S$ in $F$, i.e.\ $[k]$ is identified with $S$,
        $B_0 \uplus \dots \uplus B_k$ is a partition of $V(F)$,
        every $F[B_i]$ for $i \in [k]$ is connected,
        for $ij \in E(S)$, there exist an edge $uv \in E(F)$ with $u \in B_i$ and $v \in B_j$.
        
        Write $t \colon [k] \to \{A, B\}$ for the prescribed bipartition of $S$.
        Let $F'$ denote the subgraph of $F$ with $V(F') = V(F)$ containing all edges in between vertices in $B_i$ for $i \in [k]$,
        and a unique edge between $B_i$ and $B_j$ for $ij \in E(S)$, and no other edges.
        Let $s \colon V(F') \to \{A, B\}$ denote the map induced by $t$ taking arbitrary values on $B_0$.
        With this choice of $F'$ and $s$, the bipartite simple graph $F' \oslash s$ differs from $S$ only in potential isolated vertices.

        The remainder of the proof follows the structure of the proof of \cref{lem:single-hom-projection}.
        Analogous to \cref{eq:subgraphs-inclusion-exclusion},
        it holds that
        \begin{equation*}
            \hom_{F, n}(\boldsymbol{1} + \boldsymbol{x}) = \sum_{F' \subseteq F} \hom_{F', n}(\boldsymbol{x}) 
        \end{equation*}
        where the sum ranges over all subgraphs $F'$ of $F$ with vertex set $V(F') = V(F)$.
        By \cref{lem:quotient}, this allows to construct a polynomial-size constant-depth circuit for
        \begin{equation}\label{eq:subgraphs-and-quotients}
            q_n \coloneqq \sum_{F' \subseteq F} \sum_{s \colon V(F) \to \{A, B\}} \hom_{F' \oslash s, n}
         \end{equation}
         using oracle gates for $\hom_{F, 2n}$.
         Next, we apply \cref{lem:uncolour} to \cref{eq:subgraphs-and-quotients} obtaining a circuit 
         for
         \begin{align*}
             q_{|V(S)| \cdot n} 
             &= \sum_{F' \subseteq F} \sum_{s \colon V(F) \to \{A, B\}} \hom_{F' \oslash s, |V(S)| \cdot n} \\
             &= \sum_{F' \subseteq F} \sum_{s \colon V(F) \to \{A, B\}} \sum_{c \colon V(F' \oslash s) \to V(S)} \colhom_{F' \oslash s, c, n}.
         \end{align*}
         Finally, by \cref{lem:colhom-interpolation,claim:cfi-colourful-s}, as in the proof of \cref{lem:single-hom-projection},
         we obtain a circuit computing $\colhom_{S, n}$ of polynomial size and constant depth requiring oracle gates for $\hom_{F, 2^d \cdot |V(S)| \cdot n}$ where $d$ is the maximum degree of any vertex in $S$.
    \end{proof}

    Before proving \cref{thm:uncoloured-hom-complexity}, we note that the analogue of \cref{lem:colourful-containment} for uncoloured hom-polynomials is also true.

    \begin{lemma}\label{lem:uncoloured-containment}
        For every $p$-family $(F_n)_{n \in \mathbb{N}}$ of bipartite multigraphs,
        \begin{enumerate}
            \item $(\hom_{F_n, n})_{n \in \mathbb{N}}$ is in \VNP.
            \item if $\tw(F_n) \in O(1)$, then $(\hom_{F_n, n})_{n \in \mathbb{N}}$ is in \VP,
            \item if $\pw(F_n) \in O(1)$, then $(\hom_{F_n, n})_{n \in \mathbb{N}}$ is in \VBP,
            \item if $\td(F_n) \in O(1)$, then $(\hom_{F_n, n})_{n \in \mathbb{N}}$ is in \VF.
        \end{enumerate}
    \end{lemma}
    \begin{proof}
        The first assertion follows from
        \cite[Proposition~3.3.1]{saurabh_analysis_2016} by Valiant's criterion \cite{valiant_completeness_1979},
        the second from \cite[Theorem~5.3]{DPS2025},
        the third from \cref{lem:pathwidth-small-circuit},
        and the fourth  from \cref{lem:singleTreedepthHomCount}.
    \end{proof}

    This concludes the preparations for the proof of \cref{thm:uncoloured-hom-complexity}.

    \thmUncolouredHom*
    
    \begin{proof}
        Containment follows from \cref{lem:uncoloured-containment}.
        As in the proof of \cref{thm:colourful-hom-vnp-hard},
        there exists a polynomial $p$ such that, depending on the statement to be proven,
        for every sufficiently large $n \in \mathbb{N}$,
        \begin{enumerate}
            \item $F_{p(n)}$ contains the grid $G_{n \times n}$ as a minor,
            \item $F_{p(n)}$ contains the complete binary tree $B_{n}$ as a minor,
            \item $F_{p(n)}$ contains the path $P_n$ as a minor.
        \end{enumerate}
        Write $S_n \in \{ G_{n \times n}, B_n, P_n\}$ for the respective minor.
        Note that in each case the maximum degree of any vertex in $S_n$ is at most $4$ and the number of vertices in $S_n$ is at most $n^2$.
        Furthermore, $p$ may be chosen such that
        \[
            2^{4} \cdot  n \cdot |V(S_{n})| \leq 16 n^3 \leq p(n)
        \]
        for all sufficiently large $n \in \mathbb{N}$.
        By \cref{lem:single-hom-projection-minors},
        $(\colhom_{S_n, n})_{n \in \mathbb{N}}$ constant-depth $c$\nobreakdash-reduces to $(\hom_{F_n, n})_{n \in \mathbb{N}}$.
        By \cref{thm:colourful-hom-vnp-hard}, $(\colhom_{S_n, n})_{n \in \mathbb{N}}$
        is \VNP-, \VP, or \VBP-hard, respectively, depending on the initial assumption.
    \end{proof}

    \subsection{Linear combinations of homomorphism polynomials}
    \label{ssec:linearcombinations}
    
    We finally turn to linear combinations of (uncoloured) homomorphism polynomials. 
    Here, we employ an interpolation argument based on the following \cref{lem:bipartite-multigraph-linear-independent}.

    In \cref{ex:not-linearly-independent}, we observed that homomorphism polynomials of non-isomorphic patterns are not necessarily linearly independent.
    Nonetheless, linear independence holds if the patterns are not larger than the host size and the patterns do not have isolated vertices.
    The following lemma follows from \cite[Proposition~5.44 and Exercise~5.47]{lovasz_large_2012}. For completeness, we include a proof.
    
    The assumption that the patterns do not have isolated vertices does not constitute a loss of generality since $\hom_{F + K_1, n, m} = n \cdot \hom_{F, n,m }$ if the isolated vertex $K_1$ is added to the left part of the bipartition.
    In other words, isolated vertices can be handled in our non-uniform setting by rescaling the coefficients.

    \begin{lemma}\label{lem:bipartite-multigraph-linear-independent}
        Let $\mathbb{K}$ denote a field of characteristic zero.
        Let $n,m \in \mathbb{N}$.
        Let $F_1, \dots, F_r$ be bipartite multigraphs on at most $(n,m)$ without isolated vertices.
        If the  $F_1, \dots, F_r$ are pairwise non-isomorphic,\footnote{An isomorphism of bipartite multigraphs preserves the fixed bipartitions and the edge multiplicities.}
        then the polynomials $\hom_{F_1, n,m}, \dots, \hom_{F_r, n,m} \in \bbQ[\mathcal{X}_{n,m}]$
        are linearly independent.
    \end{lemma}
    \begin{proof}
Consider the embedding polynomials from \cite[Appendix~A]{DPS2025}.
        For a bipartite multigraph $F$ with bipartition $A \uplus B$, let
        \[
            \emb_{F, n, m} \coloneqq  \sum_{h \colon A \uplus B \hookrightarrow [n] \uplus [m]} \prod_{ab \in E(F)} x_{h(a), h(b)}.
        \]
        Since the isomorphism type of $F$ can be read off any of the monomials of $\emb_{F, n, m}$,
        it holds that $\emb_{F, n, m} = \emb_{F', n, m} \iff F \cong F'$ for all  bipartite multigraphs $F$ and $F'$ on at most $(n,m)$ vertices without isolated vertices.
        Hence, the polynomials $\emb_{F_1, n,m}, \dots, \emb_{F_r, n,m} \in \bbQ[\mathcal{X}_{n,m}]$
        are linearly independent.
        By \cite[Appendix~A]{DPS2025},
        \begin{equation}\label{eq:hom-to-emb}
            \hom_{F, n, m} = \sum_{\pi \in \Pi(A)} \sum_{\sigma \in \Pi(B)} \emb_{F/(\pi, \sigma), n, m} 
        \end{equation}
        where $\pi$ and $\sigma$ range over partitions of $A$ and $B$, respectively,
        and $F/(\pi, \sigma)$ is the bipartite quotient of $F$ according to $(\pi, \sigma)$.
        Suppose that the graph $F_i$ for $i \in [r]$ has vertex set $A_i \uplus B_i$.
        Let $\alpha_1, \dots, \alpha_r \in \bbQ$ be coefficients such that 
        \begin{align*}
            0 & = \sum_{i=1}^r \alpha_i \hom_{F_i, n,m} = \sum_{i=1}^r \sum_{\pi \in \Pi(A_i)} \sum_{\sigma \in \Pi(B_i)} \alpha_i \emb_{F_i/(\pi, \sigma), n, m}.
        \end{align*}
        By the previous observation,
        the $\emb_{F_i / (\pi, \sigma), n, m}$ are linearly independent.
        Thus, $\alpha_1 = \dots = \alpha_r = 0$, as desired.
    \end{proof}

    In order to turn \cref{lem:bipartite-multigraph-linear-independent} into an interpolation argument,
    we require the following standard fact.

    \begin{fact}\label{fact:interpolation}
        Let $\mathbb{K}$ be a field of characteristic zero.
        For $n,d \in \mathbb{N}$,
        let $f_1, \dots, f_n \colon \mathbb{K}^d \to \mathbb{K}$ be  functions.
        If the $f_1, \dots, f_n$ are linearly independent, then there exist points $\boldsymbol{x}_1, \dots, \boldsymbol{x}_n \in \mathbb{K}^d$ such that the matrix $(f_i(\boldsymbol{x}_j))_{i,j \in [n]}$ is invertible. 
    \end{fact}
    \begin{proof}
        We give a self-contained proof inspired by \cite{grossmann_existence_2020}.
        Consider the subspace
        \[
            U \coloneqq \operatorname{span}\{ (f_1(\boldsymbol{x}), \dots, f_n(\boldsymbol{x}) ) \mid \boldsymbol{x} \in \mathbb{K}^d \} \subseteq \mathbb{K}^n.
        \]
        If $U = \mathbb{K}^n$, then there exist evaluation points $\boldsymbol{x}_1, \dots, \boldsymbol{x}_n$ such that their corresponding vectors in $U$ span $\mathbb{K}^n$. For these vectors, the matrix $(f_i(\boldsymbol{x}_j))_{i,j \in [n]}$ is invertible.
        Otherwise, there exists a non-zero vector $\boldsymbol{c} \in U^\bot$ in the orthogonal complement of $U$.
        For every $\boldsymbol{x} \in \mathbb{K}^d$,
        \[
            c_1 f_1(\boldsymbol{x}) + \dots + c_n f_n(\boldsymbol{x}) = 0
        \]
        contradicting that the functions $f_1, \dots, f_n$ are linearly independent.
    \end{proof}

   This concludes the preparations for the proof of \cref{lem:uncoloured-simple-lovasz-interpolation},
   which allows to reduce a homomorphism polynomial for a single pattern to a linear combination of homomorphism polynomials of sublinear volume containing it.
   Here, we make use of the non-uniformity of algebraic computation to obtain interpolation coefficients. 
   In the setting of \cite{curticapean_count_2024},
   this could not be done within the given computational boundaries.
    
    \begin{lemma}\label{lem:uncoloured-simple-lovasz-interpolation}
        Let $n,r,N \in \mathbb{N}$.
        Let $F_i$ for $i \in [r]$ be a collection of pairwise non-isomorphic bipartite multigraphs without isolated vertices on at most $N$ vertices.
        Let $\alpha_i$ for $i \in [r]$ be non-zero coefficients.
        Over every field of characteristic zero,
        for every $\ell \in [r]$,
        there exists a constant-depth circuit of size polynomial in $n + N + r$ 
        with oracle gates for $\sum_{i \in [r]} \alpha_i \hom_{F_i, n \cdot N}$ 
        computing the polynomial $\hom_{F_\ell, n}$.
    \end{lemma}
    \begin{proof}
        By \cref{lem:bipartite-multigraph-linear-independent},
        the polynomials $\hom_{F_1, N}, \dots, \hom_{F_r, N}$ are linearly independent.
        In particular, they are linear independent as functions $\mathbb{K}^{N \times N} \to \mathbb{K}$.
        By \cref{fact:interpolation},
        there exist points $\boldsymbol{x}_1, \dots, \boldsymbol{x}_r \in \mathbb{K}^{N \times N}$
        such that the matrix $M \coloneqq (\hom_{F_i, N}(\boldsymbol{x}_j))_{i,j} \in \mathbb{K}^{r \times r}$ is invertible.
        Hence, there exist coefficient vector $\boldsymbol{\beta} \in \mathbb{K}^r$ such that
        $M \boldsymbol{\beta} = e_\ell$ where $e_\ell \in \mathbb{K}^r$ denotes the $\ell$-th standard basis vector.
        In other words,
        \begin{equation}\label{eq:delta}
            \sum_{j \in [r]} \beta_j \hom_{F_i, N}(\boldsymbol{x}_j) = \delta_{i ,\ell}
        \end{equation}
        for all $i \in [r]$.
        Let $G$ be an arbitrary edge-weighted $(n,n)$-vertex bipartite graph.
        Then, by \cref{lem:product-colourful},
        \begin{align*}
            \sum_{j \in [r]} \beta_j \sum_{i \in [r]} \alpha_i \hom_{F_i, n \cdot N}(G \times \boldsymbol{x}_j)
            &= \sum_{j \in [r]} \beta_j \sum_{i \in [r]} \alpha_i \hom_{F_i, n}(G) \hom_{F_i, N}(\boldsymbol{x}_j) \\
            &= \sum_{i \in [r]} \alpha_i \hom_{F_i, n}(G) \sum_{j \in [r]} \beta_j  \hom_{F_i, N}(\boldsymbol{x}_j) \\
            &\overset{\mathclap{\eqref{eq:delta}}}{=} \sum_{i \in [r]} \alpha_i \hom_{F_i, n}(G) \cdot \delta_{i ,\ell}\\
            &= \alpha_\ell \hom_{F_\ell, n}(G).
        \end{align*}
        Since $\alpha_\ell \neq 0$, we may divide by $\alpha_\ell$ and obtain a circuit for $\hom_{F_\ell, n}(G)$.
    \end{proof}

    We prove \cref{thm:uncoloured-lincomb-hom-complexity} on the complexity of linear combinations of homomorphism polynomials of bipartite multigraphs.
    The proof employs \cref{lem:uncoloured-simple-lovasz-interpolation} to reduce to \cref{thm:uncoloured-hom-complexity}.

    \thmLincombHomComplexity*
    \begin{proof}
        Containment follows from \cref{lem:uncoloured-containment} and observing that \VNP, \VP, and \VBP are closed under taking linear combinations of polynomially many terms.
        Intuitively, the polynomials $(\hom_{F_n, n})_{n \in \mathbb{N}}$ constant-depth $c$\nobreakdash-reduce to $(\sum \alpha_{F, n}\hom_{F, n})_{n \in \mathbb{N}}$ via \cref{lem:uncoloured-simple-lovasz-interpolation}.
        Hardness then follows from \cref{thm:uncoloured-hom-complexity}.

        For $n \in \mathbb{N}$, write $F_n$ for the large-width graph appearing with non-zero coefficient $\alpha_{F_n, n}$ in the linear combination.
        Since $|V(F_n)| \cdot n^{\epsilon} \leq n^{1 - \epsilon} \cdot n^{\epsilon}  = n$,
        $(\hom_{F_n, n^{\epsilon}})_{n \in \mathbb{N}}$
        constant-depth $c$\nobreakdash-reduces to
        $(p_n)$ via \cref{lem:uncoloured-simple-lovasz-interpolation}.
        Here, the size of the oracle circuit is polynomial in $\vol(n) + \dim(n) + n$, which is assumed to be polynomial in~$n$.
        Finally, observe that
        \begin{enumerate}
            \item $\tw(F_{n^{1/\epsilon}}) \geq (n^{1/\epsilon})^\epsilon = n$,
            \item $\pw(F_{n^{1/\epsilon}}) \geq \frac{\epsilon}{\epsilon} \log(n) = \log(n)$,
            \item $\td(F_{n^{1/\epsilon}}) \geq \log(n)$.
        \end{enumerate}
        Hence, by \cref{thm:uncoloured-hom-complexity}, $(\hom_{F_{n^{1/\epsilon}}, n})_{n \in \mathbb{N}}$
        is 
        \VNP-, \VP-, or \VBP-complete, respectively.
    \end{proof}

    \subsection{Hardness from parametrised Valiant's classes}
    \label{ssec:parametrised}

    We finally prove \cref{thm:hom-parametrised-simplified,thm:uncoloured-lincomb-hom-complexity-parametrised} and \cref{cor:symmetric-vs-non-symmetric-computation}.
    As in \cref{lem:colourful-grid-vnp-complete}, the starting point is a reduction from clique polynomials, which were shown to be \VW-complete by \textcite{blaser_parameterized_2019}, to colourful homomorphism polynomials for grids.
    Note that $(\colhom_{G_{k \times k}, n^2})$ is not a parametrised $p$-family as defined in \cref{def:parametrised-p-family} since the number of variables is $k^2 n^2$, which is not $p$-bounded in~$n$.
    Hence, $(\colhom_{G_{k \times k}, n^2})$ does not lie in \VW.

    \begin{lemma}\label{lem:grid-parametrised}
        The family $(\colhom_{G_{k \times k}, n})$
        is \VW-hard under fpt-$c$\nobreakdash-reductions over any field.
    \end{lemma}
    \begin{proof}
        By the main result of \cite{blaser_parameterized_2019}, 
        the family $(\clique_{n,k})$ is \VW-complete under fpt-$c$\nobreakdash-reductions over any field where
        \[
            \clique_{n,k} \coloneqq \sum_{\substack{C \subseteq [n] \\ \lvert C \rvert = k}} \prod_{\substack{i,j \in C \\ i < j}} y_{i,j} \prod_{i \in C} x_i.
        \]
        Note that this polynomial differs from the \VNP-complete polynomial $(\clique_n)$ from \cref{eq:hrubes} by the parameter $k$ and the product $\prod_{i \in C} x_i$.
        The proof of \cref{lem:colourful-grid-vnp-complete} can be adjusted accordingly as follows.
        Let $G_{k \times (k+1)}$ denote the grid with vertex set $[k] \times \allowbreak \{0,\dots, k\}$.
        We define a $G$-coloured edge-weighted graph such that $\colhom_{G_{k \times (k+1)}, n^2}$ evaluated at $G$ equals $\clique_{n, k}$.
        The vertices of $G$ are given as follows:
        \begin{itemize}
            \item For each $i \in [k]$ and $v \in [n]$, the colour class of $(i,i)$ contains the vertex $(i,i, v,v)$.
            \item For each $i, j \in  \{0,\dots, k\}$, 
            $i \neq j$, and every $u, v \in [n]$ such that $u \neq v$, the colour class $(i,j)$ contains the vertex $(i, j, u, v)$.
        \end{itemize}
        The edge weights are given as follows:
        \begin{itemize}
            \item For every $i \in \{0,\dots, k\}$, $u, v,v' \in [n]$, the edge $(i,0,u,v)(i,1,u, v')$ as weight $x_u$.
            \item For every $i \in [k]$, $i \leq j \leq n-1$, $u \in [n]$, and $1 \leq v < v' \leq n$,
            the edge $(i,j,u,v)(i,j+1,u,v')$ has weight $y_{v,v'}$.
            \item For every $j \in [k]$, $1 \leq i \leq j-1$, $v \in [n]$, and $1 \leq u < u' \leq n$, the edge $(i,j,u,v)(i+1,j,u',v)$ has weight $1$.\item All other edges have weight zero.
        \end{itemize}
        By the analysis in \cref{lem:colourful-grid-vnp-complete},
        this shows that $\clique_{n, k}$ is a projection of $\colhom_{G_{k \times (k+1)}, n^2}$.
        By \cref{lem:minor-projection},
        $\colhom_{G_{k \times (k+1)}, n^2}$ is a projection of $\colhom_{G_{(k+1) \times (k+1)}, n^2}$.
        Finally, $(\colhom_{G_{(k+1) \times (k+1)}, n^2})_{n,k}$
        fpt-projects to $(\colhom_{G_{(k) \times (k)}, n})_{n,k}$ by rescaling the parameters, see \cref{def:fpt-reductions}.
    \end{proof}

    We do not show \VW-hardness for any homomorphism polynomials.
    As in \cite{curticapean_full_2021},
    the following weaker statement is sufficient to rule out polynomial-size algebraic circuits for homomorphism polynomials for patterns of unbounded treewidth.

    \begin{lemma}\label{lem:grid-not-in-vp}
        Let $f \colon \mathbb{N} \to \mathbb{N}$ be a function such that $f^{-1}(\{ n \})$ is finite for every $n \in \mathbb{N}$.
        If $(\colhom_{G_{f(n) \times f(n)}, n}) \in \VP$,
        then $\VFPT = \VW$.
    \end{lemma}
    \begin{proof}
        We show that the assumption implies that $(\colhom_{G_{k \times k}, n}) \in \VFPT$,
        thus rendering \VFPT and \VW equal since \VFPT is closed under fpt-$c$-reductions by \cite[Lemma~4.17]{blaser_parameterized_2019}.
        
        Let $F(k) \coloneqq \max\{n \in \mathbb{N} \mid f(n) \leq k \} = \max(f^{-1}([k]))$, which is well-defined by assumption.
        Let $k,n \in \mathbb{N}$.
        Distinguish two cases:
        \begin{enumerate}
            \item If $f(n) \leq k$, then the definition of $\colhom_{G_{k \times k}, n}$ in \cref{def:colourful-hom-polynomial} yields a circuit of size polynomial in $n^{k^2} k^2 \leq F(k)^{k^2} k^2 $, 
            which is fpt-bounded in $k,n$.
            \item If $f(n) > k$, then $G_{k \times k}$ is a subgraph of $G_{f(n) \times f(n)}$.
            Hence, by \cref{lem:minor-projection},
            $\colhom_{G_{k \times k}, n}$ is a projection of $\colhom_{G_{f(n) \times f(n)}, n}$.
            In particular, $\colhom_{G_{k \times k}, n}$ admits a circuit of size polynomial in $n$.
        \end{enumerate}
        In both cases, $(\colhom_{G_{k \times k}, n}) \in \VFPT$.
        This implies by \cref{lem:grid-parametrised}  that $\VFPT = \VW$.
    \end{proof}

    \begin{observation}\label{obs:nondecreasing}
        Let $f \colon \mathbb{N} \to \mathbb{N}$ be unbounded and non-decreasing.
        Then $f^{-1}(\{ n \})$ is finite for every $n \in \mathbb{N}$.
    \end{observation}
    \begin{proof}
        For every $n \in \mathbb{N}$, there exists $\ell \in \mathbb{N}$ such that $f(\ell) > n$ since $f$ is unbounded.
        Since $f$ is non-decreasing, it holds that $f(\ell') > n$ for all $\ell' \geq \ell$.
        Hence, there are at most finitely many numbers whose image under $f$ is $\ell$.
    \end{proof}

    This concludes the preparations for the proof of \cref{thm:hom-parametrised-simplified}.

    \thmUncolouredhomParametrised*
    
    \begin{proof}[Proof of \cref{thm:hom-parametrised-simplified}]
        By \cref{lem:uncoloured-containment},
        it remains to argue that, if $\tw(F_n)$ is non-decreasing and in $\Omega(1)$,
        then $(\hom_{F_n, n}) \not\in \VP$,
        unless $\VFPT = \VW$.

        Let $\tau > 0$ denote the universal constant from \cref{thm:effective-minor}.\footnote{For the purpose of this proof, the original Grid Minor Theorem \cite{robertson_graph_1986} would be sufficient.}
        Let $g(n) \coloneqq (\tw(F_n))^{1/\tau}$ be a function such that $F_n$ contains a $g(n) \times g(n)$ grid as a minor.
        By assumption, the function $g$ is unbounded and non-decreasing, hence it meets the requirements of \cref{lem:grid-not-in-vp} by \cref{obs:nondecreasing}.
        Finally, let $h(n) \coloneqq \min\{g(n), n^{1/4} \}$.
        By \cref{lem:single-hom-projection-minors,lem:minor-projection},
        for all $n \in \mathbb{N}$,
        \begin{align*}
            \size(\colhom_{G_{h(n) \times h(n)}, \frac18 \sqrt{n}})
            & = \size(\colhom_{G_{h(n) \times h(n)}, \frac{n}{8 (n^{1/4})^2}}) \\
            & \leq \size(\colhom_{G_{h(n) \times h(n)}, \frac{n}{8 h(n)^2}}) \\
            & 
            \leq \size(\colhom_{F_n,  n})  + \operatorname{poly}(n).
        \end{align*}
        By letting $\ell \coloneqq \frac18 \sqrt{n}$,
        \[
           \size(\colhom_{G_{h(64\ell^2) \times h(64\ell^2)}, \ell})
           \leq \operatorname{poly}(\ell).
        \]
        The function $\ell \mapsto h(64\ell^2)$ satisfies the assumption of \cref{lem:grid-not-in-vp}.
        Hence, $\VFPT = \VW$, as desired.
    \end{proof}

    Towards \cref{cor:symmetric-vs-non-symmetric-computation}, we observe that uncoloured homomorphism polynomials are in \VW when viewed as parametrised families of polynomials, cf.\ \cref{def:parametrised-p-family}.
    
    \begin{lemma}\label{lem:containment-vw1}
        For every $p$-family $(F_k)$ of simple bipartite graphs,
        $(\hom_{F_k, n})_{n,k \in \mathbb{N}} \in \VW$.
    \end{lemma}
    \begin{proof}
        First observe that $(\hom_{F_k, n})_{n,k \in \mathbb{N}}$
        is a parametrised $p$-family, see \cref{def:parametrised-p-family}.
        The number of variables in $\hom_{F_k, n}$ is $n^2$.
        The degree of $\hom_{F_k, n}$ is equal to the number of edges in $F_k$ and thus polynomial in $k$, since $(F_k)$ is a $p$-family.
        Containment in \VW follows by generalising the argument from \cite{blaser_parameterized_2019} that $(\clique_{n,k}) \in \VW$.
        Recall \cref{def:vw-t}.
        We define the following universal polynomial
        $g_n$ in variables $x_{i,j}, y_{u,v}, z_{u, i}$ for $i,j,u,v \in [n]$.
        \begin{equation} \label{eq:gn}
            g_n \coloneqq \prod_{u,v\in [n]} \prod_{i, j \in [n]} (x_{i, j} y_{u,v}z_{u, i}z_{v, j} + 1 - z_{u, i}z_{v, j} y_{u,v}) \cdot 
            \alpha_n^{n}
            \cdot
            \prod_{u \in [n]} \prod_{1 \leq i < j \leq n} (z_{u, i} + z_{u, j} -2)
        \end{equation}
        for $\alpha_n \coloneqq (-2)^{- (n-1)(n-2)} (-1)^{n-1}$.
        The variables $x_{i,j}$ will become the variables of the homomorphism polynomials.
        The variables $y_{u,v}$ indicate edges in the pattern graph and will be assigned $0$ or $1$.
        Finally, the variables $z_{u,i}$ encode the homomorphism by indicating that the vertex $u$ is mapped to the vertex $i$.

        \begin{claim}
            $(\hom_{F_k, n})$
            fpt-$s$-reduces to
            $(\sum_{\boldsymbol{z} \in \bitvectors{n^2}{k}} g_n(\boldsymbol{x}, \boldsymbol{y}, \boldsymbol{z}))$.
        \end{claim}
        \begin{claimproof}
            Fix $n,k \in \mathbb{N}$.
            Write $\ell \coloneqq |V(F_k)|$
            and $N \coloneqq \max\{n,\ell\}$.
            We describe a substitution of $\sum_{\boldsymbol{z} \in \bitvectors{N^2}{\ell}} g_N(\boldsymbol{x}, \boldsymbol{y}, \boldsymbol{z})$ that evaluates to $\hom_{F_k, n}$. 

            To ease notation, suppose that $N = n = \ell$.
            The general case can be treated via a padding argument.
            The input values of $\hom_{F_k, n}$ are assigned to $\boldsymbol{x}$.
            It is $y_{u,v} = 1$ if, and only if, $uv \in E(F)$.

            Observe that $\boldsymbol{z} \in \{0,1\}^{V(F) \times n}$ ranges over  bit vectors with precisely $|V(F)|$ ones.
            The entry $z_{u, i}$ indicates that $u \in V(F)$ is mapped to $i \in [n]$.
            The two nested products on the right of \cref{eq:gn} enforce that, for every $u \in V(F)$, precisely one of $z_{u, i}$ for $i \in [n]$ is one.
            In other words, $\boldsymbol{z}$ encodes a map $V(F) \to [n]$.
        \end{claimproof}

        By renaming variables, all four products of unbounded fan-in in \cref{eq:gn} can be combined into one product.
        Hence, $g_n$ is computed by a constant-depth polynomial-size weft-$1$ circuit.
        It follows that $(\hom_{F_k, n})_{n,k \in \mathbb{N}} \in \VW$ 
        by \cref{def:vw-t}.
    \end{proof}
    
    \corSymNonSym*
    \begin{proof}
        Recall the following implications, which hold unconditionally.
        By \cref{thm:symmetric-classes-single-hom},
        for every family $(F_n)$ of bipartite multigraphs,
        \[
            \tw(F_n) \in O(1) \iff (\hom_{F_n, n}) \in \symVP.
        \]
        By \cref{lem:uncoloured-containment},
        for every $p$-family $(F_n)$ of bipartite multigraphs,
        \[
            (\hom_{F_n, n}) \in \VP \impliedby \tw(F_n) \in O(1).
        \]

        Suppose that $\VFPT \neq \VW$
        and let $(F_n)$ be a $p$-family of bipartite multigraphs of non-decreasing $\tw(F_n)$.
        By \cref{thm:hom-parametrised-simplified},
        \begin{equation}\label{eq:very-last-equation}
            (\hom_{F_n, n}) \in \VP \implies \tw(F_n) \in O(1)
        \end{equation}
        as desired.

        Conversely, suppose that $\VFPT = \VW$.
        Let $(G_{k \times k})$ denote the the $p$-family of $k \times k$ grids.
        Their treewidth is $\tw(G_{k \times k}) = k$.
        By \cref{lem:containment-vw1} and \cite[Lemma~4.17]{blaser_parameterized_2019}, 
        there exists a function $f \colon \mathbb{N} \to \mathbb{N}$
        and a $p$-bounded function $q$ such that 
        $\size(\hom_{G_{k \times k}, n}) \leq f(k) q(n)$ for all $n,k \in \mathbb{N}$.
        Without loss of generality, suppose that $k \leq f(k)$ for all $k \in \mathbb{N}$.
        Let $g(n) \coloneqq \max\{k \in \mathbb{N} \mid f(k) \leq q(n)\}$.
        Then $g(n) \leq f(g(n)) \leq q(n)$ and hence, $(G_{g(n) \times g(n)})$ is a $p$-family.
        By construction,
        $\size(\hom_{G_{g(n) \times g(n)}, n}) \leq f(g(n)) \cdot q(n) \leq (q(n))^2$.
        Hence, $(\hom_{G_{g(n) \times g(n)}, n}) \in \VP$.
        Finally, the treewidth of $(G_{g(n) \times g(n)})_{n \in \mathbf{N}}$ is unbounded and non-decreasing since $g$ is unbounded and non-decreasing.
        Hence, $(G_{g(n) \times g(n)})_{n \in \mathbf{N}}$ is a counter example to \cref{eq:very-last-equation}.
    \end{proof}

    The proof of \cref{thm:uncoloured-lincomb-hom-complexity-parametrised} is now conducted analogous to the proof of \cref{thm:uncoloured-lincomb-hom-complexity}.

    \thmLincombHomComplexityParametrised*
    \begin{proof}
        First suppose that $\max \tw(p_n) \in O(1)$.
        Since $(p_n)$ is a $p$-family and has sublinear volume,
        each constituent homomorphism polynomial counts homomorphism from a pattern $F$ with $\lVert F \rVert$ bounded polynomially in $n$, see \cref{lem:bipartite-multigraph-linear-independent}.
        Hence,
        by \cref{lem:uncoloured-containment}, 
        each constituent homomorphism polynomial is in \VP.
        Since $(p_n)$ has polynomial dimension, it is itself in \VP.

        Conversely, suppose that $\max \tw(p_n) \in \Omega(1)$ and $(p_n) \in \VP$.
        For every $n \in \mathbb{N}$,
        let $F_n$ denote a pattern with $\tw(F_n) = \max \tw(p_n)$ and non-zero coefficient $\alpha_{F_n, n} \neq 0$ in $p_n$.
        As in the proof of \cref{thm:uncoloured-lincomb-hom-complexity},
        by \cref{lem:uncoloured-simple-lovasz-interpolation},
        $\size(\hom_{F_n, n^\epsilon}) \leq \size(p_n) + \poly(n) \leq \poly(n)$.
        The treewidth of $(F_{n^{1/\epsilon}})_{n \in \mathbb{N}}$ is unbounded and non-decreasing.
        Hence, by \cref{thm:hom-parametrised-simplified},
        $(\hom_{F_n, n^\epsilon}) \in \VP$ implies that $\VFPT = \VW$, as desired.
    \end{proof}

    \section{Acknowledgements}

    We are grateful for fruitful discussions with and feedback from Radu Curticapean and Meena Mahajan.

    Prateek Dwivedi thanks the Independent Research Fund Denmark (grant agreement No. 10.46540/3103-00116B) and the support of Basic Algorithms Research Copenhagen (BARC), funded by VILLUM Foundation Grant 54451.
    
    Benedikt Pago received funding from UK Research and Innovation (UKRI) under the UK government’s Horizon Europe funding guarantee: grant number EP/X028259/1.
    
    Tim Seppelt is funded by the European Union (CountHom, 101077083). Views and opinions expressed are however those of the author(s) only and do not necessarily reflect those of the European Union or the European Research Council Executive Agency. Neither the European Union nor the granting
    authority can be held responsible for them.

    \newpage
	\printbibliography
\end{document}